\documentclass[11pt]{amsart}

\pdfoutput = 1

\usepackage{fouriernc}
\usepackage{macros,slashed,amsaddr}
\usepackage{setspace}
\usepackage{caption}
\usepackage{mathtools}
\usepackage{extarrows}

\usepackage{tikz-feynman}

\numberwithin{equation}{section}
\newcommand{\nocontentsline}[3]{}
\newcommand{\tocless}[2]{\bgroup\let\addcontentsline=\nocontentsline#1{#2}\egroup}
\newcommand{\changelocaltocdepth}[1]{%
  \addtocontents{toc}{\protect\setcounter{tocdepth}{#1}}%
  \setcounter{tocdepth}{#1}%
}
\newcommand{\sym}{\mathrm{Sym}}

\newcommand{\define}{\overset{\rm def}{=}}

\DeclareMathOperator{\coker}{coker}

\DeclareMathOperator{\Nlet}{Nilp}

\def\Nilpp(#1,#2){Y\qty({#1};{#2})}
\def\Nilp(#1,#2){C\Nlet\qty({#1};{#2})}

\begin{document}
\onehalfspacing

\title{Perspectives on the pure spinor superfield formalism}

\author{Richard Eager}
\address{Kishine Koen \\ 222-0034 Yokohama, Japan}
\email{eager@mathi.uni-heidelberg.de}

\author{Fabian Hahner}
\address{Mathematisches Institut der Universit\"at Heidelberg \\ Im Neuenheimer Feld 205 \\ 69120 Heidelberg, Deutschland}
\email{fhahner@mathi.uni-heidelberg.de}

\author{Ingmar Saberi}
\address{Ludwig-Maximilians-Universit\"at M\"unchen, Fakult\"at f\"ur Physik \\ Theresienstra\ss{}e 37 \\ 80333 M\"unchen, Deutschland}
\email{i.saberi@physik.uni-muenchen.de}

\author{Brian R. Williams}
\address{School of Mathematics\\ University of Edinburgh \\ Edinburgh, UK}
\email{brian.williams@ed.ac.uk}

\begin{abstract}
  In this note, we study, formalize, and generalize the pure spinor superfield formalism from a rather nontraditional perspective. To set the stage, we review the notion of a multiplet for a general super Lie algebra, working in the context of the BV and BRST formalisms. Building on this, we explain how the pure spinor superfield formalism can be viewed as constructing a supermultiplet out of the input datum of an equivariant graded module over the ring of functions on the nilpotence variety. We use the homotopy transfer theorem and other computational techniques from homological algebra to relate these multiplets to more standard component-field formulations. Physical properties of the resulting multiplets can then be understood in terms of algebrogeometric properties of the nilpotence variety. We illustrate our discussion with many examples in various dimensions.
\end{abstract}

\maketitle
\thispagestyle{empty}

\setcounter{tocdepth}{2}
\newpage
\tableofcontents

\section{Introduction}

Speaking broadly, a classical field theory concerns itself with the study of the sheaf of solutions to particular partial differential equations on the spacetime manifold, or more properly on the site of manifolds equipped with appropriate structure. Over an open set $U$, one considers solutions to the equations of motion of the theory on~$U$, considered up to gauge equivalence; since the equations of motion that are of physical interest tend to arise from variational principles, we will refer to it with the suggestive notation $\Crit(S)/\G$, where $\G$ refers to the group of local gauge transformations.

In general, this sheaf has several properties: First and foremost, its sections over~$U$ can be thought of as a covariant version of the phase space associated to~$\partial U$~\cite{CPS}, and thus have the structure of a symplectic space. (We are passing over numerous technical subtleties in silence; in particular, degeneracies of various kinds can and do occur, notably in the theory of constrained systems. Such examples arise naturally in our context~\cite{CederwallM5,twist20}, though we do not treat degeneracies in any detail here.) As  already indicated above, it may not consist just  of  the space of solutions to the equations of motion, but of its quotient by gauge equivalences. Lastly, since the degrees of freedom of many quantum field theories include fermions, it should most properly be understood as a (possibly singular,  stacky, or infinite-dimensional) supermanifold or graded space. 

In studying field theories, \emph{symmetries} play a crucial  role. Let $\lie{g}$ be a sheaf of Lie algebras. A classical theory has a {symmetry} by $\lie{g}$ when it is equipped with a map 
\deq{
	\rho: \lie{g} \to \Vect(\Crit(S)/\G)
}
of sheaves of Lie algebras. Usually, $\lie{g}$ is either a constant or a locally free sheaf (though other examples are possible, notably in holomorphic field theories). 
In the former case, one refers to a ``global'' symmetry, and in the latter to a ``local'' symmetry. By Noether's second theorem, local symmetries correspond to degeneracies in the variational problem of precisely the kind we ruled out above; as such, local symmetries are usually only relevant when gauged, and the terms ``local symmetry'' and ``gauge symmetry'' are often used interchangeably. 
Examples of symmetries abound; for example, any field theory on affine space should admit the Lie  algebra of infinitesimal affine transformations (the ``Poincar\'e  algebra'') as a symmetry, reflecting the coordinate  invariance (homogeneity and isotropy) of its dynamics. 

Since fermions are typically present in the theory, $\Vect(\Crit(S)/\G)$ is most naturally not a Lie algebra, but a \emph{graded} or \emph{super} Lie algebra. The most important examples of super Lie algebras extend the Poincar\'e symmetry by odd spacetime symmetries transforming in the spin representation of the Lorentz group; a field theory that admits an action of such an algebra is called \emph{supersymmetric}. 
The problem of constructing supersymmetric field theories has a long history in physics, dating back to the first explorations of the subject in the seventies~\cite{GolfandLikhtman,GervaisSakita,VolkovAkulov}. 

It is common wisdom in physics that representations of supersymmetry algebras in typical field theory models can be quite intricate. Often, the supersymmetry algebra closes only on-shell or up to gauge transformations. In other words, while a symmetry of the theory in the above sense can be defined, it does not arise in a straightforward manner from an action on the larger space of fields inside  of which  the equations of motion are solved. This leads, among other issues, to difficulties in quantizing the theory.  

In typical field theory models the structure of supersymmetry transformations roughly falls into four distinct cases:
\begin{itemize}
  \item There is a set of fields on which the supersymmetry algebra is represented on the nose. This is the case, for example, for the four-dimensional $\cN = 1$ chiral multiplet.
  \item The supersymmetry algebra is only represented after taking the quotient by the action of the gauge group. This happens, for example, for the four-dimensional vector multiplet.
	\item The supersymmetry algebra is represented only after imposing the equations of motion. Here, the six-dimensional hypermultiplet is an example.
        \item The supersymmetry algebra is represented only after taking the quotient by gauge transformation and imposing the equations of motion. This most general case appears in ten-dimensional super Yang--Mills theory, among other examples.
\end{itemize}
The first objective of this note is to formalize these considerations using the language of homotopical algebra; we work in the context of the BRST and BV formalisms, which seek to respectively replace the quotient by gauge symmetries and the imposition of equations of motion by appropriate derived analogues. In~\S\ref{sec: prelim}, we set up some necessary preliminaries for this context; in particular, we give a definition of a multiplet that is designed to capture all these different aspects of symmetry in our context.

Once this terminology is established, we turn our attention towards the construction of supermultiplets via the pure spinor superfield formalism; see~\cite{BerkovitsSuperparticle}, and especially the review~\cite{Cederwall} and references therein. 
Our perspective is somewhat nontraditional. In~\S\ref{sec: formalism} we set up the formalism in a generalized setting (without restricting to supersymmetry algebras of physical interest), clarify its relation to various standard constructions in homological algebra, and give an explicit account of calculational techniques from commutative algebra. 

In our interpretation, which builds on that in~\cite{NV}, the pure spinor superfield formalism constructs a supermultiplet out of the datum of an equivariant module over the ring of functions $\cO_Y$ on the nilpotence variety $Y$ of the relevant superalgebra. Speaking roughly, the output of the formalism is a rather large cochain complex that is automatically equipped with a strict action of the supersymmetry algebra---indeed, which is quasi-isomorphic to a standard component-field description of the multiplet in the BRST or BV formalism, but which is free over \emph{superspace} rather than just over the spacetime manifold. We can then recover the usual component-field description by moving from this large resolution to a smaller, quasi-isomorphic cochain complex of vector bundles over spacetime, which is in a certain sense ``minimal'' resolution of this kind. A particular filtration on the pure spinor cochain complex produces the component-field formulation in canonical fashion; the set of component fields is identified with the vector bundle associated to the representation of Lorentz and $R$-symmetry  on the Koszul homology of the input module.

One can then transfer the various structures present on the large complex to the component fields, using the homotopy transfer theorem. As we will see, this procedure links the component field description of the multiplet closely to the minimal free resolution of the equivariant module over the ambient polynomial ring. In particular, we find that the non-derivative part of the supersymmetry transformations can be read off directly from the resolution differential. This provides a proof for a conjecture made by Berkovits in~\cite{BerkovitsSupermembrane}.

Given our presentation of the pure spinor superfield formalism, it is natural to ask questions how algebraic properties of $\cO_Y$-modules are related to physical properties of the resulting multiplet. In~\S\ref{sec: data}, we point out that the Gorenstein property ensures the existence of a pairing on the multiplet; this pairing, however, can admit various different physical interpretations. We furthermore study dualizing modules and explain how the Cohen--Macaulay property is related to antifield multiplets.

Throughout the text we illustrate the procedure with examples in different dimensions and with various amounts of supersymmetry. In particular, we provide a detailed discussion of ten-dimensional super Yang--Mills theory showing how all the different structures present in the component field formulation arise via homotopy transfer.

\changelocaltocdepth{1}
\subsection*{Acknowledgements}
We would like to give special thanks to J.~Walcher for conversations and collaboration on related projects, and to M.~Cederwall for conversation and inspiration. We also gratefully acknowledge conversations with I.~Brunner, C.~Elliott, B.~Haake, and S.~Raghavendran. This work is funded by the Deutsche Forschungsgemeinschaft (DFG, German Research Foundation) under Germany's Excellence Strategy EXC 2181/1 --- 390900948 (the Heidelberg STRUCTURES Excellence Cluster). I.S. is supported by the Free State of Bavaria.  The work of B.W. is supported by the University of Edinburgh.

\changelocaltocdepth{2}

\section{Preliminaries}
\label{sec: prelim}
\subsection{Gradings and basic definitions}
Throughout this paper, we work with objects (be they vector spaces, vector bundles, associative algebras, or Lie algebras) that are graded by $\Z\times \Z/2\Z$. We will use the abbreviation ``dgs,'' for ``differential graded super,'' to refer to objects of this sort, at least for emphasis. 
\begin{dfn} 
  A \emph{dgs vector space} is a $\Z \times \Z/2\Z$-graded vector space $E^\bu$, equipped with a square-zero differential $d$ of bidegree $(1,+)$. Equivalently, $E^\bu$ is a cochain complex in the category of super vector spaces. 
  We can thus write
  \deq{
    E = \bigoplus_k \left( E^k_+ \oplus E^k_- \right).
  }
  The \emph{total parity} $|v| \in \Z/2\Z$ of a homogeneous element $v \in E^n$ is defined by 
  \deq{
    |v| = \begin{cases} n \bmod 2, & v \in E^n_+; \\
      n + 1 \bmod 2, & v \in E^n_-.
    \end{cases}
  }
\end{dfn}
We remark that each of these gradings has a clear physical meaning: the integer grading corresponds to the ghost number or homological degree, whereas the $\Z/2\Z$ grading corresponds to the intrinsic parity (fermion number modulo two). In some contexts, it will be useful for us to think about defining an integer-valued fermion number---in other words, lifting the $\Z/2\Z$ grading to a second integer grading. This motivates the following definition:
\begin{dfn} 
  A \emph{lift} of a dgs vector space $E^\bu$ is an integer grading 
  \deq{ E^k = \bigoplus_i E^k_i }
  on each graded component of~$E^\bu$ that lifts the intrinsic parity, such that the differential $d$ has bidegree $(1,0)$. In other words,
  \deq{
    E^k_+ = \bigoplus_{i \text{ even}} E^k_i, \qquad
    E^k_- = \bigoplus_{i\text{ odd}} E^k_i.
  }
  A lifted dgs vector space is thus a cochain complex in the category of graded vector spaces---in other words, a bigraded vector space with a differential of degree~$(1,0)$.
\end{dfn}
\begin{dfn}
  A \emph{commutative dgs algebra}, or cdgsa, is a dgs vector space $A^\bu$ equipped with a bilinear multiplication 
  \deq{
    m_2: A^\bu \otimes A^\bu \to A^\bu.
  }
  The multiplication is required to be a cochain map of bidegree $(0,+)$; furthermore, it should be commutative with respect to the Koszul sign rule determined by the total parity. That is, 
  \deq{
    ab = (-1)^{|a||b|} ba.
  }
We remark that a cdgsa is a commutative differential graded algebra in the category of super vector spaces. There is also an obvious notion of a lift of a cdgsa, such that a lifted cdgsa is a commutative differential graded algebra in the category of graded vector spaces.
Finally, we can extend our definitions to encompass super $A_\infty$ algebras: a (lifted) super $A_\infty$ algebra $A^\bu$ is an $A_\infty$ algebra in the category of super (or graded) vector spaces. That is, it is a collection 
\deq{
  A^\bu = \bigoplus_k A^k
}
of super (or graded) vector spaces, equipped with maps $m_n$ of arity $n$ and bidegree $(2-n,+)$ or $(2-n,0)$ that satisfy the usual $A_\infty$ relations.
\end{dfn}
\begin{eg}
  Let $V^\bu$ be a dgs vector space. The \emph{polynomial algebra} $\Sym(V^\bu)$ is the free dgs-commutative algebra generated by~$V^\bu$. Concretely, it is the quotient 
\deq{
  \Sym(V^\bu) = T(V^\bu)/\langle x y - (-1)^{|x||y|} y x \rangle
}
of the tensor algebra by the ideal generated by all (anti)commutators of homogeneous elements, where (anti)commutativity is determined by the Koszul sign rule for the total parity.
\end{eg}

Of course, all of the notions we have introduced for associative algebras have parallels for Lie algebras, which we now quickly introduce.
Let $x_1,\ldots, x_n$ be homogeneous elements of a dgs vector space $V^\bu$, and $\sigma \in S_n$ a permutation. Then the \emph{Koszul sign} 
$\epsilon(x_1,\ldots, x_n ; \sigma)$ of the permutation is defined by the relation 
\deq{
x_1 \cdots x_n = \epsilon(x_1,\ldots, x_n ; \sigma) x_{\sigma(1)} \cdots x_{\sigma(n)}  .
}
in the algebra $\Sym(V^\bu)$.
Furthermore define $\chi(\sigma) = (-1)^{{\rm sgn}(\sigma)} \epsilon(x_1,\ldots,x_n ; \sigma)$.

\begin{dfn}
  Let $\fg$ be a (lifted) dgs vector space.
  A \emph{(lifted) super $L_\infty$ algebra structure} on $\fg$ is a collection of multilinear maps 
\begin{equation}
	\mu_k \colon \fg^{\times k} \to \fg
\end{equation}
for $k \geq 1$, of bidegree $(2-k,+)$ (or $(2-k,0)$, respectively), such that the following two conditions hold:
\begin{itemize}
\item[(1)] Graded skew symmetry. 
For all $\sigma \in S_k$, $x_i \in \fg$ one has 
\begin{equation}
	\mu_k\left( x_{\sigma(1)},\ldots, x_{\sigma(k)} \right) = \chi(\sigma) \mu_k \left( x_1,\ldots, x_k \right) \: .
\end{equation}

\item[(2)] Higher Jacobi identities. 
For all $x_i \in \fg$ one has
\begin{equation}
	\sum_{i+j = k+1} \ \sum_{\sigma \in S(i;k)} (-1)^{i(j-1)} \chi(\sigma) \mu_j \left( \mu_i \left( x_{\sigma(1)}, \ldots, x_{\sigma(i)} \right) , x_{\sigma(i+1)}, \ldots, x_{\sigma(k)} \right) = 0 \: .
\end{equation}
\end{itemize}
Here $S(i;k) \subset S_k$ denotes all permutations such that $\sigma(1) \leq \dots \leq \sigma(i)$ and $\sigma(i+1) \leq \dots \leq \sigma(k)$.  We remark that a (lifted) super $L_\infty$ algebra is just an $L_\infty$ algebra in super (respectively, in graded) vector spaces. We further remark that the datum of a (lifted) super $L_\infty$ algebra structure is equivalent to a square-zero derivation of bidegree $(1,+)$ (or $(1,0)$ in the lifted case) on the free dgs commutative algebra $\Sym(\fg^\vee[-1])$. This derivation $\d_\fg$ defines the complex computing Lie algebra cohomology,
\deq{
\clie^\bu(\fg) := \left( \Sym(\fg^\vee[-1]) \: , \: \d_\fg \right).
}
The shift is with respect to the homological degree.
\end{dfn}

There are some special cases of this definition that we point out. When $\lie{g}$ is supported purely in even parity, we recover the 
ordinary notion of an $L_\infty$ algebra~\cite{HinichSchechtman, LadaMarkl}. 
On the other hand, when $\lie{g}$ is supported in degree zero, we recover the notion of a super Lie algebra (or, in the lifted case, a graded Lie algebra). When $\mu_k = 0$ for all $k>2$, we obtain the notion of a \emph{dg super Lie algebra}.

\begin{eg} \label{ex: End}
  Let $V^\bu$ be a (lifted) dgs vector space. Then $\End(V^\bu)$ is a dg super Lie algebra; the bracket $\mu_2$ is given by the commutator 
  \deq{
    \mu_2(x,y) = [x,y] = xy - (-1)^{|x||y|} yx
  }
  whereas the differential arises via
  \deq{
  d_{\mathrm{End}(V^\bu)} = [d,-] \: .
  }
  We remark that $\End(V^\bu)$ is in fact naturally a dgs associative algebra; the dgs Lie structure is obtained by applying the usual forgetful functor. 
\end{eg}

\begin{dfn} \label{dfn:Lmap}
A {\em $L_\infty$ map} between super $L_\infty$ algebras 
\[
\Phi \colon \fg \rightsquigarrow \fh
\]
is a map of graded super commutative algebras 
\deq{
  \Phi^*: \clie^\bu(\fh) \rightarrow \clie^\bu(\fg).
}
that preserves the augmentation map to constants in degree zero.
\end{dfn}

\begin{dfn} \label{dfn:Lmodule}
  Let $\fg$ be a super $L_\infty$ algebra. An \emph{$L_\infty$ dgs module} is a dgs vector space $V^\bu$, together with an $L_\infty$ map
  \deq{
    \lie{g} \rightsquigarrow \End(V^\bu).
  }
\end{dfn}

\subsection{Homotopy transfer}

We will repeatedly make use of the \emph{homotopy transfer theorem} in various contexts. We refrain from giving a general review of homotopy algebraic structures here; the reader is referred to~\cite{Vallette,LadaMarkl}. Nonetheless, we will quickly recall the general idea.

It is common knowledge that various mathematical objects---for example sheaves or modules---admit interesting ``higher structures.'' This might include higher sheaf cohomology groups, for example, or more generally other derived functors such as $\Ext$ and $\Tor$. These higher structures originate, in some sense, from the ``constraints" imposed on these objects: for example, the failure of a module to be free. 

To compute higher derived functors, one technique is to replace the object one wants to study by a ``resolution.'' This is a cochain complex of simpler objects (for example, free modules) that is quasi-isomorphic to the complicated object one wants to study. In derived geometry, one views this cochain complex as a replacement of the underlying object. 

Just as the equations defining a non-free module lead to higher structures and need to be resolved, many algebraic structures are defined by collections of structure morphisms that satisfy certain strict equations. (For example, one requires associativity in the form $((ab)c) = (a(bc))$, or the Jacobi identity for a Lie bracket.) When such equations are imposed in a cochain complex, they do not play well with homotopy-theoretic operations or notions of equivalence such as quasi-isomorphism. The remedy consists of ``resolving'' the equations that are imposed on the defining maps of the algebraic structure. In technical language, one resolves the operad defining the algebraic structure one is interested in by a free dg operad. (See~\cite{Markl} for discussion of this perspective.) 

There is then a collection of general results, which state that a homotopy algebraic structure may be transferred along \emph{homotopy data} between two quasi-isomorphic cochain complexes (for example, a deformation retract) by summing over all marked trees in a consistent fashion. Vertices are to  be labeled with operations of the structure to be transferred, and internal edges with the homotopy. See~\cite{Vallette} for more details.

The phenomenon of homotopy transfer is very broad, and encompasses many examples from throughout mathematics, both more and less familiar. We mention some examples:
\begin{itemize}
  \item A cochain complex is defined by a grading, together with a single endomorphism $D$ of degree $+1$, satisfying the equation $D^2 = 0$.
  A cochain complex in cochain complexes is a bicomplex: we give a second grading on $(C^\bu, d)$, together with a square-zero cochain map $D$. Resolving the equation $D^2 = 0$ gives rise to an operad known as the $D_\infty$ operad: it encodes a sequence of maps $D_i$ of bidegree $(1-i,i)$, which obey the relations 
    \deq{
      d D_n + (-1)^n D_n d = \sum_{i+j=n} (-1)^i D_i D_j.  
    }
    Homotopy transfer of~$D$ to $H^\bu(C,d)$ generates a $D_\infty$ module structure whose constituent maps encode the higher differentials of the spectral sequence of the bicomplex. This will play a role for us in describing the relation of pure spinor superfields to their component-field descriptions; see~\S\ref{sec: formalism}.
  \item The operad governing associative algebras is resolved by the $A_\infty$ operad, which has operations $\{m_n\}$ of arity $n$ and degree $2-n$ for all $n\geq 1$. Similarly, the operad governing Lie algebras is resolved by the $L_\infty$ operad, which has bracket operations $\mu_n$ of arity $n$ and degree $2-n$ for all $n\geq 1$ as we reviewed explicitly above. Transferring the associative algebra structure on de Rham forms to cohomology produces an $A_\infty$ structure with vanishing $m_1$ and $m_2$ the ordinary cup product. Higher $m_n$'s correspond to the classical Massey product operations.
  \item In the BV formalism, a perturbative classical field theory is described by a cyclic local $L_\infty$ algebra whose differential encodes the linearized equations of motion and gauge invariances of the free theory. Homotopy transfer to the cohomology of the differential is related to the interaction picture in quantum field theory; the diagrams that describe the transferred $L_\infty$ structure on on-shell states are precisely tree-level Feynman diagrams, where the homotopy is the Feynman propagator. The operations of the transferred $L_\infty$ structure correspond to tree-level amplitudes~\cite{Kajiura, MacrelliScattering}. Homotopy transfer of~$L_\infty$ structures will be relevant for us when discussing interactions for pure spinor superfields and their relation to the component-field formalism; see~\S\ref{sec: 10d} for an example.
  \item In the BV or BRST formalism, the symmetries of a field theory are encoded as $L_\infty$ module structures on the complex of fields. Moving to another quasi isomorphic complex of fields (e.g by integrating out an auxiliary field), one can obtain the new module structure via homotopy transfer. We will use this to derive the action of the supersymmetry algebra on the component fields in the pure spinor superfield formalism. An explicit account on the homotopy transfer for module structures is given in Appendix~\ref{ap:modules}.
    
\end{itemize}

\subsection{Maurer--Cartan elements and nilpotence varieties}
We recall that the Maurer--Cartan equation in an $L_\infty$ algebra $\lie{g}$ takes the form 
\begin{equation}\label{eq:MC}
  \sum_{k\geq 1} \frac{1}{k!} \mu_k \left( x,\ldots,x \right) = 0.
\end{equation}
  
Here $x \in \lie{g}$ is an element of degree one; each of the terms in the above equation thus carries degree two.
We can clearly generalize this definition to super $L_\infty$ algebras by asking for Maurer--Cartan elements $x$ of bidegree $(1,+)$. It is straightforward to see that Maurer--Cartan elements of this form correspond precisely to deformations of the super $L_\infty$ algebra structure. We will write $\MC(\lie{g})$ for the space of Maurer--Cartan elements; note that we do not quotient this space by the action of degree-zero elements, preferring to think of $\MC(\lie{g})$ as a space equipped with a $\lie{g}_0$-action by vector fields.

Now, given any super $L_\infty$ algebra, we can forget the $\ZZ \times \ZZ/2\ZZ$-grading down to a $\Z/2\Z$-grading by remembering only the total parity. This is enough information to define the appropriate Koszul signs, and $\mu_k$ is then simply a multilinear operation with appropriate symmetry properties and parity $(-1)^k$. We will call the resulting object a $\Z/2\Z$-graded $L_\infty$ algebra. We can then ask about the space $Y(\lie{g})$ of \emph{odd} elements satisfying the Maurer--Cartan equation~\eqref{eq:MC}. Elements of this space will correspond to deformations of~$\lie{g}$ as a $\Z/2\Z$-graded $L_\infty$ algebra and there will be an injective map $\MC(\lie{g}) \hookrightarrow Y(\lie{g})$. We call $Y(\lie{g})$ the \emph{nilpotence variety} of~$\lie{g}$; when $\lie{g}$ is a super Lie algebra, this agrees with the notion given in~\cite{NV}.

For the purposes of this paper, $Y$ is an affine scheme; we take 
\deq{
Y = \Spec \O_Y, \qquad
\O_Y = R/I,
}
where $R=\Sym(\Pi \fg_-^\vee)$ is a polynomial ring in commuting variables, and $I$ is the ideal generated by the Maurer--Cartan equations~\eqref{eq:MC}. In this work we will only be concerned with the case where $\fg$ is a super Lie algebra, thus $I$ will be generated by quadratic equations and $\cO_Y$ is a graded ring. Since we view $Y$ as an affine scheme, we will move back and forth freely between discussing the geometry of~$Y$ and the graded ring $\O_Y$; hopefully, no confusion should arise. (Sometimes we may also consider the geometry of the projective scheme $\Proj \O_Y$; in either case, the essential object is the graded ring $\O_Y$.)
The distinction between a variety and a scheme will, in fact, play an important role in applications; see~\S\ref{ssec: schemes}.

\subsection{Multiplets and local modules}

In this section, we move towards the setting of field theory by introducing the new ingredient of \emph{locality}.

\subsubsection{Local modules}

Let $X$ be a manifold thought of as spacetime. 
There is an obvious notion of a \emph{dgs vector bundle:} concretely, we mean a $\ZZ \times \ZZ/2$-graded vector bundle 
\begin{equation}
	E = \bigoplus_k E^k = \bigoplus_k \left(E_+^k \oplus E_-^k\right)
\end{equation}
equipped with a collection of differential operators $D \colon \cE_{\pm}^k \to \cE_{\pm}^{k+1}$ such that $D \circ D = 0$.
Here,~$\cE^k_\pm = \Gamma(X, E^k_\pm)$ denotes the $C^\infty$ sections of $E^k_\pm$.\footnote{This is close to, but differs from, the notion of a ``flat superconnection''. 
For one, our operator $D$ is required to be {\em even}.}

Suppose that $\fg$ is a super $L_\infty$ algebra. 
We will define a local $\fg$-module to be a dgs vector bundle on $X$ 
equipped with a sufficiently local homotopy action of~$\fg$. 

To give the precise definition we first need a small bit of background.
Consider the $\ZZ \times \ZZ/2$-graded vector space $\cE = \Gamma(X , E)$. As explained in Example~\ref{ex: End}, the endomorphisms $\mathrm{End}(\cE)$ naturally form a dg super Lie algebra: the structure maps consist of the commutator and the differential $[D,-]$.
Inside $(\mathrm{End(\cE)},[D,-])$ there is a sub dg super Lie algebra consisting of all endomorphisms which are differential operators. 
We will denote it by $(\cD(E) , [D,-])$.

\begin{dfn}
A {\em local (super $L_\infty$) $\fg$-module} is a dgs vector bundle $(E, D)$ equipped with a super $L_\infty$-map (see Definition \ref{dfn:Lmap}):
\begin{equation}
	\rho \colon \fg \rightsquigarrow \big(\cD(E)\; , \; [D,-] \big) \:.
\end{equation}

We will refer to the data of a local $\fg$-module by a triple $(E, D, \rho)$. 
\end{dfn}

The space of sections of any dgs vector bundle is a dgs vector space. 
The space of sections (over any open set) of a local $\fg$-module $(\cE, D)$ is a dgs $L_\infty$ module for the super Lie algebra $\fg$, see Definition~\ref{dfn:Lmodule}. 

Concretely, the data of $\rho$ consists is a collection maps
\begin{equation}
\rho^{(j)} : \fg^{\otimes j} \longrightarrow \cD(E)[1-j] ,\qquad j \geq 1
\end{equation}
satisfying some compatibility relations, the lowest of which reads
\begin{equation} \label{L-action}
[\rho^{(1)}(x),\rho^{(1)}(y)] - \rho^{(1)}([x,y]) = [D, \rho^{(2)}(x,y)] \: .
\end{equation}
Note that if the left hand side were zero, then we would have a strict Lie algebra action. Thus, $\rho^{(2)}$ provides a homotopy correcting the failure of $\rho^{(1)}$ to be strict.

One way to unravel this definition is in terms of the cochain complex computing the Lie algebra cohomology of $\fg$. 
The map $\rho$ is equivalent to an element 
\begin{equation}
	\rho = \sum_k \rho^{(k)} \in \clie^\bu(\fg) \otimes \cD(E) , \quad \rho^{(k)} \in \clie^k(\fg) 
e\end{equation}
of bidegree $(1,+)$ which satisfies the Maurer--Cartan equation
\begin{equation}
  \d_\fg \rho + \frac12 [\rho, \rho] = 0 \: .
\end{equation}
Here $\d_{\fg}$ denotes the Chevalley--Eilenberg differential of $\fg$ and $[\cdot,\cdot]$ is the commutator of differential operators. 

We observe that $\rho$ determines a super $L_\infty$ structure on $\fg \oplus \cE$ in such a way that there is a short exact sequence of $L_\infty$ algebras
\begin{equation}
	0 \to \cE \to \fg \oplus \cE \to \fg \to 0 
\end{equation}
where $\cE$ is thought of as an $L_\infty$ algebra with $\mu_k = 0$ for $k > 1$. 
 
Let us take some time to reflect on this definition from the physics point of view. It is well known that the supersymmetry algebra is sometimes only realized on-shell or up to gauge transformations. This is precisely captured in the fact that we used a \textit{super $L_\infty$-map} $\fg \rightsquigarrow \big(\cD(E)\; , \; [D,-]\big)$ to define a multiplet instead of a super Lie map. The higher order terms $\rho^{(j)}$ for $j \geq 2$ precisely correspond to closure terms correcting $\rho^{(1)}$ by a gauge transformation or contributions proportional to an equation of motion.

This discussion explains how the supersymmetry algebra acts on the fields of the theory. The operators of the theory consist of functionals of the fields and are denoted by $\cO(\cE)$. For any point $x \in X$, we can define the local operators via
\begin{equation}
	\cO_x(\cE) = \sym^\bu(J^\infty E|_x)^\vee \: ,
\end{equation}
where $J^\infty E$ denotes the jet bundle of $E$. In other words, the local operators at $x$ evaluate polynomials in the fields and derivatives of fields at $x$. Given a map
\begin{equation}
\rho \colon \fg \rightsquigarrow \big(\cD(E)\; , \; [D,-] \big) \:,
\end{equation}
the dual maps $(\rho^{(j)})^\vee$ define an action on the linear local operators, which extends to $\cO(\cE)_x$ via the Leibniz rule. Fixing an element $Q \in \fg$ we can define a map
\begin{equation}
	\delta_Q = \sum_j \rho^{(j)}(Q, \dots, Q)^\vee : \cO_x(\cE) \longrightarrow \cO_x(\cE) \: ,
\end{equation}
which defines the action of $Q \in \fg$ on the operators of the theory.

\subsubsection{Local algebras}
For completeness let us briefly remark that there is a natural way to make the symmetry algebra $\fg$ local as well. This is relevant if $\fg$ encodes a gauge symmetry.
\begin{dfn}
  A \emph{local super $L_\infty$ algebra} on a manifold $X$ is a dgs vector bundle $L \to X$, equipped with a collection of polydifferential operators
  \deq{
    \mu_k: (\cL)^{\times k} \to \cL
  }
  of bidegree $(2-k,+)$ that satisfy the relations of a super $L_\infty$ algebra structure. Here $\cL = \Gamma(X,L)$ are the smooth sections of $L$.
\end{dfn}
The definition of a local module structure generalizes in obvious fashion. We note that, given a super $L_\infty$ algebra $\lie{g}$, the constant sheaf $\ul{\fg}$ is \emph{not} an example of a local super $L_\infty$ algebra for $d>0$, since it is not given as the smooth sections of any dgs vector bundle. However, we can remedy this by resolving the constant sheaf by the de Rham complex: $\Omega^\bu(X)\otimes \lie{g}$ is a local super $L_\infty$ algebra on~$X$. (This example is relevant to Chern--Simons theory.)

Furthermore the above definition is important in the general context of the BV formalism: A perturbative classical field theory in the BV formalism will be equivalent to a local super $L_\infty$ algebra on~$X$, equipped with a trace map of degree $-3$. We will further review this perspective in what follows.

\subsubsection{Multiplets}
In the context of supersymmetry, we are interested in local modules that satisfy an additional  compatibility condition.
For now, let  $X = V_\RR = \RR^d$ be a  $d$-dimensional affine space and let $V = \CC^d$ be its complexification.\footnote{For the pure spinor superfield formalism it will be useful for us to use complex Lie algebras.}
The Poincar\'{e} group is the group of affine transformations of this space; it is of the form 
\begin{equation}\label{key}
	  \Aff(V) = V \rtimes \Spin(V) \: .
\end{equation}

The complexified Lie algebra $\aff(V)$ is
\begin{equation}
  V \rtimes \lie{spin}(V) \cong V \rtimes \wedge^2 V.
\end{equation}

A \emph{multiplet} is a local module structure for a dgs $L_\infty$ algebra on an \emph{affine}\footnote{A dgs vector bundle is called affine if the total space carries an action of the affine group such that the projection is equivariant with respect to the action of the affine group on $\R^d$.} dgs vector bundle on~$\R^d$, where the $\lie{g}$-action is required to be compatible with the action of the affine algebra in a certain sense. We make this precise with the following definition.

\begin{dfn} \label{def:supermultiplet}
  Let $E$ be an affine dgs vector bundle on~$X = V_\R$, and $\fg$ a super $L_\infty$ algebra equipped with a map 
  \deq{
    \phi: \aff(V) \to \fg \: .
  }
  A \emph{$\fg$-multiplet} is a local $\fg$-module structure on~$E$, such that the pullback of the module structure along~$\phi$ agrees with the natural action on sections of the affine vector bundle. Concretely, this means that the following diagram commutes.
  \begin{equation}
  \begin{tikzcd}
  	\fg \arrow[r, "\rho^{(1)}"] & \cD(E) \\
  	\aff(V) \arrow[ru] \arrow[u,"\phi"]
  \end{tikzcd}
  \end{equation}
\end{dfn}
We think of a multiplet as a derived replacement for the (not necessarily locally free) ``supersymmetric sheaf'' $H^\bu(E)$. Even though this sheaf could be regarded as the central object of study in physics, it is more natural from either the BRST/BV perspective or the perspective of derived geometry to just work at the cochain level. There is again a generalization of this definition to local super $L_\infty$ algebras, where the global affine algebra is replaced by a local $L_\infty$ algebra modeling local isometries. We do not pursue this further here.

We briefly note that this definition implies that the image of $\phi$ is represented strictly on the fields. Furthermore, since the natural action of the affine algebra is effective, the above definition requires implicitly that $\phi$ be injective. So multiplets naturally lead us to study superalgebras that contain the affine algebra as a subalgebra.

We take note of the following examples:
\begin{itemize}
  \item Let $\fh$ be a Lie algebra, and consider the product $\lie{g} = \lie{h} \oplus \aff(V)$, equipped with the obvious choice of~$\phi$. Then a $\lie{g}$-multiplet contains a collection of fields transforming in a local representation of~$\fh$. Flavor symmetry multiplets are an example of this kind.
  \item Let $\conf(V)$ be the Lie algebra of conformal vector fields on~$V$. There is a canonical embedding of~$\aff(V)$ in~$\conf(V)$. Then a $\conf(V)$-multiplet encodes the notion of a conformally invariant multiplet of fields. 
  \item Let $\fp$ be the super-Poincar\'e algebra. It contains $\aff(V)$ as a subalgebra, and a $\fp$-multiplet recovers the usual notion of a supermultiplet.
\end{itemize}
Historically speaking, the construction of interesting multiplets for algebras that were not products was the motivation that led to the origin of supersymmetry; we return to this point (and construct examples of the relevant algebras of physical interest) below. 

\subsection{Further structures on multiplets}
\label{sec: bv and brst}
As we will see in the following sections, the pure spinor superfield formalism naturally produces \emph{multiplets} for the supersymmetry algebra. Some extra data is required to produce a theory out of a multiplet; furthermore, depending on whether or not supersymmetry closes off-shell, the resulting theory may be a BRST or a BV theory, so that the additional data required may differ. There are also conditions on the additional data that ensure that the theory is nondegenerate in an appropriate sense. We set up some formalism for the required additional structure in this section.

\subsubsection{BRST data}

In the BRST formalism, a perturbative field theory is described by a local  super $L_\infty$ algebra $\cL$ equipped with a BRST action functional $S$, which is invariant for the $L_\infty$ structure. 
The $L_\infty$ structure describes the (higher) infinitesimal gauge transformations and the variation of the BRST action gives rise to the equations of motion. 

\begin{dfn}
	A {\em BRST datum} on the $\fg$-multiplet $(F, D, \rho)$ consists of
	\begin{itemize}
          \item a local super $L_\infty$ structure $\{\mu_k\}$ on $L \define F[-1]$ such that $\mu_1 = D$, and whose associated 
		Chevalley--Eilenberg differential we denote by $Q_\text{BRST}$; and
              \item a local functional $S_0 \in \oloc(F)$ of bidegree $(0,+)$, called the ``BRST action functional,'' which is closed with respect to~$Q_\text{BRST}$. 
	\end{itemize}
	This data should be such that all maps in the short exact sequence 
	\begin{equation}
		0 \to \cL \to \fg \oplus \cL \to \fg \to 0
	\end{equation}
	are $L_\infty$ maps, and $S_0$ is invariant for the $L_\infty$ action $\rho$. 
\end{dfn}

\subsubsection{BV data}
\label{sec:BV}

The (classical) Batalin--Vilkovisky (BV)~\cite{BV1,BV2,BV3} formalism is a generalization of the BRST formalism whereby the equations of motion are encoded in a derived way. 
For a comprehensive review of the classical BV formalism we refer to \cite{CG2} (see also~\cite{JurcoReview,MnevBV}). We recall the general idea briefly.
 
Perturbatively, a BV theory is described by a local super $L_\infty$ algebra $L_{\rm BV}$ equipped with an invariant, skew-symmetric, non-degenerate, local pairing of degree $-3$ (see the definition below). 
The space of ``BV fields" is the space of sections of the bundle $E = L_{\rm BV}[1]$ given by the shift in $\ZZ$-degree of the $L_\infty$ algebra encoding the BV theory.

The local pairing on $L_{\rm BV}$ defines a local pairing of degree $-1$ on the space of BV fields $\cE$.
There is an equivalent way to encode the BV formalism in terms of the {\em BV action}. In turn, the $L_\infty$ structure on $L_{BV}$ is encoded by the BV action functional $S_{\rm BV}$ in the sense that the Lie algebra cohomology of $L_{\rm BV}$ can be expressed as
\begin{equation}
	\clie^\bu(L_{\rm BV}) = \big(\cO(\cE) \; , \; \{S_{\rm BV}, -\} \big) \: .
\end{equation}
Here $\{-,-\}$ is the BV bracket defined by the non-degenerate pairing on the fields. 

The zeroth cohomology of this cochain complex is the space of functions on the critical locus of the BRST action modulo gauge equivalence. 
The condition that $\{S_{\rm BV},-\}$ define a differential is equivalent to the ``classical master equation''
\begin{equation}
	\{S_{\rm BV}, S_{\rm BV}\} = 0 \: .
\end{equation}

Of course, proper care must be taken to make rigorous sense of the BV complex above. 
As we are working perturbatively, the space of BV fields will arise as the space of sections of some graded vector bundle on spacetime. 
Furthermore, the BV action will be given as the integral of a Lagrangian density of the fields. 
More details on the BV formalism can be found in~\cite{CG2}.

For any multiplet, we will define a notion of ``BV datum,'' which  consists  of the set of data necessary to construct a BV theory (a $(-1)$-shifted invariant symplectic pairing, with respect to which the homotopy $\fg$-action is defined by Hamiltonian vector fields, and a BV action functional that is compatible with the action of~$\fg$). A BV theory will then consist of a multiplet equipped with a BV datum that satisfies an additional nondegeneracy condition.
\begin{samepage}
	\begin{dfn}
		\label{dfn:BVdatum}
		A {\em BV datum} on a $\fg$-multiplet $(E, D, \rho)$ consists of:
		\begin{itemize}
			\item a graded antisymmetric map
			\begin{equation}
			\langle -,- \rangle_\text{loc} : E \otimes E \longrightarrow \mathrm{Dens}_X
			\end{equation}
			of bidegree $(-1, +)$, which is fiberwise non-degenerate; and
			\item a $\clie^\bu(\fg)$-valued BV action
			\begin{equation}
			S_{{\rm BV}, \fg} = \sum_k S_{{\rm BV}, \fg}^{(k)} \in \clie^\bu(\fg) \otimes \oloc(E), \quad S_{{\rm BV}, \fg}^{(k)} \in \clie^k(\fg) \otimes \oloc(E)
			\end{equation}
			of bidegree $(0, +)$ of the form
			\begin{equation}
			S^{(0)}_{{\rm BV}, \fg} (\Phi) = \int_X \langle \Phi, D \Phi \rangle_\text{loc} + I_{\rm BV} (\Phi)
			\end{equation}
			where $I_{\rm BV}(\Phi)$ is a Lagrangian that is at least cubic in the fields and where
			\begin{equation}
			S^{(k)}_{{\rm BV}, \fg} (x_1,\ldots, x_k ; \Phi) = \int_X \langle \Phi, \rho^{(k)} (x_1,\ldots , x_k) \Phi \rangle_{\rm loc} 
			\end{equation}
		\end{itemize}
		such that 
		\begin{itemize}
			\item[(i)]
			$\<-,-\>_{\rm loc}$ is invariant for the $L_\infty$ action $\rho$;
			\item[(ii)] the total action $S_{{\rm BV}, \fg}$ satisfies the $\fg$-equivariant master equation
			\begin{equation}
			\d_{\fg} S_{{\rm BV}, \fg} + \frac12 \{S_{{\rm BV}, \fg}, S_{{\rm BV}, \fg}\} = 0 .
			\end{equation}
		\end{itemize}
	\end{dfn}
\end{samepage}

If $D$ is elliptic, then $S_{{\rm BV}, \fg}^{(0)} (\Phi)$ is a $\fg$-equivariant perturbative BV theory in the sense of~\cite{CG2}. 
According to the terminology in loc.~cit., this total action $S_{{\rm BV}, \fg}$ endows $S_{{\rm BV}, \fg}^{(0)}(\Phi)$ with the structure of a $\fg$-equivariant theory. We will refer to a multiplet equipped with a BV datum for which $D$ is elliptic as a $\fg$-equivariant \emph{BV theory}.

To go from a multiplet with BRST datum to a multiplet with BV datum, one considers 
\begin{equation}
	L_{\rm BV} = L \oplus L^\vee[-3]
\end{equation}
which is equipped with a canonical evaluation pairing of degree $(-3)$. 
The BRST action deforms the obvious $L_\infty$ structure on the direct sum of $L$ with $L^\vee[-3]$, thus giving rise to an $L_\infty $ structure on $L_{\rm BV}$ for which the evaluation pairing is invariant (after an application of the homological perturbation lemma, which can be thought of as solving the classical master equation for $S_{BV}$ order by order). 

We will say that a multiplet equipped with a BRST datum is a \emph{BRST theory} when the corresponding BV datum itself defines a BV theory, meaning that the kinetic term in the BV action involves an elliptic operator.

Note that, in the way we have set things up, any multiplet can be equipped with a trivial BRST datum, whereas a BV datum may not always exist. In~\S\ref{sec: data} we will see that some of the multiplets produced in the pure spinor formalism can be naturally equipped with nondegenerate BV data, while this is not possible for others. Of course, a degenerate BRST datum does not, in itself, define a BRST theory.

For a multiplet with BV datum $(E,D,\rho, \langle.,.\rangle )$, the inner product always allows us to write
\begin{equation}
	E  = F \oplus F^\vee[-1] \: ,
\end{equation}
where $F = \oplus_{k\leq 0} E_{\rm BV}^k$. 
This induces a splitting on the space of sections
\begin{equation}
	\cE = \cF \oplus \cF^![-1] \: . 
\end{equation}
Note that this is a splitting on the level of super vector spaces.

\begin{dfn} \label{admitsBRST}
  A BV multiplet $(E,D,\rho)$ is \emph{off-shell} if the above splitting exists on the level of $\fg$-modules. Then $F$ is naturally a BRST multiplet, and $(F^\vee[-1],D|_{\cF^!},\rho|_{\cF^!})$ is called the antifields multiplet for $(F,D|_\cF,\rho|_\cF)$.
\end{dfn}

Intuitively, this definition means that it is possible to consider the $\fg$-action separately on the fields and antifields. Then, the equations of motions are not needed to close the algebra and the only corrections for the action come from gauge transformations.

\section{The pure spinor superfield formalism}
\label{sec: formalism}
\subsection{A universal construction}
Let $\fg$ be a super Lie algebra, and $Y$ its nilpotence variety, viewed as an affine scheme as discussed above. Let $M$ be any (dgs) $\fg$-module, and $\Gamma$ any graded module for the graded ring~$\O_Y$. (We can equivalently view $\Gamma$ as defining a sheaf on $\Proj \O_Y$.) Then there is a map
\deq{
\rho: \fg \to \End(M) 
}
defining  the $\fg$-module structure, and an obvious map  
\deq{
m: \fg_-^\vee \to \End(\Gamma)
}
given by left multiplication (after recalling that $\fg_-^\vee$ includes into~$\O_Y$ in weight one). 
If we consider the tensor product
$M \otimes \Gamma$,
the above two maps define a map
\deq{
\rho\cdot m: \fg_- \otimes \fg_-^\vee \to \End(M\otimes \Gamma) 
}
as explained in the following diagram.
\begin{equation}
\begin{tikzcd}
\fg_- \otimes \fg_-^\vee \arrow[r, "\rho \cdot m"] \arrow[d, "\rho \otimes m"] & \End(M \otimes \Gamma) \\
\End(M) \otimes \End(\Gamma)  \arrow[hookrightarrow, d,] \\
\End(M \otimes \Gamma) \otimes \End(M\otimes \Gamma) \arrow[uur, "\cdot"]
\end{tikzcd}
\end{equation}

That is, we apply $\rho \otimes m$, include the resulting element into $\End(M\otimes \Gamma) \otimes \End(M\otimes \Gamma)$ and finally multiply to obtain an endomorphism of $M \otimes \Gamma$.

The map $\rho\cdot m$ equips $M \otimes \Gamma$ with a canonical square-zero differential $\cD$, defined to be the image of the canonical element 
\deq{
1 \in \fg_- \otimes \fg_-^\vee \cong \End(\fg_-).
}
The square of this differential sits in the defining ideal of~$\O_Y$, and thus is zero for any $\O_Y$-module $\Gamma$. 

\begin{rmk}
In the case that $\Gamma = \O_Y$ is the structure sheaf, we note that this construction is closely related to the following construction: As in derived geometry, we define the ``classifying space'' of a super $L_\infty$ algebra $\fg$ to be the derived scheme $B\fg$ whose ring of functions consists of the Lie algebra cochains $\clie^\bu(\fg)$. Then a version of the associated bundle construction associates a sheaf on~$B\fg$ to any $\fg$-module $M$; the global sections of this sheaf are $\clie^\bu(\fg;M)$. 
In the cases we are interested in, there is a close connection between $\O_Y$ and~$\clie^\bu(\fg)$; see~\S\ref{ssec: CE}.
\end{rmk}

\subsection{The case of interest: from sheaves to multiplets}
We now consider a graded Lie algebras $\fg$ which is concentrated in degrees $0$, $1$, and~$2$. 
In keeping with the above discussion, we regard this as a lifted super $L_\infty$ algebra that is concentrated in homological degree zero; as such, only the binary bracket operation can be nonvanishing for degree reasons.

Spelling this out, we have a decomposition
\begin{equation}
	\mathfrak{g} = \mathfrak{g}_0 \oplus \mathfrak{g}_1 \oplus \mathfrak{g}_2  ,
\end{equation}
where $\mathfrak{g}_0$ is an ordinary Lie algebra which acts on $\mathfrak{g}_1$ and $\mathfrak{g}_2$, $\mathfrak{g}_2$ is an abelian Lie algebra, and there is a $\fg_0$-equivariant anticommutator map
\begin{equation}
	\{\cdot,\cdot\} : \mathfrak{g}_1 \otimes \mathfrak{g}_1 \longrightarrow \mathfrak{g}_2.
\end{equation}
Note that there is a subalgebra $\fg_{>0}$ which is an extension of~$\fg_1$ by~$\fg_2$.

The most important examples will be super-Poincar\'{e} algebras; we review how these are constructed in detail below, but just remark here that $\lie{g}_2$ consists of translations and $\lie{g}_1$ of supersymmetries in that case.

As we will see momentarily, $\fg_2$ will play the role of the spacetime on which the multiplet is constructed. We note that much of the construction would go through if $\fg$ were any nonnegatively graded Lie algebra. In such a case, however, the bosonic part of $\fg_{>0}$ may not be abelian, and an interpretation of the construction in terms of multiplets on an affine supermanifold will not be immediate. As such, we do not study any examples of this sort here.  

To be very explicit, if we choose a basis $d_\alpha$ of $\mathfrak{g}_1$ and a basis $e_\mu$ of $\mathfrak{g}_2$, we can express the anticommutator map in terms of structure constants $\Gamma^\mu_{\alpha \beta}$
\begin{equation}
	\{d_\alpha,d_\beta\} = \Gamma^\mu_{\alpha \beta} e_\mu \: .
\end{equation}
We denote by $\lambda^1,\dots,\lambda^n$ coordinates on $\mathfrak{g}_1$ dual to the basis $d_\alpha$. Then the defining ideal $I$ of the nilpotence variety is generated by the equations
\begin{equation}
	I = (\lambda^\alpha \Gamma^\mu_{\alpha \beta} \lambda^\beta).
\end{equation}
Its ring of functions is then the quotient ring
\begin{equation}
	R/I = \mathbb{C}[\lambda^1,\dots,\lambda^n]/I.
\end{equation}
We emphasize again that, if $I$ is a radical ideal, then $R/I = \cO(Y)$ coincides with the ring of functions on $Y$ in the sense of classical algebraic geometry. However, this does not need to be the case.

We are interested in a particular example of the construction above, where $M$ is taken to be the $\fg$-module consisting of smooth functions on $\fg_{>0}$ (viewed as a supermanifold). Concretely, 
\deq{
M = C^\infty(\fg_{\geq 0}) = C^\infty(X) \otimes_\C \wedge^\bu( \fg_1^\vee) \: ,
}
where we already identified $X = \fg_2$. There are now two commuting actions of $\fg$ on~$M$, on the left and the right; we denote these by
\begin{equation}
	R,L : \mathfrak{g} \longrightarrow \mathrm{End}(M) .
\end{equation}
Now, for any graded $R/I$-module $\Gamma$ that is equivariant for the $\mathfrak{g}_0$-action, applying the construction above to~$M$ (with respect to the \emph{right} action of~$\fg$) produces a cochain complex
\begin{equation}
	A^\bu(\Gamma) = ( \Gamma \otimes_\C M , \,  {\cD}).
\end{equation}
Explicitly, let $x^\mu$ be linear coordinate functions on~$\fg_2$ and~$\theta^\alpha$ be (odd) linear coordinate functions on~$\fg_1$, dual to the basis $(e_\mu, d_\alpha)$ above. Then the differential is given in coordinates by
\deq{
\cD = \lambda^\alpha R(d_\alpha) = \lambda^\alpha \left( \pdv{ }{\theta^\alpha} - \Gamma^\mu_{\alpha\beta} \theta^\beta \pdv{ }{x^i} \right),
}
where the differential operators act in~$M = C^\infty(\fg_{>0})$ and $\lambda^\alpha$ acts on~$\Gamma$ via the $\O_Y$-module structure.
Checking that $\cD$ squares to  zero explicitly is also straightforward:
\begin{equation}
\begin{aligned}
	\cD^2 &=
	 \lambda^\alpha \lambda^\beta R(d_\alpha) R(d_\beta) 
	= \frac{1}{2} \lambda^\alpha \lambda^\beta \{R(d_\alpha),R(d_\beta)\} \\
	&= \frac{1}{2} \lambda^\alpha \lambda^\beta R(\{d_\alpha,d_\beta\}) 
	= \frac{1}{2} \lambda^\alpha \lambda^\beta \Gamma^\mu_{\alpha \beta} R(e_\mu)
	= 0  .
\end{aligned}
\end{equation}

$A^\bu(\Gamma)$ naturally has the structure of a dgs vector space: we assign bidegree $(1,-)$ to $\lambda^\alpha$, $(0,-)$ to~$\theta^\alpha$, and $(0,+)$ to $x^i$. Since $\fg$ admits a natural lift, there is also a natural  candidate for  a lifted dgs vector  space structure, in which $\lambda^\alpha$ carries bidegree $(1,-1)$, $\theta^\alpha$ bidegree $(0,-1)$, and $x^i$ bidegree $(0,-2)$. However, this lift only defines a sensible bigrading on polynomial  functions on $\fg_2$, rather than on  all smooth functions.  This bigrading is  often  referenced in the pure spinor  superfield literature, often under the names ``ghost number'' and  ``dimension.'' We will  not need it in what  follows,  and  will  view $A^\bu(\Gamma)$ just as a dgs vector space. However, a filtration related to the dimension will play an important role for us.

From this discussion, it is clear that $A^\bu(\Gamma)$ can be viewed as the global sections of an affine dgs vector bundle $E \rightarrow X$ over $X = \fg_2$ with typical fiber
\begin{equation}
	E_x^k \cong \wedge^\bu \fg_1^\vee \otimes_\C (\Gamma)^k \: .
\end{equation}
This is the underlying vector bundle of the multiplet we aim to construct.

We note some properties of this construction below:
\begin{enumerate}
  \item By construction, the left action of $\fg_{>0}$ commutes with the differential $\cD$. As such, the left action defines a \emph{strict} $\fg_{>0}$-module structure, which is equivariant with respect to~$\Aut(\fg_{>0})$ and as such can be extended to a strict action of~$\fg$.
  \item There is an obvious sense in which (a subgroup of) the even part of~$\fg$ consists of affine transformations acting on~$M$. The $\fg$-action is compatible with this inclusion map, and thus makes $A^\bu(\Gamma)$ into a $\fg$-multiplet.
  \item In the definition given above, the notion of a multiplet was designed to capture the notion of a sheaf over spacetime admitting an action of supersymmetry. In physical terms, this sheaf could be thought of as either on-shell or off-shell field configurations up to gauge equivalence. A multiplet, that is a cochain complex of vector bundles with a homotopy action of supersymmetry, can be thought of as a resolution of this sheaf. (This corresponds to studying off-shell supersymmetric theories in the BRST formalism, and on-shell theories in the BV formalism.) The multiplet $A^\bu(\Gamma)$ goes one step further: it resolves a supersymmetric sheaf not just freely over spacetime, but \emph{freely over superspace}. The action of the supersymmetry algebra is thus just the obvious one on functions on superspace, which is both strict and geometric in nature.
  \item It is apparent that $A^\bu(\O_Y)$ has the structure of a commutative algebra, and therefore that the supersymmetric sheaf $H^\bu(A^\bu(\O_Y))$ is also an $A_\infty$ algebra in a canonical way. $A^\bu(\O_Y)$ is a  strict model of this $A_\infty$ structure. 
  \item To sum up, we have constructed a canonical way of associating a multiplet to any equivariant sheaf on~$Y$. Schematically, we depict the construction as an assignment
\begin{equation}
\{\text{Graded equivariant } R/I\text{-modules} \} \ \xrightarrow{\text{Pure spinor formalism}} \ \{\mathfrak{g}\text{-Multiplets} \} \: .
\end{equation}
Better yet, $A^\bu$ defines a functor from the category of (chain complexes of) equivariant $\O_Y$-modules to the category of dgs $A^\bu(\O_Y)$-modules. 
\end{enumerate}

In many examples, there is further structure available, and $(\Gamma \otimes M, \cD)$ can be equipped with a collection of higher brackets endowing it with the structure of an $L_\infty$ algebra. By homotopy transfer this yields an $L_\infty$-structure on the cohomology. In physically relevant examples, such $L_\infty$ structures precisely correspond to those appearing in the BV or BRST description of the underlying field theory. 

To give one example, the ten-dimensional super Yang--Mills multiplet is constructed by considering $A^\bu(\O_Y)$ for the ten-dimensional $\N=1$ supersymmetry algebra. Since $A^\bu(\O_Y)$ is a commutative dgs algebra, we can tensor with any finite-dimensional Lie algebra $\fh$. Then $A^\bu(\O_Y)  \otimes \fh$ is a  dgs Lie algebra that freely resolves the $L_\infty$ structure of the BV description of interacting $\N=1$ super Yang--Mills theory. This description is well-known from work of Berkovits and Cederwall, but we review it in our language below in~\S\ref{sec: 10d} and explicitly derive the standard structures using homotopy transfer.

The general construction we have outlined so far produces a ``large'' multiplet, which, as outlined above, can be thought of resolving a sheaf over spacetime with an action of supersymmetry. Of course we can just move to the cohomology of our multiplet to recover this sheaf; however, one might wonder whether and how a smaller multiplet resolving the same sheaf can be extracted. For example, is there any way of connecting a pure spinor multiplet to the typical component field multiplets, that is to a finite-rank resolution by vector bundles over spacetime? In fact, there is a  general technique for producing ``minimal'' resolutions of this kind, which was discussed in~\cite{MovshevSchwarzXu,LosevCS}. We review it in our language below and give a proof that highlights the relation to standard constructions in algebraic geometry and homological algebra. After that, we will construct our first examples of physically relevant algebras and multiplets.

\subsection{Filtrations and Koszul homology}
\label{ssec: filter}

The object $A^\bu(\Gamma)$ that we have constructed admits a natural filtration $F^\bu A^\bu(\Gamma)$; understanding the spectral sequence associated to this filtration will allow us to relate the multiplets we are constructing to finite-rank vector bundles over the spacetime $X$. The filtration is associated to a second integer grading; we will find that, while not all of the structures we are interested in preserve this second grading, they do play nicely with the associated filtration. 
The filtration degree is defined by the assignments in the following table:
\deq{
  \begin{array}{c|c|c|c}
    & \text{homological degree} & \text{intrinsic parity} & \text{filtered weight} \\ \hline
    x & 0 & + & 0 \\
    \lambda & 1 & - & 1 \\
    \theta & 0 & - & 1 
  \end{array}
}
(These conventions for the filtration follow those used in~\cite{spinortwist}.)

Since $C^\infty(X)$ plays no role in the filtration, we are exhibiting $A^\bu(\Gamma)$ as a filtered dgs vector bundle over~$X$. Moreover, since the filtration plays well with the product structure on the algebra $A^\bu(\O_Y)$, it gives rise to the structure of a filtered commutative dgs algebra there. However, 
we observe that the tautological differential does not respect the integer grading by filtration weight. Recall that, in coordinates,
\deq{
  \cD = \cD_0 + \cD_1 = \lambda^\alpha  \pdv{ }{\theta^\alpha} - \lambda^\alpha  \Gamma^\mu_{\alpha\beta} \theta^\beta \pdv{ }{x^\mu} .
}
As the notation suggests, the differential is the sum of two terms, which have filtered weight zero and two respectively. The associated graded is thus equipped only with the differential $\cD_0$, which is independent of smooth functions on $X$. We could then write the resulting complex in the following form:
\deq{
  \Gr A^\bu(\Gamma) = \left( C^\infty(X) \otimes_\C \left( \Gamma \otimes_\C \C[\theta^\alpha]\right) , \cD_0 = \lambda^\alpha \pdv{ }{\theta^\alpha} \right) 
  \cong C^\infty(X) \otimes_\C K^\bu(\Gamma).
}
Here we have defined the \emph{Koszul homology} of any $R$-module in standard fashion:
\deq{
  K^\bu(\Gamma) := \left( \Gamma \otimes_\C \C[\theta^\alpha] \: , \: \cD_0 = \lambda^\alpha \pdv{ }{\theta^\alpha} \right).
}
The fact that $\Gamma$ is an $R/I$-module is of course vitally important for our construction, but Koszul homology makes sense for any $R$-module. 
In the pure spinor superfield literature, the cohomology of $\Gr A^\bu$ is often referred to as ``zero mode cohomology''~\cite{Cederwall}. 

If we consider the spectral sequence associated to this filtration, we find that the $E_1$ page is just given by 
\deq{
  H^\bu(\Gr A^\bu(\Gamma)) = C^\infty(X) \otimes_\C H^\bu(K^\bu(\Gamma)).
}
Since $\Gamma$ is a graded module, the Koszul homology of~$\Gamma$ is a finite-dimensional bigraded representation of the Lorentz group. As such, $H^\bu(\Gr A^\bu(\Gamma))$ determines a vector bundle over $X = \fg_2 \cong \R^n$ with fibers
\begin{equation}
	(E'_x)^{k} \cong H^\bu(K^\bu(\Gamma))^{(k)} \: . 
\end{equation}
We emphasize that the homological degree of $E'$ is determined by the internal (weight) grading on $\Gamma$, whereas the parity is determined by the homological degree in Koszul homology modulo two.
$\cD_1$ induces a new differential $\cD'$ acting on the sections of this vector bundle via homotopy transfer of $D_\infty$-algebras. In addition, the $\fg$-module structure transfers as well such that $(E',\cD',\rho')$ is again a multiplet. This multiplet precisely corresponds to the component field description of multiplets as they are known from the physics literature. The transferred differentials play the role of BRST or BV differentials.

Of course one could go on and consider the full cohomology of $A^\bu(\Gamma)$. If the transferred differential $\cD'$ on the component field level does not already vanish, then the resulting object will no longer be free over spacetime, i.e. it does not consist of vector bundles and thus does not fit our definition of a multiplet. It is, however, still a sheaf on spacetime which carries a $\fg$-action. Physically speaking this sheaf consists of the on-shell, gauge invariant states of the multiplet.

Let us summarize these relations by the following diagram.
\begin{equation}
\begin{tikzcd}
(A^\bu(\Gamma),\cD) \arrow[d, "HT"] & \text{Free over superspace}\\
(H^\bu(\mathrm{Gr}A^\bu), \cD')  \arrow[d, "HT"] & \text{Free over spacetime} \\
(H^\bu(A^\bu(\Gamma)),0) & \text{Not necessarily free}
\end{tikzcd}
\end{equation}

The compatibility of the differential with the filtration in fact arises from a compatibility of the left and right $\fg$-actions with the filtration, once $\fg$ is filtered in an appropriate way. Using the standard definition of a complete filtered Lie algebra~\cite{KN,Koch}, we can equip $\fg$ with a filtered structure by setting 
\deq{
  \fg = \fg^{(-1)} \supset \fg^{(0)} = \fg_+. 
}
We observe that this filtration corresponds to the one we defined above, viewing the pure spinor superfield as constructed from functions on superspace together with the degree-zero Lie algebra cohomology of~$\fg_{>0}$ (see~\S\ref{ssec: CE}). 

The associated graded super Lie algebra $\Gr(\fg)$ is then the extension of $\fg_0$ by the \emph{abelian} module consisting of $\fg_1 \oplus \fg_2$; said differently, we set the bracket between odd elements to zero. It is immediate that there is a $\Gr(\fg)$-module structure on~$\Gr A^\bu(\Gamma)$. We will be able to derive this module structure, which consists of ``all supersymmetry transformations that are independent of spacetime derivatives,'' efficiently in examples, using purely algebrogeometric information about~$\Gamma$.

\subsection{Examples of interest: Supersymmetry algebras}

We are mostly interested in multiplets for supersymmetry algebras on an affine spacetime $X=V_\R$.
Depending on the dimension, $\Spin(V)$ will have either one or two spinor representations, which we note by $S$ or~$S_\pm$ respectively; furthermore, there will be an equivariant map $\Gamma$ that witnesses $V$ as a submodule of the tensor square of the spin representation.

We construct the space $\fp_1$ by taking the tensor product of a spin representation with an auxiliary vector space~$U$, which (depending on dimension) may or may not be equipped with either a symmetric or antisymmetric bilinear form. The bracket must be constructed from the pairing $\Gamma$; if $\Gamma$ pairs one spin representation with the other (in dimension $0\bmod 4$), we tensor one spin representation with~$U$ and the other with~$U^\vee$. If $\Gamma$ is a symmetric self-pairing (as in dimensions $1$, $2$, and~$3\bmod 8$), $U$ should have  a symmetric bilinear form; similarly, if $\Gamma$ is an antisymmetric self-pairing on a spin representation (as in dimensions $5$, $6$, and~$7\bmod8$), $U$ must be a symplectic vector space. The ``degree of extended supersymmetry,'' denoted $\N$, is the dimension of~$U$ as a multiple of the smallest possible dimension (two in the symplectic case and one otherwise). In cases where a self-pairing exists on chiral spin representations (dimension $2$ and $6\bmod8$), two independent choices of $\N$ are possible, one for each chirality.
By abuse of notation, we will also write $\Gamma$ for the symmetric pairing on~$\fp_1$. 

The supertranslation algebra $\ft := \fp_{>0}$ is then an extension
\deq{
  0 \rightarrow \fp_2 \rightarrow \ft \rightarrow \fp_1 \rightarrow 0,
}
where $\fp_1$ and~$\fp_2$ are abelian graded Lie algebras and the bracket on~$\ft$ consists of the equivariant symmetric map $\Gamma$ constructed above. The full super-Poincar\'e algebra $\fp$ adds in the automorphisms of~$\ft$ in degree zero; these consist of~$\text{Lie}(\Spin(V)) = \lie{so}(V)$, together with the automorphisms of~$U$ that preserve the pairing if present: either $\lie{gl}(U)$, $\lie{so}(U)$, or~$\lie{sp}(U)$, depending on dimension. In physics, this additional automorphism is known as $R$-symmetry.

The nilpotence varieties of these algebras were studied systematically in~\cite{NV}; most  examples were already present in the previous pure spinor literature. It is worth commenting briefly on the connection to  the classical notion of a ``pure spinor'' given by Cartan. Recall that the spin representation of~$\Spin(V_\R)$ is constructed by choosing a maximal isotropic subspace $L \subset V_\C$. Then $S = \wedge^\bu(L^\vee)$, and $V_\C = L \oplus L^\vee$ acts via Clifford multiplication just by wedging and contracting. (In odd dimensions, $V_\C = L \oplus L^\vee \oplus (L^\perp/L)$, and the single generator in~$L^\perp/L$ acts diagonally by the parity operator.) 

Given the construction of the brackets in~$\fp$, it is clear that an element lying in~$\wedge^0(L^\vee)$ (tensored with any element of~$U$) is automatically square-zero, and that it will be a ``minimal'' or holomorphic supercharge (the image of~$[Q,-]$ is just $L^\perp$). Considered as a projective variety, the space of such elements thus consists of the product of the projective space $P(U)$ and the space $\OGr(n,d)$ of isotropic subspaces $L = \C^n \subset V_\C = \C^d$. (Here $n = \lfloor d/2 \rfloor$.) The latter is the space of Cartan pure spinors, the minimal nonzero $\Spin(d)$ orbit in the spin representation. However, we emphasize that the nilpotence  variety in general contains many more strata, and may even include nonminimal orbits in the spin representation, quite independently of $R$-symmetry (as in eleven dimensions). 

We will not construct all supersymmetry algebras in detail here (for discussion that uses similar style and notation, see~\cite{NV}). We will just introduce examples as we need them, beginning with the four-dimensional $\N=1$ algebra in the next section.

\subsection{Motivating example: the 4d chiral multiplet}
As an explicit example, let us consider the $\mathcal{N}=1$ supersymmetry algebra in four dimensions. 
A related discussion of the chiral multiplet already appeared in~\cite{NV}.

Since the dimension is zero modulo four, $U$ carries no pairing and can be taken to be one-dimensional. $\fp_1$ is then $S_+ \oplus S_-$, and the bracket is constructed using the isomorphism
\deq{
  S_+ \otimes S_- \cong V
}
of~$\Spin(4)$ representations. 
Because this is an isomorphism, the self-bracket of an element $Q\in\fp_1$ is zero precisely when either $Q \in S_+$ or $Q \in S_-$; as such, $Y$ consists of two coordinate planes of the form $\C^2 \subset \C^4$, intersecting at the origin. More precisely,
\begin{equation}
Y = S_+ \cup_{\{0\}} S_- .
\end{equation}

We repeat the same computation in coordinates for emphasis. A general supercharge $Q$ can be written in the form 
\begin{equation}
Q = \lambda^\alpha Q_\alpha + \bar{\lambda}^{\dot{\beta}} \bar{Q}_{\dot{\beta}} \: .
\end{equation}
Accordingly, the equation $\{Q,Q\}=0$ reduces to the four quadratic equations
\begin{equation}
\lambda^\alpha \bar{\lambda}^{\dot{\beta}}  \Gamma^\mu_{\alpha \dot{\beta}} = 0 \: .
\end{equation}
With respect to the decomposition into $S_+$ and $S_-$, the $\Gamma$-matrices are off-diagonal with blocks consisting of the Pauli-matrices $\sigma^\mu$. Multiplying matrices gives the four equations
\begin{equation}
\begin{split}
\lambda^1\bar{\lambda}^1 + \lambda^2 \bar{\lambda}^2 &= 0 \\
\lambda^1\bar{\lambda}^1 - \lambda^2 \bar{\lambda}^2 &= 0 \\
\lambda^1\bar{\lambda}^2 + \lambda^2 \bar{\lambda}^1 &= 0 \\
\lambda^1\bar{\lambda}^2 - \lambda^2 \bar{\lambda}^1 &= 0 \: .
\end{split}
\end{equation}
Adding and subtracting these equations one finally finds
\begin{equation}
\lambda^1 \bar{\lambda}^1 = \lambda^2 \bar{\lambda}^2 = \lambda^1 \bar{\lambda}^2 = \lambda^2 \bar{\lambda}^1 = 0 \: ,
\end{equation}
which implies that $\lambda^\alpha$ or $\bar{\lambda}^{\dot{\beta}}$ vanish and recovers our result from above.

To construct a multiplet, we have to choose an $\cO_Y$-module. One possible choice is $\Gamma = \mathbb{C}[\bar{\lambda}_{\dot{\alpha}}]$, which corresponds to the pushforward of the structure sheaf of~$S_+$ to~$Y$ along the inclusion map. We form the pure spinor complex:
\begin{equation}
  \left( A^\bu(\Gamma)\: , \: \cD \right) = \left( C^\infty(T) \otimes \mathbb{C}[\bar{\lambda}_{\dot{\alpha}}] \ , \ \mathcal{D} = \bar{\lambda}_{\dot{\alpha}} \frac{\partial}{\partial \bar{\theta}_{\dot{\alpha}}} + \bar{\lambda}^{\dot{\alpha}} \theta^\alpha \Gamma^\mu_{\alpha \dot{\alpha}} \partial_\mu \right)  .
\end{equation}
As emphasized above, we can relate this multiplet to the component field formulation by computing the Koszul homology of~$\Gamma$.  Using $\mathfrak{t}_1=S_+ \oplus S_-$, we see that the relevant complex can be written as 
\begin{equation}
\left( \wedge^\bullet S_+ \otimes \wedge^\bullet S_- \otimes \mathbb{C}[\bar{\lambda}_{\dot{\alpha}}] \ , \ \mathcal{D}_0 = \bar{\lambda}_{\dot{\alpha}} \frac{\partial}{\partial \bar{\theta}_{\dot{\alpha}}} \right).
\end{equation}
Here we introduced coordinates on $S_+$ denoted by $\theta_\alpha$ and on $S_-$ written as $\bar{\theta}_{\dot{\alpha}}$. Since $\theta_\alpha$ does not occur in the differential $\mathcal{D}_0$, we find that the cohomology is a tensor product
\begin{equation}
\wedge^\bullet S_+ \otimes H^\bullet(\wedge^\bullet S_- \otimes \mathbb{C}[\bar{\lambda}_{\dot{\alpha}}]) \: .
\end{equation}
However, it is easy to see that the second factor is acyclic, i.e. $H^\bullet(\wedge^\bullet S_- \otimes \mathbb{C}[\lambda_\alpha]) = \mathbb{C}$. Thus, reinstalling the spacetime dependence, the $\cD_0$-cohomology reads
\begin{equation}
\wedge^\bullet S_+ \otimes C^\infty(V) \: .
\end{equation}
We immediately see that we are dealing with two scalar fields in degrees $0$ and $2$ and a Weyl fermion in degree $1$. This is precisely the field content of the chiral multiplet. In Table~\ref{tab:4d chiral reps} we display the corresponding representatives and relate them to the component fields of the chiral multiplet.
\begin{table}
	\begin{center}
		\begin{tabular}{|c|c|}
			\hline
			Field & Representative in the $\cD_0$-cohomology \\
			\hline
			$\phi$  & $\phi$ \\
			\hline
			$\psi$ &$\psi \theta$ \\
			\hline
			$F$ &$F \theta_1 \theta_2$ \\
			\hline
		\end{tabular}
	\end{center}
	\caption{Representatives for the $\cN=1$ chiral multiplet in four dimensions organized by $\theta$-degree.}
	\label{tab:4d chiral reps}
\end{table}

It is clear that the differential $\cD'_1$ acts trivially on these component fields. Hence there are also no further terms induced by homotopy transfer. We thus obtain a multiplet described by a of super vector bundle
\begin{equation}
	E' = V \times \wedge^\bullet S_+  \: ,
\end{equation}
with differential $\cD'=0$.

As the differential $D$ vanishes, this is one of the rare cases where the supersymmetry algebra acts strictly on the component fields. Expanding $Q = \epsilon^\alpha Q_\alpha$ and $\bar{Q} = \bar{\epsilon}^{\dot{\alpha}} \bar{Q}_{\dot{\alpha}}$ we have
\begin{equation} \label{4dQ}
\begin{split}
\rho(Q) = \epsilon \cQ &= \epsilon^\alpha \frac{\partial}{\partial \theta^\alpha} - i (\epsilon \sigma^\mu \bar{\theta}) \partial_\mu \\
\rho(\bar{Q}) = \bar{\epsilon} \bar{\cQ} &= \bar{\epsilon}^{\dot{\alpha}} \frac{\partial}{\partial \bar{\theta}^{\dot{\alpha}}} +i (\theta \sigma^\mu \bar{\epsilon}) \partial_\mu \: .
\end{split}
\end{equation}
The transferrred action only has a $\rho'^{(1)}$ component, which is given explicitly by
\begin{equation}
\begin{split}
\rho'^{(1)}(Q) = p \circ \rho(Q) \circ i &= \epsilon^\alpha \frac{\partial}{\partial \theta^\alpha} \\
\rho'^{(1)}(\bar{Q}) = p \circ \rho(\bar{Q}) \circ i &= i(\theta \sigma^\mu \bar{\epsilon}) \partial_\mu \: .
\end{split}
\end{equation}
Now we can apply these to the representatives to find
\begin{equation}
\begin{matrix}
\rho'^{(1)}(Q) (\phi) = 0 & \rho'^{(1)}(\bar{Q})(\phi) = -i \bar{\epsilon} \slashed{\partial} \phi \theta \\
\rho'^{(1)}(Q) (\psi \theta) = \epsilon \psi  & \rho'^{(1)}(\bar{Q}) (\theta^\beta \psi_\beta) = i \bar{\epsilon} \slashed{\partial} \psi \theta^1 \theta^2 \\
\rho'^{(1)}(Q) (F \theta^1 \theta^2 ) = \epsilon F \theta  & \rho'^{(1)}(\bar{Q}) (\theta^1 \theta^2 F) = 0 \: .
\end{matrix}
\end{equation}
Writing these relations dually in terms of operators we obtain the usual supersymmetry transformation rules.
\begin{equation}
\begin{split}
&\delta \phi = \epsilon \psi \\
&\delta \psi = i \bar{\epsilon}\slashed{\partial} \phi + \epsilon F \\
&\delta F = -i\bar{\epsilon} \slashed{\partial} \psi \: .
\end{split}
\end{equation}

\subsection{Computational techniques: Koszul homology via free resolutions}
\label{ssec: resolve}

In the above example, we were able to compute the cohomology by hand and even could write down explicit representatives easily. In general, such computations are much more convoluted, and we will rely heavily on more advanced techniques. In this section, we show how the cohomology can be computed from the minimal free resolution of the module $\Gamma$ and the corresponding Hilbert series. This allows for a fairly direct identification of the ingredients of the multiplet. Further, using tools from the study of spectral sequences, we can write down explicit formulas for the representatives.

Let us fix a nilpotence variety $Y$ and an $R/I$-module $\Gamma$. To understand the component field description of the multiplet~$A^\bu(\Gamma)$, we are interested in the Koszul homology of $\Gamma$. The following proposition shows that we can understand this by considering a free minimal resolution of $\Gamma$ as an $R$-module.
\begin{prop}[\cite{MovshevSchwarzXu,LosevCS}] \label{prop:res}
	Let $L^\bullet \longrightarrow \Gamma \longrightarrow 0$ be the minimal free resolution of $\Gamma$ in free $R$-modules. Then
	\begin{equation}
	H^\bullet(K^\bullet(\Gamma)) \cong L^\bullet \otimes_R \mathbb{C} \: .
	\end{equation}
\end{prop}
\begin{proof}
	We denote the differential on the minimal free resolution $L^\bullet$ by $d_L$. By definition we have
	\begin{equation}
	H^k(L^\bullet,d_L) = 
	\begin{cases}
	\Gamma, & \text{if } k=0 \\
	0, & \text{else}.
	\end{cases}
	\end{equation}
	This implies that there is a quasi-isomorphism
	\begin{equation}
	\left( \wedge^\bullet \mathfrak{t}_1^\vee \otimes \Gamma \ , \ \mathcal{D}_0 \right) \simeq \left( \wedge^\bullet \mathfrak{t}_1^\vee \otimes L^\bullet , \mathcal{D}_0 + d_L \right).
	\end{equation}
	Thus we may as well compute the cohomology of the complex on the right. Since the differential decomposes into two pieces there, we can use a spectral sequence for this task. Therefore we start with
	\begin{equation}
	\left( \wedge^\bullet \mathfrak{t}_1^\vee \otimes L^\bullet \ , \ \mathcal{D}_0 \right).
	\end{equation}
	It is easy to see that
	\begin{equation}
	H^k\left( \wedge^\bullet \mathfrak{t}_1^\vee \otimes R[-l] \ , \ \mathcal{D}_0 \right) =
	\begin{cases}
	\mathbb{C}, & \text{if} \ k=l \\
	0, & \text{else.}
	\end{cases}
	\end{equation}
	This means that we obtain a copy of $\mathbb{C}$ for each generator of $L^\bullet$. In total we get
	\begin{equation}
          H^\bullet\left( \wedge^\bullet \mathfrak{t}_1^\vee \otimes L^\bullet \ , \ \mathcal{D}_0 \right) = L^\bullet \otimes_R \mathbb{C},
	\end{equation}
        where the $R$-module structure on~$\C$ is obtained by applying the canonical augmentation (quotienting out the maximal ideal). 
        The differential on the first page is just the morphism induced by~$d_L$. However, since $L^\bullet$ is minimal, $d_L$ contains no constant terms and therefore induces the zero map on the first page, implying that the result is already exact. Thus we find that
	\begin{equation}
          H^\bullet\left( \wedge^\bullet \mathfrak{t}_1^\vee \otimes \Gamma \ , \ \mathcal{D}_0 \right) \cong L^\bullet \otimes_R \mathbb{C},.
	\end{equation}
        as claimed.
\end{proof}
The proposition reduces the task of computing Koszul homology to the task of finding a minimal (equivariant) free resolution of~$\Gamma$. This can easily be done with commutative algebra software such as~\textit{Macaulay2}~\cite{M2}. As a result one obtains the Betti numbers of the complex. This gives us information about the number of fields in our multiplet.
To understand which fields are part of our multiplet we have to identify the cohomology not only as a vector space, but as a representation of the Lorentz and the $R$-symmetry group. This is accomplished by introducing additional gradings for the $\lambda^i$, which allows us to extract the relevant information from the Hilbert series. Let us quickly review the main ingredients.

Let $\Gamma = \bigoplus_{i\ge 0} \Gamma_i$ be a graded $R$-module generated by finitely many elements in positive degree. The Hilbert series of $\Gamma$ is defined as the formal power series
\begin{equation}
\textit{HS}_\Gamma = \sum_{n=0}^{\infty} \text{dim}(\Gamma_n) \ T^n.
\end{equation}
Let $R=\CC[\lambda]$ be the polynomial ring in a single variable $\lambda$. Since there is only a single monomial in each degree the Hilbert series takes the form
\begin{equation}
	\textit{HS}_R = \sum_{n=0}^{\infty} T^n = \frac{1}{1-T} \: .
\end{equation}
As the dimension is multiplicative under the tensor product, the Hilbert series of a polynomial ring in $n$ variables $R = \CC[\lambda_1,\dots,\lambda_n] = \CC[\lambda_1] \otimes \dots \otimes \CC[\lambda_n]$ is just the product
\begin{equation}
	\textit{HS}_R = \frac{1}{(1-T)^n} \: .
\end{equation}
Now suppose we perform a shift $R(-d)$ with respect to the polynomial degree such that the constants are in degree $d$. We obtain for the Hilbert series
\begin{equation}
\begin{split}
	\textit{HS}_{R(-d)} &= \sum_{n=0}^{\infty} \mathrm{dim}(R(-d)_n) \ T^n \\
	&= \sum_{n=0}^{\infty} \mathrm{dim}(R_{n-d}) \ T^n \\
	&= T^{d} \textit{HS}_R \\ 
	&= \frac{T^d}{(1-T)^n} \: .
\end{split}
\end{equation}
Thus, considering a free $R$-module $\Gamma$ generated by elements in degree $d_1,\dots,d_k$ we find
\begin{equation} \label{eq:hs free}
	\textit{HS}_{\Gamma} = \frac{T^{d_1}+ \dots + T^{d_k}}{(1-T)^n} \: .
\end{equation}
The Hilbert series is additive with respect to short exact sequences. This means given a sequence
\begin{equation}
0 \longrightarrow A \longrightarrow B \longrightarrow C \longrightarrow 0 \: ,
\end{equation}
we find
\begin{equation}
\textit{HS}_B = \textit{HS}_A + \textit{HS}_C \: .
\end{equation}
If $L^\bullet$ is a free resolution of $\Gamma$, we have  a sequence
\begin{equation}
\Gamma \longleftarrow L^0 \longleftarrow L^{1} \longleftarrow \dots \longleftarrow L^{k-1} \longleftarrow L^k \longrightarrow 0 \: .
\end{equation}
Then the additivity implies
\begin{equation}
\textit{HS}_\Gamma = \sum_{j=1}^{k} (-1)^{j-1} \textit{HS}_{L^j}.
\end{equation}
Using this together with~\eqref{eq:hs free}, we can express the Hilbert series of $\Gamma$ in terms of the degrees of the basis vectors of the free resolution
\begin{equation} \label{eq:hilbertres}
\textit{HS}_\Gamma = \sum_{j=1}^{k} (-1)^{j-1} \frac{T^{d_1^j}+\dots+T^{d^j_{n_j}}}{(1-T)^n}.
\end{equation}
Coming back to our original question, we see that the Hilbert series of $\Gamma$ as a $R$-module contains all the information about the degrees of a basis of the minimal free resolution, which in turn coincides with the cohomology. All we have to do is to store the information about the transformation behavior under Lorentz and R-symmetry in the grading. Therefore, we assign to $\lambda^i$ the degree
\begin{equation}
\mathrm{deg}(\lambda^i) = (1,w_1^i,\dots,w_l^i) \: ,
\end{equation}
where $w_1^i,\dots,w_l^i$ are the weights of the Lorentz and $R$-symmetry representation. The first entry $1$ remembers the cohomological degree. The Hilbert series then becomes a polynomial in $l+1$ variables $T_0,\dots,T_l$. Equation~(\ref{eq:hilbertres}) remains valid, but we have to replace $T^{d^j_i}$ by products of $T_0,\dots,T_l$ where each factor is exponentiated with a separate degree. Initializing such a grading in \textit{Macaulay2} and computing the Hilbert series, we can read off the weights of a basis of the cohomology in each degree, allowing to identify the cohomology as a representation of Lorentz- and R-symmetry.

Examining the proof of Proposition~\ref{prop:res} closely, we can deduce a procedure to write down explicit representatives for the cohomology classes. Recall that we used the quasi-isomorphism
\begin{equation}
\left( \wedge^\bullet \mathfrak{t}_1^\vee \otimes \Gamma \ ,\ \mathcal{D}_0 \right) \simeq \left( \wedge^\bullet\mathfrak{t}_1^\vee \otimes L^\bullet \ , \ \mathcal{D}_0 + d_L \right).
\end{equation}
On the right side we have a double complex of the form shown in Table~\ref{tab: bicomplex}.
\begin{table}
\begin{equation} \label{double}
\begin{tikzcd}
\vdots \arrow[d, "\mathcal{D}_0"] & \vdots \arrow[d, "\mathcal{D}_0"] & \vdots \arrow[d, "\mathcal{D}_0"] \\
L^0 \otimes \wedge^2\ft_1^\vee  \arrow[d, "\mathcal{D}_0"] & L^1 \otimes \wedge^2\mathfrak{t}_1^\vee \arrow[l, "d_L"] \arrow[d, "\mathcal{D}_0"] & L^2 \otimes \wedge^2\ft_1^\vee \arrow[l, "d_L"] \arrow[d, "\mathcal{D}_0"] & \dots \arrow[l, "d_L"]  \\
L^0 \otimes \wedge^1\ft_1^\vee  \arrow[d, "\mathcal{D}_0"] & L^1 \otimes \wedge^1\ft_1^\vee \arrow[l, "d_L"] \arrow[d, "\mathcal{D}_0"] & L^2 \otimes \wedge^1\ft_1^\vee \arrow[l, "d_L"] \arrow[d, "\mathcal{D}_0"] & \dots \arrow[l, "d_L"] \\
L^0 \otimes \wedge^0\ft_1^\vee & L^1 \otimes \wedge^0\ft_1^\vee \arrow[l, "d_L"] & L^2 \otimes \wedge^0\ft_1^\vee \arrow[l, "d_L"] & \dots \arrow[l, "d_L"]
\end{tikzcd}
\end{equation}
\caption{The bicomplex obtained by using a free resolution to compute Koszul homology}
\label{tab: bicomplex}
\end{table}
There are two different spectral sequences computing the total cohomology: the horizontal sequence starting with the differential $d_L$ and the vertical sequence starting with $\mathcal{D}_0$. In the proof of Proposition~\ref{prop:res} we have seen that the latter already gives the exact result for the zero mode cohomology on the first page. The computational procedure amounts to coming to a better understanding of this bicomplex.

It is a fact that any bicomplex can be understood (non-canonically) as the sum of different indecomposable pieces~\cite{Stelzig}. These pieces are squares
\begin{equation}
\begin{tikzcd}
\bullet \arrow[d] & \bullet \arrow[l] \arrow[d] \\
\bullet & \bullet \arrow[l]
\end{tikzcd}
\end{equation}
and stairs of different lengths
\begin{equation}
\begin{tikzcd}
\bullet
\end{tikzcd}
\qquad
\begin{tikzcd}
\bullet \arrow[d] \\ \bullet
\end{tikzcd}
\qquad
\begin{tikzcd}
\bullet & \bullet \arrow[l]
\end{tikzcd}
\qquad
\begin{tikzcd}
\bullet \arrow[d] \\ \bullet & \bullet \arrow[l]
\end{tikzcd}
\dots
\end{equation}
Here, the bullet denotes the underlying field $\bullet = K$. Crucially, the decomposition can be chosen such that all the arrows are just identity maps. The length of a stair is the number of bullets $\bullet$ occurring. 

One can understand the behavior of spectral sequences by thinking about the ways that these indecomposable pieces contribute to cohomology. It is a matter of inspection to see that stairs of even length are acyclic at the $E_1$ page of one of the two spectral sequences of the bicomplex, but contribute two generators to the $E_1$ page of the other, that cannot be cancelled by the differential on that page just for degree reasons. It is precisely the (vertically or horizontally oriented) stairs of length $2k$ that contribute to differentials on the $E_{k}$ page of the corresponding spectral sequence. 
Stairs of odd length contribute to the total cohomology of the complex, but do so in a bidegree that depends on which spectral sequence is being considered. 
If we consider such a stair, we see that the cohomology with respect to the horizontal differential is concentrated at the upper end, while the cohomology with respect to the vertical differential lives at the lower end. 
They are thus responsible for the breaking of the bigrading to the single homological grading of the total complex.

Now note that the cohomology of our double complex~(\ref{double}) is concentrated in the bottom row $L^\bullet \otimes \wedge^0\ft_1^\vee$ (for the vertical differential $\mathcal{D}_0$) and on the left column $L^0 \otimes \wedge^\bullet\ft_1^\vee$ (for the horizontal differential $d_L$). This implies that we  have odd stairs contributing to the cohomology in the following manner:
\begin{equation}
\begin{tikzcd}
\bullet \arrow[d,blue] & \dots \\
\bullet \arrow[d,red] & \bullet \arrow[d, blue] \arrow[l,blue] & \dots \\
\bullet \arrow[d] & \bullet \arrow[d, red] \arrow[l, red] & \bullet \arrow[d,blue] \arrow[l,blue] & \dots \\
\bullet & \arrow[l] \bullet & \bullet \arrow[l, red] & \bullet \arrow[l, blue] & \dots
\end{tikzcd}
\end{equation}
Classes in the total cohomology by elements on either end of the stair. However, if we want to view the representatives as elements in~$\wedge^\bullet \ft_1^\vee \otimes \Gamma$, we have to apply the spectral sequence starting with~$d_L$, which amounts to choosing the representatives on the upper end $\wedge^\bullet \ft_1^\vee \otimes L^0$ and then projecting onto the quotient. On the other hand, a basis of the vertical $\mathcal{D}_0$-cohomology is clearly provided by the standard basis $e_i \subseteq L_k \otimes \wedge^0\ft_1^\vee = R^{n_k}$. In order to get the desired basis in $\wedge^k \ft_1^\vee \otimes L^0$ we have to walk up the corresponding stair. Since we are now working only with Koszul complexes of maximal ideals in polynomial rings, this can be done explicitly by defining a simple ``inverse'' or ``adjoint'' differential to $\mathcal{D}_0$ by the formula 
\begin{equation}
\mathcal{D}_0^\dagger = \theta^\alpha \frac{\partial}{\partial \lambda^\alpha} .
\end{equation}
Then our discussion implies the following lemma.
\begin{lem} \label{lem:reps}
	Let $\pi: L_0 \longrightarrow R/I$ be the projection. The elements $\pi((\mathcal{D}_0^\dagger d_L)^k e_i)$ form a basis of the cohomology $H^\bullet \left(\wedge^\bullet\ft_1^\vee \otimes R/I\right)$ in $\theta$-degree $k$.
\end{lem}

\subsection{Homotopy transfer to component fields}

The new differential acting on the component fields, as well as the action of the supersymmetry algebra and, if present, an $L_\infty$ structure are obtained from the respective structures via homotopy transfer. For this we need homotopy data
\begin{equation} \label{eq:homotopy data 1}
\begin{tikzcd}
\arrow[loop left]{l}{h}(\Gr A^\bu, \cD_0)\arrow[r, shift left, "p"] &(H^\bullet(\Gr A^\bu)  , \, 0)\arrow[l, shift left, "i"] \: .
\end{tikzcd}
\end{equation}
Using Lemma~\ref{lem:reps}, we can define an inclusion map
\begin{equation}
  i: H^\bullet(\Gr A^\bu) \hookrightarrow \left( A^\bu , \cD_0 \right)
\end{equation}
by sending a cohomology class to this representative. This inclusion is a quasi-isomorphism. In addition, choosing a complementary subspace inside $A$ gives the projection $p$. (We always work equivariantly with respect to Lorentz and $R$-symmetry.) 

\subsubsection*{The differential}
Recall that we decomposed the differential on~$A^\bu$ as the sum of two pieces, of filtered weight zero and two, respectively:
\deq{\cD = \lambda^\alpha \pdv{ }{\theta^\alpha} + \lambda^\alpha \theta^\beta \Gamma^i_{\alpha\beta} \pdv{ }{x^i} = \cD_0 + \cD_1.
}
We can thus view $\cD_1$ as defining a deformation of the differential on~$\Gr A^\bu$, which in turn equips $H^\bu(\Gr A^\bu)$ with a new differential $\cD'$ that is obtained by homotopy transfer of $D_\infty$ structure~\cite{Dotsenko}\cite{operadsBook}\cite{Lapin}. This uses the choice of a homotopy datum to write all of the higher differentials of a spectral sequence as terms in a single differential, acting on the $E_1$ page, whose cohomology is~$E_\infty$. In formulas, we have
\begin{equation}
	\cD' = \sum_{n=0}^{\infty} \cD'_n
\end{equation}
where the pieces are given by
\begin{equation}
	\cD'_n = p \circ \left( (\cD_1 h)^n \cD_1 \right) \circ i .
\end{equation}
(Note that, due to our conventions for the filtration, only differentials on \emph{even} pages are non-trivial; the differential on page~$2n$ is represented by~$\cD'_n$ above.)
Furthermore, we can fix new homotopy data~\cite{Lapin}
\begin{equation}\label{eq:homotopy data 2}
\begin{tikzcd}
\arrow[loop left]{l}{h'}(A, \cD)\arrow[r, shift left, "p'"] &(H^\bullet(A, \cD_0) \, , \, \cD')\arrow[l, shift left, "i'"] \: ,
\end{tikzcd}
\end{equation}
where
\begin{equation}
\begin{split}
	i' &= \sum_{n = 0}^{\infty} i'_n  =  \sum_{n = 0}^{\infty} (h \cD_1)^n \circ i \\
	p' &=  \sum_{n = 0}^{\infty} p'_n  =  p \circ \sum_{n=0}^{\infty} (\cD_1 h)^n \\
	h' &=  \sum_{n = 0}^{\infty} h'_n = h \circ \sum_{n=0}^{\infty} (\cD_1 h)^n \: .
\end{split}
\end{equation}
We can use this homtopy data to transfer further structures, such as the action of the supersymmetry algebra or an $L_\infty$ structure, from $A^\bu(\Gamma)$ to the component field description. Note that, in terms of sum-over-trees formulas, homotopy transfer with respect to the new homotopy data from~\eqref{eq:homotopy data 2} is expressed in terms of~\eqref{eq:homotopy data 1} simply by allowing for unary vertices which are decorated by $\cD_1$.

\subsubsection*{The supersymmetry action}
The supersymmetry action is obtained by a homotopy transfer of $L_\infty$ module structure. As a result one obtains an map of super $L_\infty$ algebras
\begin{equation}
	\rho' \colon \fp \rightsquigarrow \big(\cD(E')\; , \; [\cD',-] \big) ,
\end{equation}
whose component maps can be obtained via sum over trees formulas. For example $\rho'^{(2)}$ is given by
\begin{equation}
	\rho'^{(2)}(x_1,x_2) = p' \circ \left( \rho(x_1) \circ h' \circ \rho(x_2) \pm \rho(x_2)\circ h' \circ \rho(x_1) \right) \circ i \: .
\end{equation}
Interestingly, there is a close link between the resolution differential and the action of the supersymmetry algebra. This connection was already conjectured in~\cite{BerkovitsSupermembrane}, where it was noticed that the non-derivative supersymmetry transformations and their closure terms appear in the resolution differential of eleven-dimensional supergravity. Using our knowledge on the representatives and the homotopy transfer description of the action of the supersymmetry transformations we can make this observation precise and also provide a proof.

For this note that the strict part of a non-derivative supersymmetry transformation acts by
\begin{equation}
	\cQ_0 := \rho_{\partial_x = 0}(Q) = \epsilon^\alpha \frac{\partial}{\partial \theta^\alpha} \: .
\end{equation}
In addition it is easy to see that
\begin{equation} \label{eq:anticomm}
	\{\cQ_0 , \cD_0^\dagger \} = \epsilon^\alpha \frac{\partial}{\partial \lambda^\alpha}
\end{equation}
and obviously
\begin{equation}
	[\cQ_0 , d_L] = 0 \: .
\end{equation}
Now suppose $\cQ_0$ acts on a representative in $\theta$-degree $k$
\begin{equation}
\begin{split}
	\rho'^{(1)}_{\partial_x = 0} (f) &= p \circ \cQ_0 \circ i (f) \\
	&= p \circ \cQ_0 \circ \pi (\cD_0^\dagger d_L)^k (f^i e_i^{(k)}) \: .	
\end{split}
\end{equation}
Here $(e_i^{(k)})$ denotes a basis of $L^k \otimes_R \CC$ and $\pi: L^0 \longrightarrow \Gamma$ the projection. Note that $\cQ_0 \circ \pi = \pi \circ \cQ_0$. In the following, we abbreviate the components of the resolution differential by $d_k := (d_L)_k$. Now we can use the anticommutator relation~\eqref{eq:anticomm} to bring $\cQ_0$ to the right. We find
\begin{equation}
	\rho'^{(1)}_{\partial_x = 0}(f) = p \circ \pi \left( \sum_{j=1}^{k} \cD_0^\dagger d_1 \dots \cD_0^\dagger d_{j-1} \epsilon \frac{\partial}{\partial \lambda} d_j \cD_0^\dagger \dots d_k (f) \right) \: .
\end{equation}
Since we already know the explicit form of the representatives, we can carry out the projection to $\cD_0$-cohomology directly. The only remaining term is the following.
\begin{equation}
		\rho'^{(1)}_{\partial_x = 0}(f) = \pi \left((\cD_0^\dagger d_L)^{k-1} \epsilon \frac{\partial}{\partial \lambda} d_k (f^i e_i^{(k)}) \right) \: .
\end{equation}
Furthermore, only the part of $(d_L)_k$ linear in $\lambda$ can contribute in $\cD_0$-cohomology. Then $\epsilon \frac{\partial}{\partial \lambda}$ simply replaces $\lambda$ with $\epsilon$ in the $d_k$. Let us denote the resulting map by $d^\epsilon_k$ and its components by $(d^\epsilon_k)_i^{\ j}$. Then we find
\begin{equation}
\begin{split}
	\rho'^{(1)}_{\partial_x = 0} (f) &= \pi \left( (\cD_0^\dagger d_L)^{k-1} (d^\epsilon_k)_i^{\ j} e_j^{(k-1)} f^i \right) \\
	&= \pi \left( (\cD_0^\dagger d_L)^{k-1} (e_j^{(k-1)}) (d^\epsilon_k)_i^{\ j} f^i \right)
\end{split}
\end{equation}
Identifying the representative in degree $k-1$ and writing the transformation rule dually in terms of operators, we find
\begin{equation}
	\delta g^j = (d^\epsilon_k)_i^{\ j} f^i \: ,
\end{equation}
where $g^j$ denotes the operator corresponding to the respective representative in $\theta$ degree $k-1$. This shows that linear parts in the resolution differential precisely correspond to the strict part of the non-derivative supersymmetry transformations.

This generalizes to the higher components of the supersymmetry action. For $n \geq 2$, the non-derivative part of $\rho^{(n)}$ acts is given by
\begin{equation}
	\rho'^{(n)} = p \circ \left( \cQ_0 \circ h \circ \cQ_0  \right)^{n-1} \circ i \: .
\end{equation}
For example one finds for $\rho'^{(2)}$
\begin{equation}
\begin{split}
	\rho'^{(2)}_{\partial_x = 0} (Q,Q)(f) &= p \circ \cQ_0 \circ h \circ \cQ_0 \circ i (f) \\
	& = p \circ \cQ_0 \circ h \circ \pi \left( \sum_{j=1}^{k} \cD_0^\dagger d_1 \dots \cD_0^\dagger d_j \epsilon \frac{\partial}{\partial \lambda} d_j \cD_0^\dagger \dots d_k (f) \right)
\end{split}
\end{equation}
Now assuming that the homotopy $h$ acts via $h \circ \pi = \pi \circ \cD_0^\dagger$ we find using $(\cD_0^\dagger)^2 = 0$
\begin{equation}
\begin{split}
	\rho'^{(2)}_{\partial_x = 0} (Q,Q)(f) &= p \circ \pi \left( \cQ_0 \cD_0^\dagger \epsilon \frac{\partial}{\partial \lambda} d_1 \cD_0^\dagger \dots d_k (f) \right) \\
	&= p \circ \pi \left( \epsilon \frac{\partial}{\partial \lambda} \cQ_0 \cD_0^\dagger d_1 \cD_0^\dagger \dots d_k (f) \right) \: ,
\end{split}
\end{equation}
where we used that $\epsilon \frac{\partial}{\partial \lambda}$ commutes with both $\cD_0^\dagger$ and $\cQ_0$.

Now we can again use the anticommutator relation~\eqref{eq:anticomm} to find
\begin{equation}
	\rho'^{(2)}_{\partial_x = 0} (Q,Q)(f) = p \circ \pi \left( \epsilon \frac{\partial}{\partial \lambda} \sum_{j=1}^{k} \cD_0^\dagger d_1 \dots \cD_0^\dagger d_j \epsilon \frac{\partial}{\partial \lambda} d_j \cD_0^\dagger \dots d_k (f) \right) \: .
\end{equation}
Carrying out the projection $p$ on $\cD_0$-cohomology, we see that only one term survives.
\begin{equation}
	\rho'^{(2)}_{\partial_x = 0} (Q,Q)(f) = \pi \left( (\cD_0^\dagger d)^{k-1} (e^{(k-1)}_j) (d_k^{\epsilon^2})^{\ j}_i f^i \right) \: ,
\end{equation}
where $d_k^{\epsilon^2}$ denotes the quadratic part of the resolution differential with $\lambda$'s replaced by $\epsilon$'s. Written in terms of operators this gives a transformation rule
\begin{equation}
	\delta g^j = (d^{\epsilon^2}_k)_i^{\ j} f^i \: .
\end{equation}

Using a similar calculation as above one sees that only the part of order $n$ in the resolution differential contributes to a supersymmetry transformation and we obtain supersymmetry transformation rules of the form
\begin{equation}
	\delta g^j = (d^{\epsilon^n}_k)_i^{\ j} f^i \: .
\end{equation}
Interestingly this provides a direct link between the polynomial degree of the terms in the resolution differential and the homotopy action of the supersymmetry algebra. That is, if the resolution differential is at most quadratic, then the $L_\infty$ module structure will contain at most $\rho'^{(2)}$ corrections.

\subsubsection*{$L_\infty$ structures}
If $(A,\cD)$ carries an $L_\infty$ structure with differential $\cD$, this structure can be transferred as well. For this one uses the usual sum over trees formulas. As we will see below, the transferred $L_\infty$ structure on the component fields can encode the structure of gauge transformations and in some cases also interactions. Note that the new $L_\infty$ structure has $\mu_1' = \cD'$ the transferred differential. We will see this explicitly in the case of ten-dimensional super Yang--Mills theory.

\subsection{An example of the  technique: the 4d gauge multiplet}
To illustrate these techniques, we are going to perform all the necessary calculations for the $d=4$, $\mathcal{N}=1$ vector multiplet by hand. Let $Y = Y(4;1)$ be the nilpotence variety of the $\mathcal{N}=1$ super Poincar\'{e} algebra in four dimensions. We choose the structure sheaf $\cO_Y$ as our equivariant module. Using \textit{Macaulay2} we can compute the minimal free resolution. Its Betti numbers are displayed in the following table.
\begin{center}
	\begin{tabular}{c|ccccc} 
		& $0$ & $1$ & $2$ & $3$ \\
		\hline
		$0$ & $1$ & $-$ & $-$ & $-$ \\ 
		$1$ & $-$ & $4$ & $4$ & $1$ \\
	\end{tabular}
	\captionof{table}{Betti numbers of the minimal free resolution. The horizontal axis denotes degree in $\theta$, while the vertical axis counts powers in $\lambda$.} \label{t:fieldcontent}
\end{center}
To analyze the field content of the multiplet as representations of the Lorentz group, we assign gradings to the generators $\lambda$ and $\bar{\lambda}$, corresponding to their weights under
\begin{equation}
\mathfrak{so}(4) \cong \mathfrak{su}(2) \times \mathfrak{su}(2) \: .
\end{equation}
Concretely this means that we assign the grading
\begin{equation}
\begin{split}
\mathrm{deg}(\lambda_1)&=(1,1,0) \qquad \mathrm{deg}(\lambda_2)=(1,-1,0) \\
\mathrm{deg}(\bar{\lambda}_1)&=(1,0,1) \qquad \mathrm{deg}(\bar{\lambda}_2)=(1,0,-1) \: .
\end{split}
\end{equation}
Then we examine the numerator of the Hilbert series. We organize the terms by degree in the variable $T_0$, which indicates the total degree in the complex. In degree 0 we simply obtain $1$, which means the field in total degree 0 is a scalar. In degree 2 we find the term
\begin{equation}
-T_0^2(T_1 T_2 + T_1 T_2^{-1} + T_1^{-1} T_2 + T_1^{-1} T_2) \: .
\end{equation}
Reading off the highest weights we see that the corresponding representation of $SU(2) \times SU(2)$ is
\begin{equation}
[1,1] = [1,0] \otimes [0,1] \: ,
\end{equation}
which shows that the field in degree 2 is a vector. In degree 3 we obtain
\begin{equation}
T_0^3(T_1 + T_1^{-1} + T_2 + T_2^{-2}) \: .
\end{equation}
Correspondingly, the representation in degree 3 is a direct sum
\begin{equation}
[1,0] \oplus [0,1] \: .
\end{equation}
Hence the field in degree 3 is a Dirac fermion. Finally the term of order 4 is just $-T_0^4$ indicating that the field in degree $4$ is a scalar. This means that we recover the usual field content of the $d=4$, $\mathcal{N}=1$ vector multiplet. \\ \\
To find representatives with the procedure explained above, we need the differential on the free resolution. The minimal free resolution is of the form
\begin{equation}
R \otimes \left( \CC \xleftarrow{(d_L)_1} V \xleftarrow{(d_L)_2} S_+ \oplus S_- \xleftarrow{(d_L)_3} \CC  \right) \: .
\end{equation}
The differential can be described by the matrices
\begin{equation} \label{eq: res diff}
\begin{split}
(d_L)_1 &= \begin{pmatrix}
\lambda_1\bar{\lambda}_1 & \lambda_1\bar{\lambda}_2 & \lambda_2\bar{\lambda}_1 & \lambda_2\bar{\lambda}_2
\end{pmatrix} \\
(d_L)_2 &= \begin{pmatrix}
0 & -\bar{\lambda}_2 & 0 & -\lambda_2 \\
0 &  \bar{\lambda}_1 & \lambda_2 & 0 \\
-\bar{\lambda}_2 & 0 & 0 & \lambda_1 \\
\bar{\lambda}_1 & 0 & \lambda_1 & 0 
\end{pmatrix} \\
(d_L)_3 &= \begin{pmatrix}
\lambda_1 \\ -\lambda_2 \\ -\bar{\lambda}_1 \\ \bar{\lambda}_2
\end{pmatrix} \: .
\end{split}
\end{equation}
Choosing a basis $e_{\alpha \dot{\alpha}}$ of $V$ and $(s_\alpha, \bar{s}_{\dot{\alpha}})$ of $S_+ \oplus S_-$, these maps can be conveniently packaged as follows.
\begin{equation}
\begin{matrix}
(d_L)_1 &: & V & \longrightarrow & \CC &,& A & \mapsto & \lambda^\alpha \bar{\lambda}^{\dot{\alpha}} A_{\alpha \dot{\alpha}} \\
(d_L)_2 &: & S_+ \oplus S_- & \longrightarrow & V &,& (\psi, \bar{\psi}) & \mapsto & (\lambda^\alpha \bar{\psi}^{\dot{\alpha}} + \psi^\alpha \bar{\lambda}^{\dot{\alpha}}) e_{\alpha \dot{\alpha}} \\
(d_L)_3 &: & \CC & \longrightarrow & S_+ \oplus S_- &,& D & \mapsto & (\lambda^\alpha s_\alpha - \bar{\lambda}^{\dot{\alpha}} \bar{s}_{\dot{\alpha}}) D
\end{matrix}
\end{equation}
Note that we can apply the isomorphism $S_+ \otimes S_- \cong V$ by a change of basis $e_\mu = (\sigma_\mu)^{\alpha \dot{\alpha}} e_{\alpha \dot{\alpha}}$. With this description, it is easy to identify representatives in $\cD_0$-cohomology. For example, the vector is represented by
\begin{equation}
A \xmapsto{(d_L)_1} (\lambda \sigma^\mu \bar{\lambda}) A_\mu \xmapsto{\mathcal{D}_0^\dagger} (\lambda \sigma^\mu \bar{\theta} + \bar{\lambda} \sigma^\mu \theta) A_\mu \: .
\end{equation}
For the fermions we find
\begin{equation}
\psi \xmapsto{(d_L)_2} \psi^\alpha \bar{\lambda}^{\dot{\alpha}} e_{\alpha \dot{\alpha}} \xmapsto{\mathcal{D}_0^\dagger} \psi^\alpha \bar{\theta}^{\dot{\alpha}} e_{\alpha \dot{\alpha}} \xmapsto{(d_L)_1} \psi^\alpha \bar{\theta}^{\dot{\theta}} \lambda_\alpha \bar{\lambda}_{\dot{\alpha}} \xmapsto{\mathcal{D}_0^\dagger} \psi^\alpha \bar{\theta}^{\dot{\alpha}} (\lambda_\alpha \bar{\theta}_{\dot{\alpha}} + \theta_\alpha \bar{\lambda}_{\dot{\alpha}})
\end{equation}
A similar calculation gives the complex conjugate representative for $\bar{\psi}$. Finally we can apply the procedure to the auxiliary field.
\begin{equation}
D \xmapsto{\cD_0^\dagger \circ (d_L)_3 } (\theta s - \bar{\theta} \bar{s}) D \xmapsto{(d_L)_2} (\theta^\alpha \bar{\lambda}^{\dot{\alpha}} - \lambda^\alpha \bar{\theta}^{\dot{\alpha}} ) e_{\alpha \dot{\alpha}} D \xmapsto{\mathcal{D}_0^\dagger} 2 \theta^\alpha \bar{\theta}^{\dot{\alpha}} e_{\alpha \dot{\alpha}} D \xmapsto{(d_L)_1} 2 (\theta \lambda) (\bar{\theta} \bar{\lambda}) D \xmapsto{\mathcal{D}_0^\dagger} 2 (\theta^2 \bar{\lambda} \bar{\theta} + \bar{\theta}^2 \lambda \theta) D
\end{equation}
We can summarize these representatives in Table~\ref{tab:4d reps}.
\begin{table}[h]
	\caption{Representatives for the $4D$ $\cN=1$ vector multiplet organized by $\theta$-degree.}
	\begin{center}
		\begin{tabular}{|c|c|}
			\hline
			Field & Representative in the $\cD_0$-cohomology \\
			\hline
			$c$  & $c$ \\
			\hline
			$A$ &$(\lambda \sigma^\mu \bar{\theta} + \theta \sigma^\mu \bar{\lambda}) A_\mu$ \\
			\hline
			$\psi$ &$\psi^\alpha \bar{\theta}^{\dot{\alpha}} (\lambda_\alpha \bar{\theta}_{\dot{\alpha}} + \theta_\alpha \bar{\lambda}_{\dot{\alpha}})$ \\
			$\bar{\psi}$ & $\bar{\psi}^{\dot{\alpha}} \theta^\alpha (\bar{\lambda}_{\dot{\alpha}} \theta_\alpha + \bar{\theta}_{\dot{\alpha}} \lambda_\alpha)$ \\
			\hline
			$D$ & $(\theta^2 \bar{\lambda} \bar{\theta} + \bar{\theta}^2 \lambda \theta) D$ \\
			\hline
		\end{tabular}
	\end{center}
	\label{tab:4d reps}
\end{table}
Note that these representatives are not unique. Other choices are possible; for example one can simplify these representatives by eliminating terms in the image of $\cD_0$. For instance the antisymmetric expression
\begin{equation}
\lambda_\alpha \bar{\theta}_{\dot{\alpha}} - \bar{\lambda}_{\dot{\alpha}} \theta_\alpha
\end{equation}
is clearly in the image of $\cD_0$. This implies that we could represent the vector equally well by $\lambda_\alpha \bar{\theta}_{\dot{\alpha}}$. Similar observations also hold for the other fields.

Let us now study the structure of the multiplet defined by $\cD_0$-cohomology.

\subsubsection*{The differential} By degree reasons, only the first order part $\cD'_1$ of the transferred differential $\cD'$ can act non-trivially on the component fields. Recall
\begin{equation}
\cD'_1 = p \circ \cD_1 \circ i \: , 
\end{equation}
where
\begin{equation}
	\mathcal{D}_1 = (\lambda \sigma^\mu \bar{\theta} + \bar{\lambda} \sigma^\mu \theta)  \partial_\mu \: .
\end{equation}
The only non-vanishing contribution arises by acting on the ghost. There we find
\begin{equation}
\mathcal{D}_1 c = (\lambda \sigma^\mu \bar{\theta} + \bar{\lambda} \sigma^\mu \theta) \partial_\mu c \:.
\end{equation}
Identifying the representative of the gauge field, we see that the differential is simply the de Rham differential
\begin{equation}
c \mapsto dc \: .
\end{equation}
Written dually in terms of operators this gives the BRST differential
\begin{equation}
Q_{\textit{BRST}} A_\mu = \partial_\mu c \: .
\end{equation}
The following picture summarizes the complex on the component field level.
\begin{equation}
\begin{tikzcd}[row sep=0.7cm, column sep=0.7cm]
\Omega^0(\mathbb{R}^4) \arrow[dr, "d"] \\  & \Omega^1(\mathbb{R}^4) & \Gamma(\mathbb{R}^4,S_+ \oplus S_-) & \Gamma(\mathbb{R}^4, \mathbb{C}) \: .  \\
\end{tikzcd}
\end{equation}

\subsubsection*{The action of the supersymmetry algebra} 
As explained above, we can read off the non-derivative part of the supersymmetry transformations directly from the resolution differential. This gives transformation rules
\begin{equation}
\begin{split}
\delta c &= (\epsilon \sigma^\mu \bar{\epsilon}) A_\mu \\
\delta A_\mu &= \epsilon \sigma_\mu \bar{\psi} + \psi \sigma_\mu \bar{\epsilon} \\
\delta \psi &= \epsilon D  \\
\delta \bar{\psi} &= -\bar{\epsilon} D\\
\delta D &= 0 \: .
\end{split}
\end{equation}
Note that there is one higher order component indicating that the action of the supersymmetry algebra is not strict. We will come back to this in a moment.

Now, let us investigate the contribututions containing derivatives. By degree reasons there cannot appear any higher order contributions containing derivatives, such that we can focus on the strict part. The derivative part of $\rho'^{(1)}$ acts on the representatives by
\begin{equation}
\cQ_1 = \epsilon \sigma^\mu \bar{\theta} \partial_\mu + \theta \sigma^\mu \bar{\epsilon} \partial_\mu \: .
\end{equation}
For example we can act on the fermions to find
\begin{equation}
\begin{split}
\cQ_1(\psi) =& (\bar{\epsilon} \sigma_\mu \theta) \partial_\mu \psi^\alpha \bar{\theta}^{\dot{\alpha}} (\lambda_\alpha \bar{\theta}_{\dot{\alpha}} + \theta_\alpha \bar{\lambda}_{\dot{\alpha}}) \\
=& (\bar{\epsilon}^{\dot{\beta}} \sigma^\mu_{\beta \dot{\beta}} \partial_\mu \psi^\alpha ) (\lambda_\alpha \theta^\beta \bar{\theta}^2 + \theta_\alpha \theta^\beta \bar{\lambda}\bar{\theta}) \: .
\end{split}
\end{equation}
Projecting to cohomology this equals
\begin{equation}
\bar{\epsilon} \slashed{\partial} \psi (\lambda\theta \bar{\theta}^2 + \theta^2 \bar{\lambda} \bar{\theta}) \: ,
\end{equation}
such that we can identify a transformation rule
\begin{equation}
\delta D = \bar{\epsilon} \slashed{\partial} \psi \: .
\end{equation}
A similar calculation also holds for the complex conjugate $\bar{\psi}$, as well as for the gauge field and yield the usual supersymmetry transformation rules.

This describes the entire $L_\infty$ module structure of the superymmetry algebra on the four-dimensional, $\cN=1$ vector multiplet. The $\rho'^{(1)}$ part resembles the well known supersymmetry transformations from standard physics textbooks. In addition there is one higher correction. Recall that we found a transformation rule
\begin{equation}
\delta c = (\epsilon \sigma^\mu \bar{\epsilon}) A_\mu \: .
\end{equation}
This corresponds to a map $\rho'^{(2)}$ given by
\begin{equation}
\rho'^{(2)}: \mathfrak{t} \otimes \mathfrak{t} \otimes \Omega^1 \longrightarrow \Omega^0 \qquad (Q_1 \otimes Q_2 \otimes A) \mapsto \iota_{\{Q_1,Q_2\}} A \: .
\end{equation}
We can immediately check that $\rho'^{(2)}$ indeed defines a homotopy correcting for the failure of $\rho'^{(1)}$ to be strict. We clearly have
\begin{equation}
\rho'^{(1)}(Q)(c) = \rho'^{(1)}(\bar{Q}) = 0 \: .
\end{equation}
However the anticommutator of $Q$ and $\bar{Q}$ gives a translation which acts via the Lie derivative
\begin{equation}
\{Q,\bar{Q}\}(c) = L_{\{Q,\bar{Q}\}} (c) \: .
\end{equation}
Thus, according to~(\ref{L-action}) we have to check
\begin{equation}
L_{\{Q,\bar{Q}\}} (c) = -[D , \rho'^{(2)}(Q,\bar{Q})] (c) \: .
\end{equation}
Plugging in $D = d$ we obtain
\begin{equation}
\begin{split}
L_{\{Q,\bar{Q}\}} (c)=& -(d\circ \iota_{\{Q_1,Q_2\} } - \iota_{\{Q_1,Q_2\}} \circ d)(c) \\ 
=& (\iota_{\{Q_1,Q_2\}} \circ d)(c) \: ,
\end{split}
\end{equation}
where the first term vanishes by degree reasons. We immediately see that this is indeed satisfied due to Cartan's magic formula. This discussion illustrates that the $\rho'^{(2)}$-term indeed provides a homotopy for the failure of $\rho'^{(1)}$ to be strict. In terms of physics terminology, $\rho'^{(2)}$ is a closure term for the supersymmetry action, which closes only up to gauge transformations.

\subsubsection*{$L_\infty$ structure}
To treat the non-abelian vector multiplet we can tensor the entire construction with a Lie algebra $\fh$. We notice that $\cO_Y$ is not only an $\cO_Y$-module, but in fact an algebra. Hence $A^\bu(\cO_Y)$ carries an algebra structure such that the tensor product $A^\bu(\cO_Y) \otimes \fh$ comes equipped with an $L_\infty$-structure given by
\begin{equation}
\mu_1 = \mathcal{D} \otimes \mathrm{id}_\fh \qquad \mu_2 = m_2 \otimes
[-,-] \: .
\end{equation}
Here $m_2$ denotes the multiplication in $A^\bu(\cO_Y)$. Since the differential does not interfere with the Lie algebra at all, the component fields of the multiplet take values in $H^\bu(\cO_Y) \otimes \fh$. This is just the field content of the abelian version only now taking values in the Lie algbera $\fh$. The transfer of the $L_\infty$ algebra structure to the component fields is very simple. The differential only acts on the ghost fields via the de Rham differential.
\begin{equation}
	\mu'_1 = d \otimes \mathrm{id}_\fh	: \Omega^0 \otimes \fh \longrightarrow \Omega^1 \otimes \fh \: .
\end{equation}
In addition to the differential, only two-ary brackets arise.
\begin{equation}
	\begin{matrix}
	\mu'_2 :& \Omega^0 \otimes \fh \times \Omega^0 \otimes \fh & \longrightarrow & \Omega^0 \otimes \fh & \mu'_2(c,c) = [c,c] \: \phantom{.} \\
	\mu'_2 :& \Omega^0 \otimes \fh \times \Omega^1 \otimes \fh & \longrightarrow & \Omega^1 \otimes \fh  & \mu'_2(c,A) = [c,A] \: \phantom{.}\\
	\mu'_2 :& \Omega^0 \otimes \fh \times \Gamma(X,S_+ \oplus S_-) \otimes \fh & \longrightarrow & \Gamma(X,S_+ \oplus S_-) \otimes \fh & \mu'_2(c,\psi) = [c,\psi] \: .
	\end{matrix}
\end{equation}
We can also write these dually as a BRST operator.
\begin{equation}
\begin{split}
Q_{\textit{BRST}}c &= -\frac{1}{2} [c,c] \\
Q_{\textit{BRST}}A &= dc + [A,c] \\
Q_{\textit{BRST}}\psi &= [\psi,c] \\
Q_{\textit{BRST}}D &= [D,c] \: .
\end{split}
\end{equation}
Hence we recover the usual BRST complex of the $d=4$, $\mathcal{N}=1$ gauge multiplet. To equip the multiplet with a BRST datum, we could write the usual component field action for the gauge multiplet. In the terminology of~\S\ref{sec: prelim} this action then makes the multiplet a BRST theory.

\subsection{Scheme-theoretic properties: three-dimensional $\N=1$ supersymmetry}
\label{ssec: schemes}
In three dimensions we have the isomorphism $\mathrm{Spin}(3) \cong SU(2)$. The vector representation $V$ corresponds to the three-dimensional representation of $SU(2)$, while the spinor representation $S$ is given by the two-dimensional representation. The anticommutator is provided by the isomorphism
\begin{equation}
\sym^2(S) \cong V \: .
\end{equation}
Therefore the nilpotence variety is simply a point
\begin{equation}
Y = \{ 0 \} \: .
\end{equation}
Even though the nilpotence variety, regarded as a set, is just a point it still may carry an interesting structure as a scheme which allows for the construction of different multiplets. Expanding the equation $\{Q,Q\} = 0$ in coordinates $(\lambda^1,\lambda^2)$ we obtain the equations
\begin{equation}
(\lambda^1)^2 = \lambda^1 \lambda^2 = (\lambda^2)^2 = 0 \: .
\end{equation}
Clearly, the only solution to these equations is $\lambda^1 = \lambda^2 = 0$. However, as we announced earlier we view the $Y$ as the affine scheme $Y = \mathrm{Spec}(\cO_Y)$, where $I$ is the ideal generated by the above elements. Then, by definition, the global sections of its sheaf of rings are $\cO_Y = R/I$. Note that $R/I \ncong \CC$, which we would have used as a ring of functions when considering $Y$ as an affine variety. As we will see momentarily using $R/I$, or equivalently viewing $Y$ as the scheme $\mathrm{Spec}(R/I)$, allows us to construct different multiplets from $R/I$-modules even though $Y$ is just a point.
\subsubsection*{The gauge multiplet}
First of all we can consider $R/I$ itself as an equivariant module. This gives rise to the gauge multiplet in three dimensions. The minimal free resolution has the following Betti numbers.
\begin{center}
	\begin{tabular}{c|ccccc} 
		& $0$ & $1$ & $2$ \\
		\hline
		$0$ & $1$ & $-$ & $-$ \\ 
		$1$ & $-$ & $3$ & $2$ \\
	\end{tabular}
	\captionof{table}{Betti numbers of the minimal free resolution of $R/I$.}
\end{center}
In terms of representations, the free resolution takes the form
\begin{equation}
R \otimes \left( \CC \xleftarrow{(d_L)_1} V \xleftarrow{(d_L)_2} S \right)
\end{equation}
with the differentials being described by
\begin{equation}
\begin{matrix}
(d_L)_1 &: & V & \longrightarrow & \CC &,& A & \mapsto & (\lambda \sigma^\mu \lambda) A_\mu \\
(d_L)_2 &: & S & \longrightarrow & V &,& \psi & \mapsto & (\lambda \sigma^\mu \psi) e_\mu  \: .
\end{matrix}
\end{equation}
Thus, we find that the multiplet contains a one-form field together with its ghost as well as a fermion. The only differential acting on the component fields is the de Rham differential
\begin{equation}
	c \mapsto dc
\end{equation}
which encodes the gauge invariance of the one-form. The non-derivative supersymmetry transformations can be read off from the resolution differential and take the usual form.
\begin{equation}
\begin{split}
	\delta c &= (\epsilon \sigma^\mu \epsilon) A_\mu \\
	\delta A_\mu &= \epsilon \sigma_\mu \psi \\
	\delta \psi &= 0 \: . 
\end{split}
\end{equation}

\subsubsection*{The free superfield}
In addition, we can also consider $\CC = R/(\lambda^1, \lambda^2)$ as an $R/I$-module. This yields the free superfield whose Betti numbers we display in the following table.
\begin{center}
	\begin{tabular}{c|ccccc} 
		& $0$ & $1$ & $2$ \\
		\hline
		$0$ & $1$ & $2$ & $1$ \\ 
	\end{tabular}
	\captionof{table}{Betti numbers for the free superfield.}
\end{center}
Indeed, the Koszul homology of this module is just an exterior algebra $\wedge^\bullet S$ and we just recover the usual superspace description of the free superfield.

\section{From multiplets to theories}
\label{sec: data}
In~\S\ref{sec: bv and brst}, we introduced the notions of BV and BRST data for multiplets. Under certain conditions on  the  module $\Gamma$, the Koszul homology is naturally equipped with  a perfect pairing that equips the corresponding multiplet with a BV datum.
This provides an interesting link between the physics of supersymmetric multiplets and the algebraic geometry of $\O_Y$-modules. In fact, the pure spinor formalism provides many such links between algebrogeometric properties of the module $\Gamma$ and physical properties of the multiplet.

\subsection{Commutative algebra and dualizing complexes}
To approach this topic let us start with a short survey of the relevant notions from commutative algebra. To keep things simple we will work in a basic setting where $R = \CC[\lambda_1,\dots,\lambda_n]$ is the polynomial ring in $n$ variables and the modules will be $R/I$-modules for some ideal $I$. The main source for our discussion is~\cite{HGorenstein}.
\begin{dfn}
	A quotient ring $R/I$ is called a complete intersection, if $I$ can be generated by $r = \mathrm{codim}(R/I) = n - \mathrm{dim}(R/I)$ elements, i.e. $I=(f_1,\dots,f_r)$.
\end{dfn}
Intuitively this definition means that there are no non-trivial relations among the $f_i$. Equivalently we can say that $f_1,\dots,f_r$ forms a regular sequence on $R$. To be clear we recall the definition.
\begin{dfn}
	Let $S$ be a commutative ring and $M$ a $S$-module. A sequence $(x_1,\dots,x_k) \subset S$ is called regular on $M$ if $x_i$ is not a zero divisor in $M/(x_1,\dots,x_{i-1})$ for all $i=1,\dots,k$.
\end{dfn}
One can define a notion of ``size'' for modules by asking for the maximal length of a regular sequence in $M$. The resulting number is called the depth of $M$.
\begin{dfn}
	The depth of a $S$-module $M$ is the maximal length of a regular sequence in $M$ and will be denoted by $\mathrm{depth}(M)$.
\end{dfn}
On general grounds one can show that for any module $\mathrm{depth}(M) \leq \mathrm{dim}(M)$. There is an important class of modules for which equality holds. These are called Cohen--Macaulay modules.
\begin{dfn}
	A module $M$ is called Cohen--Macaulay if $\mathrm{depth}(M)=\mathrm{dim}(M)$.\footnote{Here the correct notion of dimension is the Krull dimension.} 
\end{dfn}
Let us now consider the case where $M = R/I$ is a quotient of a polynomial ring. In this case we can apply the Auslander--Buchsbaum formula
\begin{equation}
\mathrm{depth}(R/I) + l_R(R/I) = n \: ,
\end{equation}
where $l_R(R/I)$ is the length\footnote{The length of a free resolution $L^\bu_R = (L_0 \leftarrow L_1 \leftarrow \dots \leftarrow L_k \leftarrow 0)$ is $k$.} of the minimal resolution $L^\bu$ of $R/I$ by free $R$-modules. So we find that $R/I$ is Cohen--Macaulay if and only if
\begin{equation}
l_R(R/I) = n - \mathrm{dim}(R/I) = \mathrm{codim}_R(R/I) \: .
\end{equation}
This means that we can identify Cohen--Macaulay rings conveniently by their minimal free resolutions: $R/I$ is Cohen--Macaulay if and only if the length equals the codimension. \\ \\
For a quotient ring $R/I$ we can define a dualizing complex by
\begin{equation}
  \omega^\bullet_{R/I} = \RHom_R^\bullet(R/I,R) = \Hom_R(L^\bu,R). 
\end{equation}
We note that the cohomology $H^\bu(\RHom^\bu_R(R/I,R)) \cong \Ext^\bu_R(R/I,R).$
If $R/I$ is Cohen--Macaulay, this cohomology is concentrated in a single degree, namely $\codim(R/I)$. Thus the dualizing complex is in fact quasi-isomorphic to a dualizing module (often also called the canonical module). If the canonical module is trivial (free of rank one), the scheme $\Spec(R/I)$ can be thought of as analogous to a Calabi--Yau space. This property is called Gorenstein.
\begin{dfn}
	A quotient ring $R/I$ is called Gorenstein if $R/I$ is Cohen--Macaulay and the dualizing module $\mathrm{Ext}^{n-d}_R(R/I,R)=R/I$, where $d = \mathrm{dim}(R/I)$.\footnote{This is not the most general definition, but it suits our setting. In general a ring $S$ is called Gorenstein, if $S$ has finite injective dimension as an $S$-module. There is also a notion of Gorenstein modules in the literature, but we do not need this level of generality for our discussion.} 
\end{dfn}
Clearly, the Gorenstein property is stronger than the Cohen--Macaulay property. However, to be a complete intersection is an even stronger condition. We thus have the following chain of implications.
\begin{equation}
\text{Complete intersection} \implies \text{Gorenstein} \implies \text{Cohen--Macaulay}
\end{equation}
The key property of Gorenstein rings which is relevant for us is that their minimal free resolutions are self-dual: If $R/I$ is Gorenstein and $(L^\bullet,d_L)$ is a minimal free resolution, then the dual complex $((L^\bullet)^\vee,(d_L)^\vee)$ is, by definition, a resolution of the dualizing module, which, by assumption, is again $R/I$. Thus $(L^\bullet,d_L)$ and $((L^\bullet)^\vee,(d_L)^\vee)$ are both minimal free resolutions for $R/I$ and hence, due to the uniqueness of the minimal free resolution, they must be isomorphic. 

In fact one can recognize Gorenstein rings conveniently by examining their minimal free resolution:
\begin{prop}
	$R/I$ is a Gorenstein ring $\iff$ The length of the minimal free resolution $L^\bullet$ of $R/I$ is $l_R(R/I) = \mathrm{codim_R(R/I)} =: k$ and $L^k = R$.
\end{prop}
Note that this extends the above statement on the Cohen--Macaulay property. The self-duality of the minimal free resolution induces isomorphisms $L^i \cong (L^{k-i})^\vee$ and thus a non-degenerate pairing
\begin{equation}
L^i \times L^{k-i} \longrightarrow R \: .
\end{equation}
Tensoring both sides with $\CC$ we obtain a pairing
\begin{equation}
(L^i \otimes_R \CC) \times (L^{k-i} \otimes_R \CC) \longrightarrow \CC \: .
\end{equation}
As we explained in~\S\ref{ssec: resolve}, $L^\bullet \otimes_R \CC$ can be identified with Koszul homology and thus with the component fields of the multiplet. As such, if we feed a Gorenstein ring into the pure spinor superfield formalism, we can equip the component fields of the resulting multiplet with a local pairing (a density valued pairing on sections of a vector bundle on spacetime). 
The parity and homological degree will depend on the properties of the free resolution. These pairings are often of physical interest.

\subsection{Supplemental structures on multiplets}
In some cases these pairings can be used to equip multiplets obtained in the pure spinor superfield formalism with a BV datum. Here the prime example is ten-dimensional super Yang--Mills theory which we will discuss below. However, this is not the only relevant case. There are other examples of multiplets obtained from Gorenstein rings where the pairing does not give rise to a BV structure; nevertheless, the natural pairings may still be interesting.

As an easy example, let us once again come back to the chiral multiplet for $\cN =1$ supersymmetry in four dimensions. Recall that we obtained the chiral multiplet from the module $\Gamma = \CC[\bar{\lambda}_{\dot{\alpha}}] = \CC[\lambda_\alpha,\bar{\lambda}_{\dot{\alpha}}]/(\lambda_\alpha)$. This is obviously a complete intersection ring, and thus in particular Gorenstein. The minimal free resolution is of the form
\begin{equation}
  R \longleftarrow R^2 \longleftarrow R,
\end{equation}
and it is clear what the pairing looks like: $L^0 = R$ pairs with $L^2 = R$, while $L^1 = R^2$ pairs with itself. Since this is a perfect pairing on Koszul homology, we obtain
a local pairing on the component fields. Recall that the scalar field was represented simply by $\phi$, the fermion by $\psi = \psi_\alpha \theta^\alpha$ and the auxiliary by $F = F \theta^1 \theta^2$. Thus we get a pairing which is simply induced by the algebra structure on $\wedge^\bullet S_+$ and the projection on the $\theta_1 \theta_2$ component
\begin{equation}
\langle a, b\rangle = (ab)|_{\theta_1 \theta_2} \: .
\end{equation}
So this pairing gives rise to F-term Lagrangians for the chiral multiplet through the following local pairing on component fields
\begin{equation}
\langle \phi , F\rangle_\mathrm{loc}= \phi F , \quad \langle \psi , \psi \rangle_\mathrm{loc} = \psi^\alpha \psi_\alpha \: .
\end{equation}
Similar pairings of course exist for other chiral multiplets with more supersymmetry. Furthermore, we could consider the free superfield; in general, this is constructed by taking $\Gamma$ to be the structure sheaf of the cone point, which arises from the canonical augmentation of the graded ring $R/I$. In four dimensions, this module is just $\CC = \CC[\lambda_\alpha,\bar{\lambda}_{\dot{\alpha}}]/(\lambda_\alpha,\bar{\lambda}_{\dot{\alpha}})$. Then one gets a pairing which projects on the $\theta^2 \bar{\theta}^2$ component. In physics, this pairing gives rise to D-terms.

\subsection{Constructing cotangent theories: six-dimensional $\cN = (1,0)$}
\label{sec: 6d vector}
If a ring is not Gorenstein, there is no perfect pairing on Koszul homology, and the corresponding multiplet cannot obviously be equipped with a BV structure. (We note that this does \emph{not} mean that such multiplets are never on-shell; in six-dimensional $\N=(2,0)$ supersymmetry~\cite{CederwallM5,twist20} and ten-dimensional type IIB supersymmetry~\cite{NV}, BV multiplets with degenerate pairings naturally arise. Details on the pairing are given in~\cite{twist20} at the level of the component fields; we do not study the cochain-level origin of such degenerate pairings here, but hope to return to this question in future work.)
  
  For a Cohen--Macaulay module $\Gamma$ giving rise to a multiplet $(E,D,\rho)$, however, another interesting observation applies: We can consider the dualizing module $\omega_\Gamma$ in the pure spinor superfield formalism. If $(L,d_L)$ is the minimal free resolution of $\Gamma$, then $(L^\vee, (d_L)^\vee)$ is the corresponding minimal free resolution of $\omega_\Gamma$. With the obvious pairing between $L$ and $L^\vee$ we can equip the multiplet corresponding to the direct sum $L \oplus L^\vee [k]$ with a BV datum (for an appropriate shift $k$). In the terminology of Definition~\ref{admitsBRST} the resulting BV multiplet is off-shell and $\omega_\Gamma$ gives rise to the antifield multiplet of $(E,D,\rho)$.

On the other hand, if the input module is not Cohen--Macaulay, the cohomology of the dualizing complex will not be concentrated in a single degree, such that we cannot take a single dualizing module to produce an antifield multiplet. Rather, the antifield multiplet will be represented by a dg module. We will see this below in the case of four-dimensional minimal supersymmetry.

Let us now consider the example of six-dimensional $\N=(1,0)$ supersymmetry. There is an accidental isomorphism $\Spin(6) \cong SU(4)$, under which the spinor representation $S_+$ is identified with the fundamental representation of $SU(4)$ and $S_- = (S_+)^\vee$ with the antifundamental representation. The supertranslation algebra
takes the form
\begin{equation}
	\ft = V \oplus \Pi( S_+ \otimes U) \: ,
\end{equation}
where $U= (\CC^2,\omega)$ is a symplectic vector space. The $R$-symmetry group is thus $\Sp(1) \cong SU(2)$; corresponding indices will be denoted by $i,j$. There is an isomorphism
\begin{equation}
	\wedge^2 S_+ \cong V \: ,
\end{equation}
which is used to express the bracket as
\begin{equation}
	\{-,-\} = \wedge \otimes \omega \:.
\end{equation}
Since $\wedge$ is an isomorphism, an element is square-zero precisely when it is of rank one as an element of~$S_+ \otimes U$. 

In a basis, the supertranslation algebra takes the form
\begin{equation}
	\{Q_\alpha^i , Q_\beta^j\} = \Gamma^\mu_{\alpha \beta} \varepsilon^{ij} P_\mu \: .
\end{equation}
Using coordinates $\lambda^\alpha_i$, the defining equations of the nilpotence variety $Y(6;1,0)$ are given by the $2\times 2$ minors of the matrix
\begin{equation}
\begin{pmatrix}
\lambda^1_1 & \lambda^2_1 & \lambda^3_1 & \lambda^4_1 \\
\lambda^1_2 & \lambda^2_2 & \lambda^3_2 & \lambda^4_2
\end{pmatrix},
\end{equation}
which cut out the space  of  rank-one matrices.
As such $Y$ is a determinantal variety. Taking its structure sheaf $\mathcal{O}_Y$ as the equivariant module in the pure spinor formalism, we recover the $d=6$, $\mathcal{N}=(1,0)$ gauge multiplet. The Betti numbers of the free resolution are displayed in the following table.
\begin{center}
	\begin{tabular}{c|ccccc} 
		& $0$ & $1$ & $2$ & $3$ \\
		\hline
		$0$ & $1$ & $-$ & $-$ & $-$ \\ 
		$1$ & $-$ & $6$ & $8$ & $3$ \\
	\end{tabular}
	\captionof{table}{Betti numbers of the free resolution of $\mathcal{O}_Y$.}
\end{center}
Working equivariantly, one finds that these correspond to a one-form with zero-form gauge invariance, fermions in $S_+ \oplus S_-$, and a triplet of scalars in the adjoint of the $R$-symmetry group $SU(2)$. We immediately see that the Koszul homology corresponds to the field content of the BRST complex of the gauge multiplet. Since the length of the resolution equals the codimension, $R/I$ is Cohen--Macaulay. This can also be seen as a consequence of the following result on determinantal varieties.
\begin{lem}
	Let $R=\mathbb{C}[(x_{ij})]$ for $i=1\dots n$ and $j=1 \dots m$ and $I$ the ideal generated by the $r\times r$ minors of the matrix with entries $x_{ij}$. Then $R/I$ is a Cohen--Macaulay ring. Further $R/I$ is a Gorenstein ring if and only if $m=n$ or $r=1$.
\end{lem}
As we are dealing with $4\times 2$ matrices, $R/I$ is not Gorenstein; hence, we cannot expect to equip the multiplet with a  BV datum, but only with  a  BRST datum. However, $R/I$ is Cohen--Macaulay, which means that the dualizing complex is represented by a single sheaf. Thus, we can produce the corresponding antifield multiplet from that sheaf by applying the pure  spinor formalism. 
The dualizing module  is
\begin{equation}
\mathrm{Ext}^{\mathrm{codim}(Y)}_R (R/I, R) = \mathrm{Ext}^{3}_R (R/I, R) \: .
\end{equation}
Due to the Cohen--Macaulay property, this is the only non-vanishing Ext module. Its free resolution has the Betti numbers
\begin{center}
	\begin{tabular}{c|ccccc} 
		& $0$ & $1$ & $2$ & $3$ \\
		\hline
		$-4$ & $3$ & $8$ & $6$ & $-$ \\ 
		$-3$ & $-$ & $-$ & $-$ & $1$ \\
	\end{tabular}
	\captionof{table}{Betti numbers of the free resolution of $\mathrm{Ext}^{3}_R (R/I, R)$.}
\end{center}
Forming the direct sum of the structure sheaf and the dualizing sheaf and shifting appropriately, we obtain a multiplet with the following Betti numbers.
\begin{center}
	\begin{tabular}{c|ccccccc} 
		& $0$ & $1$ & $2$ & $3$ & $4$ & $5$ \\
		\hline
		$0$ & $1$ & $-$ & $-$ & $-$ & $-$ & $-$ \\ 
		$1$ & $-$ & $6$ & $8$ & $3$ & $-$ & $-$ \\
		$2$ & $-$ & $-$ & $3$ & $8$ & $6$ & $-$ \\ 
		$3$ & $-$ & $-$ & $-$ & $-$ & $-$ & $1$ \\
	\end{tabular}
	\captionof{table}{Betti numbers of the BV multiplet.}
\end{center}
This is the expected field content of the BV description for the six-dimensional gauge multiplet. The component multiplet can be equipped with a BV datum by writing the usual action as known from the component formalism. The resulting BV multiplet is off-shell; in fact, it is constructed as the cotangent theory of the corresponding BRST theory. The supersymmetry algebra closes without use of the equations of motion and the antifields can be separated from the fields. Doing this, one recovers the BRST multiplet in components.

One could also consider equipping the pure spinor superfield multiplet with a BRST datum. This was done in~\cite{Ced-6d}, where Cederwall considered a differential operator mapping pure spinor superfields for the structure sheaf to pure spinor superfields for the canonical module. This operator allows one to write a quadratic action functional for the structure sheaf multiplet, which defines a BRST datum for that multiplet. 

\subsection{Failure to be Cohen--Macaulay: the example of four-dimensional $\N=1$}
\label{ssec: failure}
As we have seen above, the pure spinor superfield formalism applied to the structure sheaf of the $d=4$, $\cN = 1$ nilpotence variety yields the BRST description of the gauge multiplet. The absence of antifields and BV differential is not particularly surprising: The failure of the supersymmetry action to be strict solely comes from gauge transformations; the equations of motions do not need to be imposed. Nevertheless one can ask if and how the corresponding antifield multiplet can be realized independently in the pure spinor superfield formalism. For this purpose, let us compute the dualizing complex of $R/I$. A model for the dualizing complex is given by
\begin{equation}
  \omega^\bullet_{R/I} = \mathrm{RHom}^\bullet_R(R/I,R) = \Hom_R(L^\bu, R).
\end{equation}
To compute this complex explicitly, we can use the minimal free resolution $L^\bullet\rightarrow R/I$ from above. By definition, the differential of the dualizing complex is
the dual map $d^\vee_L$ of the resolution differential $d_L$. In terms of matrices this means that $d^\vee_L$ is represented by the transposed matrices of~\eqref{eq: res diff}. From these matrices, the cohomology can be computed explicitly. We find that
\begin{equation}
  H^i(\omega_{R/I}^\bu) = 
  \begin{cases}
    \mathrm{coker}\left( (\lambda_1,-\lambda_2,-\bar{\lambda}_1, \bar{\lambda}_2) \right) \cong \CC , & i = 3; \\
    \CC[\lambda_1,\lambda_2] \oplus \CC[\bar{\lambda}_1,\bar{\lambda}_2] , & i = 2; \\
    0, & \text{otherwise}.
  \end{cases}
\end{equation}
Note that the codimension of $Y$ is two. If the dualizing complex were to resolve a single module, then $H^\bu(\omega_{R/I})$ should be concentrated in degree two. Instead, we find a copy of two \emph{disjoint} $\C^2$'s; $Y$ itself, of course, consists of two $\C^2$'s intersecting at the origin. This discrepancy is accounted for homologically by the presence of a single copy of the skyscraper sheaf (functions on $0 \in \C^2$) in degree three. 
At the end of the day, this means that we cannot find a single (non-dg) dualizing module for $R/I$ to feed into the pure spinor superfield formalism to obtain the antifield multiplet. This phenomenon will occur whenever $R/I$ is not a Cohen--Macaulay ring.

\subsection{A partial dictionary}
\label{ssec: dictionary} 
In this section we summarize some features of  the  correspondence between algebrogeometric properties of $\O_Y$-modules and physical features of the corresponding multiplets. This dictionary is of course by no means complete, but it should serve to provide a quick overview.
\begin{itemize}[listparindent=\parindent,leftmargin=0pt,itemindent=\parindent,label={---}]
  \item \emph{$\Gamma = \cO(S')$ for some hyperplane $S' \subseteq Y$.}

    $S'$ is a complete intersection of linear equations. The resulting multiplet is an exterior algebra $\wedge^\bullet S'$, concentrated in homological degree zero. No differentials are transferred to the component field level. The representation of the supersymmetry algebra is strict. Examples include chiral superfields ($S' = S_{\pm}$), which always exist in dimension $0\pmod 4$, and free superfields ($S' = \{0\}$), in any dimension and with any amount of supersymmetry. We emphasize that the free superfield always corresponds to the canonical augmentation of the graded ring $R/I$. 
      \item \emph{$\Gamma = \cO_Y$ is a complete intersection of quadratic equations.}

        The Koszul homology is an exterior algebra generated by the elements $\lambda \gamma^\mu \theta$ in homological degree one. The resulting multiplet can be identified with the de Rham complex $\Omega^\bullet(\RR^d)$ on spacetime. The transferred differential acts as the de Rham differential on the component fields; as such, translations act homotopically trivially. Tensoring with a Lie $(d-3)$-algebra, one obtains the BV complex of higher Chern--Simons theory. Odd elements in the supersymmetry algebra act by zero. Examples include the structure sheaves of the three-dimensional $\cN=8$ and four-dimensional $\cN = 4$ supersymmetry algebras; see~\cite{CederwallBLG} and~\cite{CederwallMax4D}, respectively. This sheaf is used, together with another equivariant sheaf, in the construction of the pure spinor resolution of the Bagger--Lambert--Gustavsson model in~\cite{CederwallBLG}.
      \item \emph{$\Gamma = \cO_Y$ is a Gorenstein ring, but not a complete intersection.}

        The resulting component multiplet is equipped with a local pairing, inherited from the perfect pairing on Koszul homology.
        For appropriate values of the spacetime dimension and the codimension of~$Y$, this local pairing defines an odd symplectic structure,  which may be used to construct a  BV datum on the multiplet. The underlying cochain complex always starts with
	\begin{equation}
	\Omega^0 \xlongrightarrow{d} \Omega^1 \longrightarrow \dots;
	\end{equation}
        as such, it always contains at least a one gauge field. By duality, the multiplet ends with the corresponding antifields,
	\begin{equation}
          \dots \longrightarrow \Omega^{d-1} \xlongrightarrow{d} \Omega^d .
	\end{equation}
        Examples include ten-dimensional super Yang--Mills theory and eleven-dimensional supergravity~\cite{Ced-towards,Ced-11d}; see also~\cite{spinortwist,MaxTwist} for treatments using a language close to this work.
        \item \emph{$\Gamma$ is Cohen--Macaulay, but not Gorenstein.}

          The resulting multiplet will not carry a pairing and thus cannot be equipped with a nondegenerate BV datum. We can interpret the multiplet as a BRST multiplet and obtain the corresponding antifield multiplet from the dualizing module. Here, the structure sheaf of six-dimensional $\cN = (1,0)$ supersymmetry is an example. To understand theories of physical interest, though, it may be necessary to consider degenerate pairings (six-dimensional $\N=(2,0)$ supersymmetry and type IIB  supergravity are examples).
        \item \emph{$\Gamma$ is not Cohen--Macaulay.}

          The resulting multiplet usually looks like a BRST multiplet. However, there is really only a dualizing complex instead of a dualizing module. As such, we cannot obtain the antifield multiplet from a single (non-dg) $\O_Y$-module via the pure spinor superfield formalism. An example is the gauge multiplet in four-dimensional $\N=1$ supersymmetry, as discussed above. It would be interesting to consider extending the formalism to dg sheaves on~$Y$.
        \item \emph{$\Gamma$ is a Golod ring.}
        
          A ring is Golod if and only if all Massey products on Koszul homology vanish~\cite{Frankhuizen}. Recall that, if $\Gamma$ is assumed to be a ring, the tensor product $A^\bu(\Gamma) \otimes \fh$ carries an $L_\infty$ structure. Transferring the $L_\infty$ structure to the component fields and then compactifying to a point yields an $L_\infty$ structure which is given by the $A_\infty$ structure on Koszul homology tensored with the Lie algebra $\fh$. The Golod property of $\Gamma$ implies that this $L_\infty$ structure is strict. For example, the presence of the three-ary product in ten-dimensional super Yang--Mills theory, which after compactification to a point gives rise to a corresponding product in the IKKT matrix model~\cite{MovshevMaxSYM}, witnesses the fact that the ring of functions of the ten-dimensional $\cN=1$ nilpotence variety is not Golod.
        
\end{itemize}

\section{Ten-dimensional super Yang--Mills theory}
\label{sec: 10d}
In this section, let us give a detailed analysis of ten-dimensional super Yang--Mills theory in the pure spinor superfield formalism. This is the initial example which sparked interest in the formalism~\cite{BerkovitsSuperparticle, CederwallM5}. As we will see, the multiplet obtained from the structure sheaf $\cO_Y$ can be naturally equipped with the full structure of a perturbative interacting BV theory within the pure spinor superfield formalism. 

\subsection{Field content and representatives}
We will denote the two 16-dimensional spin representations of $\mathrm{Spin}(10)$ by $S_+$ and $S_-$. The vector representation is, as always, denoted by $V$. The defining ideal of the nilpotence variety simply reads
\begin{equation}
I = (\lambda \gamma^\mu \lambda) \: .
\end{equation}
One finds for the minimal free resolution of $R/I$ the following Betti numbers.
\begin{center}
	\begin{tabular}{c|cccccc} 
		& $0$ & $1$ & $2$ & $3$ & $4$ & $5$ \\
		\hline
		$0$ & $1$ & $-$ & $-$ & $-$ & $-$  $-$ \\ 
		$1$ & $-$ & $10$ & $16$ & $-$ &$-$ & $-$ \\
		$2$ & $-$ & $-$ & $-$ & $16$ & $10$ & $-$ \\
		$3$ & $-$ & $-$ & $-$ & $-$ & $-$ & $1$ \\
	\end{tabular}
	\captionof{table}{Betti numbers for the ten-dimensional super Yang--Mills multiplet.}
\end{center}
More concretely, the minimal free resolution of $R/I$ in $R$-modules takes the form
\begin{equation}
L^\bullet = R \otimes \left( \CC \xlongleftarrow{(d_L)_1} V \xlongleftarrow{(d_L)_2} S_+ \xlongleftarrow{(d_L)_3} S_- \xlongleftarrow{(d_L)_4} V \xlongleftarrow{(d_L)_5} \CC \right) \: .
\end{equation}
The resolution differential can be described explicitly. Let us choose a basis $e_\mu$ of $V$ and $s_\alpha$ of $S_+$. The corresponding dual basis of $(S_+)^\vee = S_-$ is denoted by $s^\alpha$.
\begin{equation}
\begin{matrix}
(d_L)_1 &: & V & \longrightarrow & \CC &,& A & \mapsto & (\lambda \gamma^\mu \lambda) A_\mu \\
(d_L)_2 &: & S_+ & \longrightarrow & V &,& \chi & \mapsto & (\lambda \gamma^\mu \chi)  e_\mu \\
(d_L)_3 &: & S_- & \longrightarrow & S_+ &,& \chi^+ & \mapsto & (\lambda \gamma^\mu \lambda) (\chi^+ \gamma_\mu s) - 2(\chi^+ \lambda) (\lambda s) \\
(d_L)_4 &: & V & \longrightarrow & S_- &,& A^+ & \mapsto & (\lambda \gamma^\mu s) A^+_\mu \\
(d_L)_4 &: & \CC & \longrightarrow & V &,& c^+ & \mapsto & (\lambda \gamma^\mu \lambda) c^+ e_\mu \\
\end{matrix}
\end{equation}
We can perform our procedure to find the representatives. For example, starting with the gauge field,
\begin{equation}
A \xmapsto{(d_L)_1} (\lambda \gamma^\mu \lambda) A_\mu \xmapsto{\mathcal{D}_0^\dagger} (\lambda \gamma^\mu \theta) A_\mu \: ,
\end{equation}
so that the elements $(\lambda \gamma^\mu \theta) A_\mu$ represent the one-form in $\cD_0$-cohomology. For the gaugino we obtain
\begin{equation}
\chi \xmapsto{(d_L)_2} (\gamma^\mu \lambda)_\alpha \chi^\alpha e_\mu \xmapsto{\mathcal{D}_0^\dagger} (\gamma^\mu \theta)_\alpha \chi^\alpha e_\mu \xmapsto{(d_L)_1} (\lambda \gamma_\mu \lambda)(\gamma^\mu \theta)_\alpha \chi^\alpha \xmapsto{\mathcal{D}_0^\dagger} (\lambda \gamma_\mu \theta)(\gamma^\mu \theta)_\alpha \chi^\alpha \: .
\end{equation}
This means that the gaugino is represented by $(\lambda \gamma_\mu \theta)(\gamma^\mu \theta)_\alpha \chi^\alpha$ in $\cD_0$-cohomology. This procedure can also be applied to the antifields
\begin{equation}
\chi^+ \xmapsto{(d_L)_3 \circ \mathcal{D}_0^\dagger} (\lambda\gamma^\mu \theta)(\chi^+ \gamma_\mu s) \xmapsto{(d_L)_2 \circ \mathcal{D}_0^\dagger} (\lambda \gamma^\mu \theta) (\gamma^\nu \theta)_\alpha (\gamma_\mu \chi^+)^\alpha e_\nu \xmapsto{(d_L)_1 \circ \mathcal{D}_0^\dagger} (\lambda \gamma^\mu \theta) (\lambda \gamma^\nu \theta) (\gamma_\nu \theta)^\alpha (\gamma_\mu \chi^+)_\alpha
\end{equation}
We can simplify the last term to find
\begin{equation}
(\lambda \gamma^\mu \theta) (\lambda \gamma^\nu \theta) (\gamma_\nu \theta)^\alpha (\gamma_\mu \chi^+)_\alpha = (\lambda \gamma^\mu \theta) (\lambda \gamma^\nu \theta) (\gamma_{\mu \nu} \theta)^\alpha \chi^+_\alpha \: ,
\end{equation}
where $\gamma_{\mu \nu} = \gamma_{[\mu} \gamma_{\nu ]}$ denotes the antisymmetrized product of two gamma matrices. Similarly one can track down a representative for the antifield of the one-form field. The result is
\begin{equation}
(\lambda\gamma^\rho \theta) (\lambda \gamma^\nu \theta) (\theta \gamma_{\mu\nu\rho} \theta) A^{+\mu} \: .
\end{equation}
Finally, the antighost can be represented by
\begin{equation}
(\lambda\gamma^\mu \theta)(\lambda\gamma^\nu \theta) (\lambda \gamma^\rho \theta)(\theta \gamma_{\mu\nu\rho} \theta)c^+ \: .
\end{equation}
These representatives were already listed in~\cite{MovshevInvariants}. Let us summarize the results in the following table.
\begin{table}[h]
	\caption{Representatives for the $10D$ $\cN=1$ vector multiplet organized by $\theta$-degree.}
	\begin{center}
		\begin{tabular}{|c|c|}
			\hline
			Field & Representative in the $\cD_0$-cohomology \\
			\hline
			$c$  & $c$ \\
			\hline
			$A$ &$(\lambda \gamma^\mu \theta ) A_\mu$ \\
			\hline
			$\chi$ &$(\lambda \gamma_\mu \theta)(\chi \gamma^\mu \theta)$ \\
			\hline
			$\chi^+$ & $(\lambda \gamma^\mu \theta) (\lambda \gamma^\nu \theta) (\gamma_{\mu \nu} \theta \chi^+)$ \\
			\hline
			$A^+$ & $(\lambda\gamma^\rho \theta) (\lambda \gamma^\nu \theta) (\theta \gamma_{\mu\nu\rho} \theta) A^{+\mu}$ \\
			\hline
			$c^+$ & $(\lambda\gamma^\mu \theta)(\lambda\gamma^\nu \theta) (\lambda \gamma^\rho \theta)(\theta \gamma_{\mu\nu\rho} \theta)c^+$ \\
			\hline
		\end{tabular}
	\end{center}
	\label{tab:10d reps}
\end{table}

\subsection{The differential}
The first order part of the transferred differential is given by
\begin{equation}
\cD'_1 = p \circ (\lambda \gamma^\mu \theta) \partial_\mu \circ i\: .
\end{equation}
We immediately see that $\cD'_1$ acts on the ghost as the de Rham differential. 

Furthermore, the differential, $\cD'_1$ acts on the gaugino as the Dirac operator,
\begin{equation}
\chi \mapsto \slashed{\partial} \chi
\end{equation}
encoding the field equation for the gaugino.

Interestingly, this multiplet contains a second order contribution to the differential arising from homotopy transfer. As we will see momentarily, this encodes the equation of motion for the gauge field. Recall that the second order contribution to the transferred differential $\cD'$ is given by
\begin{equation}
\cD'_2 = p \circ (\cD_1 \circ h \circ \cD_1) \circ i \: .
\end{equation}
To apply $\cD'_2$ to the gauge field we need to know how the homotopy $h$ acts on expressions of the form
\begin{equation}
(\lambda \gamma^\mu \theta) (\lambda \gamma^\nu \theta) \: .
\end{equation}
Note that the naive guess $h \circ \pi = \pi \circ \cD_0^\dagger$ does not work in this case since
\begin{equation}
\cD_0^\dagger(\lambda \gamma^\mu \theta) (\lambda \gamma^\nu \theta) = 0
\end{equation}
by the symmetry of the bracket. However, the result is easily found by a representation theoretic argument. As $h$ acts as a scalar, we are looking for a representative inside
\begin{equation}
\wedge^2 V \subset \wedge^3 S_+ \otimes S_+ \: .
\end{equation}
It is easy to check that there is only one such summand in the decomposition of the right hand-side into irreducibles. This representation is spanned by the elements
\begin{equation}
(\lambda \gamma^\rho \theta) (\theta \gamma_{\mu\nu\rho} \theta) \: . 
\end{equation}
We set,
\begin{equation}
h\left( (\lambda \gamma^\mu \theta) (\lambda \gamma^\nu \theta) \right) = (\lambda \gamma^\rho \theta) (\theta \gamma_{\mu\nu\rho} \theta) \: .
\end{equation}
Equipped with this knowledge we find
\begin{equation}
\begin{split}
\cD'_2 \left( (\lambda \gamma^\mu \theta) A_\mu \right) =& p\left( \cD_1 \circ h \left( (\lambda \gamma^\mu \theta) (\lambda \gamma^\nu \theta) (dA)_{\mu \nu} \right) \right) \\
=& p\left( (\lambda \gamma^\sigma \theta) (\lambda \gamma_\rho \theta) (\theta \gamma^{\mu\nu\rho} \theta) \partial_\sigma (dA)_{\mu \nu} \right) \: .
\end{split}
\end{equation}
Projection to the cohomology gives
\begin{equation}
(\lambda \gamma_\nu \theta) (\lambda \gamma_\rho \theta) (\theta \gamma^{\mu\nu\rho} \theta) \partial^\sigma (dA)_{\sigma \mu}
\end{equation}
This shows that the transferred differential $\cD'_2$ acts via
\begin{equation}
A \mapsto \star d \star d A \: .
\end{equation}

The differentials appearing in the multiplet can be summarized by the following diagram.
\begin{equation}
\begin{tikzcd}[row sep=0.7cm, column sep=0.7cm]
\Omega^0(\RR^{10}) \arrow[dr, "d"] \\
& \Omega^1(\RR^{10}) \arrow[drrr, "\star d \star d" ' near start]  & \Gamma(\RR^{10}, S_+) \arrow[dr, crossing over, "\slashed{\partial}" ] &  \\
&  &  & \Gamma(\RR^{10}, S_-) & \Omega^{1}(\RR^{10}) \arrow[dr, "\star d \star"] \\
&  &  &  & & \Omega^0(\RR^{10}) \\
\end{tikzcd}
\end{equation}

\subsection{The supersymmetry action}
We can read off the non-derivative supersymmetry transformations directly from the resolution differential.
\begin{equation}
\begin{split}
\delta c =& (\epsilon \gamma^\mu \epsilon) A_\mu \\
\delta A_\mu =& \epsilon \gamma_\mu \chi \\
\delta \chi =& (\epsilon \gamma^\mu \epsilon) \chi^+ \gamma_\mu - 2 \epsilon (\chi^+ \epsilon) \\
\delta \chi^+ =& \epsilon \gamma^\mu A^+_\mu \\
\delta A_\mu^+ =& (\epsilon \gamma_\mu \epsilon) c^+ \\
\delta c^+ =&  0
\end{split}
\end{equation}
Note that there are two types of closure terms present. For the gauge field, there are again transformation witnessing that the supersymmetry algebra is represented only up to gauge transformations. We already encountered this type of transformation in our discussion of the four-dimensional gauge multiplet. In addition, there are now second order transformations for the gaugino, signaling that the supersymmetry algebra is represented only on-shell.

\subsection{The $L_\infty$ structure}
We can define a dg Lie algebra structure by tensoring $A^\bu(\cO_Y)$ with a Lie algebra $\fh$. Homotopy transfer gives rise to an $L_\infty$ structure on the component field multiplet. As we will see, this $L_\infty$ structure, together with the pairing, equips the ten-dimensional super Yang--Mills multiplet with the usual structure as an interacting BV theory.

The binary bracket $\mu'_2$ is given by
\begin{equation}
\begin{tikzpicture}
\begin{feynman}
\vertex at (-2,0) {$\mu'_2 \ = $};
\vertex(a) at (0,0);
\vertex(b) at (-1,1) {$i'$};
\vertex(c) at (-1,-1) {$i'$};
\vertex(d) at (1,0) {$p'$ \: .};
\diagram* {(a)--(b), (a)--(c), (a)--(d)};
\end{feynman}
\end{tikzpicture}
\end{equation}
Expressing this in terms of the unprimed homotopy data, there will be obviously a diagram of the form
\begin{equation}
\begin{tikzpicture}
\begin{feynman}
\vertex(a) at (0,0);
\vertex(b) at (-1,1) {$i$};
\vertex(c) at (-1,-1) {$i$};
\vertex(d) at (1,0) {$p$ \: .};
\diagram* {(a)--(b), (a)--(c), (a)--(d)};
\end{feynman}
\end{tikzpicture}
\end{equation}
As we already explored in the case of the four-dimensional $\cN=1$ vector multiplet, this diagram encodes the structure of gauge transformations on the component fields. In particular it yields brackets
\begin{equation}
\begin{matrix}
	\mu'_2 :& \Omega^0 \times \Omega^0 & \longrightarrow & \Omega^0 & \mu'_2(c,c) = [c,c] \: \phantom{.} \\
	\mu'_2 :& \Omega^0 \times \Omega^1 & \longrightarrow & \Omega^1  & \mu'_2(c,A) = [c,A] \: \phantom{.}\\
	\mu'_2 :& \Omega^0 \times \Gamma(X,S_+) & \longrightarrow & \Gamma(X,S_+) & \mu'_2(c,\psi) = [c,\psi] \: .
\end{matrix}
\end{equation}

Furthermore, considering degree bounds, we see that only two more diagrams can contribute, namely
\begin{equation}
\begin{tikzpicture}
\begin{feynman}
\vertex(a) at (1,0);
\vertex(b1) at (-1.5,1) {$i$};
\vertex(b)[dot] at (-0.5,1){};
\vertex(c) at (-0.5,-1) {$i$};
\vertex(d) at (2,0) {$p$};
\diagram* {(b1)--(b)--[edge label = $h$](a), (a)--(c), (a)--(d)};
\end{feynman}
\end{tikzpicture}
\begin{tikzpicture}
\begin{feynman}
\vertex at (-2,0) {and};
\vertex(a) at (0,0);
\vertex(b) at (-1,1) {$i$};
\vertex(c) at (-1,-1) {$i$};
\vertex(d)[dot] at (1,0){};
\vertex(d1) at (2,0){$p$ \: .};
\diagram* {(a)--(b), (a)--(c), (a)--[edge label = $h$](d), (d1)--(d)};
\end{feynman}
\end{tikzpicture}
\end{equation}
Here we marked the unary vertices with a dot, signaling the application of $\cD_1$.

From the first type of diagram we obtain
\begin{equation}
\begin{split}
p\left( (\lambda \gamma^\sigma \theta) h((\lambda \gamma^\mu \theta)(\lambda \gamma^\nu \theta) ) \right) [A_\sigma , \partial_\mu A_\nu] = p\left( (\lambda \gamma^\sigma \theta) (\lambda \gamma_\rho \theta) (\theta \gamma^{\mu \nu \rho} \theta) \right) [A_\sigma , \partial_\mu A_\nu]
\end{split}
\end{equation}
Using the antisymmetry in $\mu$ and $\nu$ and projecting onto $\cD_0$-cohomology this gives
\begin{equation}
(\lambda \gamma_\nu \theta) (\lambda \gamma_\rho \theta) (\theta \gamma^{\mu \nu \rho} \theta) \: [A^\sigma, (dA)_{\mu \sigma}] \: .
\end{equation}

The second diagram gives a contribution of the form
\begin{equation}
p\left( (\lambda \gamma^\sigma \theta) \partial_\sigma h\left( (\lambda \gamma^\mu \theta) (\lambda \gamma^\nu \theta) \right) [A_\mu , A_\nu] \right) = p \left( (\lambda \gamma^\sigma \theta) (\lambda \gamma_\rho\theta) (\theta \gamma^{\mu \nu \rho} \theta) \right) \partial_\sigma [A_\mu, A_\nu]
\end{equation}
Projection to the cohomology gives
\begin{equation}
(\lambda \gamma_\nu \theta) (\lambda \gamma_\rho\theta) (\theta \gamma^{\mu \nu \rho} \theta) \: \partial^\sigma [A_\mu, A_\sigma] \: .
\end{equation}
Together this gives a transferred binary product
\begin{equation}
\mu'_2 : \Omega^1 \times \Omega^1 \longrightarrow \Omega^1 \qquad \mu'_2(A,A)_\mu = [A^\sigma, (dA)_{\mu \sigma}] + \partial^\sigma [A_\mu,A_\sigma] \: .
\end{equation}
By degree reasons, there are no $\cD_1$ insertions allowed for $\mu'_3$. Hence the only contributing diagram is of the form 
\begin{equation}
\begin{tikzpicture}
\begin{feynman}
\vertex at (-2,0) {$\mu'_3 \ = $};
\vertex(a) at (-1,1) {$i$};
\vertex(b) at (-1,0) {$i$};
\vertex(c) at (-1,-1) {$i$};
\vertex(d) at (0,0.5);
\vertex(e) at (1,0);
\vertex(f) at (2,0) {$p$};
\diagram* {(a)--(d), (b)--(d), (d)--[edge label = $h$](e), (c)--(e), (f)--(e)};
\end{feynman}
\end{tikzpicture}
\end{equation}
This diagram gives a contribution of the form
\begin{equation}
p \left( (\lambda \gamma_\rho \theta)(\theta \gamma^{\mu \nu \rho} \theta) (\lambda \gamma^\sigma \theta) [A_\sigma , [A_\mu, A_\nu]] \right)
= (\lambda \gamma_\rho \theta) (\lambda \gamma_\nu \theta) (\theta \gamma^{\mu \nu \rho} \theta) \: [A^\sigma, [A_\mu, A_\sigma]] \:  .
\end{equation}
This gives a product
\begin{equation}
\mu'_3 : \Omega^1 \times \Omega^1 \times \Omega^1 \longrightarrow \Omega^1 \qquad \mu'_3(A,A,A)_\mu = [A^\sigma, [A_\mu, A_\sigma]] \: .
\end{equation}
Thus, we see that the transferred $L_\infty$ structure equips the multiplet with the usual interactions as expected for ten-dimensional super Yang--Mills theory.

\subsection{The pairing}

The ring $R/I$ is Gorenstein, which implies that the minimal free resolution, and hence the component field formulation of the multiplet is equipped with a local (in the sense of Definition \ref{dfn:BVdatum}) pairing. 
At the level of Koszul homology, the pairing is induced by multiplication and projection to the subspace spanned by the top class
\begin{equation}
(\lambda\gamma^\mu \theta)(\lambda\gamma^\nu \theta) (\lambda \gamma^\rho \theta)(\theta \gamma_{\mu\nu\rho} \theta) \: .
\end{equation}
This equips the component field multiplet with a BV structure.

We thus obtained the usual description of ten-dimensional super Yang--Mills theory as an interacting BV theory solely by homotopy transfer from the pure spinor superfield description.

\section{A bestiary of multiplets from modules}

In this final section we construct a variety of equivariant $R/I$-modules and examine the structure of the associated supersymmetric multiplets. We offer some observations connecting certain of these multiplets to constructions in the physics literature, along with some other speculations of various kinds.

\subsection{Presentations of modules and shift symmetry}
Any module $\Gamma$ over any ring $S$ can be described using a free presentation, that is an exact sequence
\begin{equation}
	0 \xlongleftarrow{} \Gamma \xlongleftarrow{} F_0 \xlongleftarrow{\varphi} F_1  \: ,
\end{equation}
where $F_0$ and $F_1$ are free $S$-modules. The module can then be identified as the cokernel of the map $\varphi$
\begin{equation}
\label{eq:shift}
\Gamma \cong \coker \varphi = F_0 / \mathrm{Im}(\varphi) \: .
\end{equation}
As $F_0$ and $F_1$ are free, we can think of $\varphi$ as a matrix with entries in $S$, these entries give the relations to obtain $\Gamma$ as a quotient from $F_0$.
In fact a free presentation is just the start of a free resolution. By resolving kernels we can extend a free presentation to a free resolution
\begin{equation}
	0 \xlongleftarrow{} \Gamma \xlongleftarrow{} F_0 \xlongleftarrow{\varphi_0} F_1  \xlongleftarrow{\varphi_1} F_2 \xlongleftarrow{\varphi_2} \dots \: .
\end{equation}
For $R=\CC[\lambda_1, \dots, \lambda_n]$ it is very easy to study such maps $\varphi$; these just correspond to matrices whose entries are polynomials in $\lambda$. The cokernels of such maps are then $R$-modules. For the pure spinor superfield formalism, it is crucial to use $R/I$-modules as this ensures that the differential $\cD$ squares to zero. Suppose we have an $R$-module defined by a free presentation
\begin{equation}
	\varphi : R^n \longrightarrow R^k \qquad \Gamma = \mathrm{coker}(\varphi) \: .
\end{equation} 
The $R$-module $\Gamma$ descends to a $R/I$-module if the image of $\varphi$ contains $I^k$, that is if the following diagram commutes. Thus we can conveniently construct $R/I$-modules by studying suitable maps between free $R$-modules. If the map $\varphi$ is also equivariant with respect to the action of the Lorentz group on $R$, then the resulting module is also equivariant. Hence, such equivariant maps between free $R$-modules precisely give rise to the desired input for the pure spinor superfield formalism. In the physics literature this procedure was used to construct multiplets in the pure spinor superfield formalism under the name shift symmetry.\footnote{The name ``shift symmetry'' arises from writing out the equivalence relation \eqref{eq:shift} as
$$f_0 \approx f_0 + \varphi(f_1) $$ with explicit representatives $f_0 \in F_0$ and $f_1 \in F_1$.}
See~\cite{Ced-towards, Ced-11d, Cederwall:2011vy, Cederwall}.

We can immediately give a free presentation for the quotient rings $R/I$ which we previously considered. The map
\begin{equation}
	\varphi: R^d \longrightarrow R, \qquad \varphi = \begin{pmatrix}
	\lambda \gamma^0 \lambda & \dots & \lambda \gamma^{d-1} \lambda \: 
	\end{pmatrix}.
\end{equation}
realizes the free presentation $\mathrm{coker}(\varphi) = R/I$.

\subsection{Motivating example of a nontrivial sheaf: the six-dimensional hypermultiplet}

As an example to demonstrate this technique, let us construct the six-dimensional hypermultiplet. We already constructed the six-dimensional vector multiplet from the structure sheaf of the nilpotence variety in~\S\ref{sec: 6d vector}. Recall that for six-dimensional $\cN=(1,0)$ supersymmetry, the odd part of the supertranslation algebra is
\begin{equation}
	S_+ \otimes U \: ,
\end{equation}
where $S_+$ is the fundamental representation of $\mathfrak{su}(4)$ and $U \cong \CC^2$ carries the fundamental representation of $\mathfrak{su}(2)$. The polynomial ring $R$ is nothing but the symmetric algebra on $S_+ \otimes U$ and comes with the natural action of $\mathfrak{su}(4) \times \mathfrak{su}(2)$. There is a unique equivariant map
\begin{equation}
	S_+ \otimes R \longrightarrow U \otimes R
\end{equation}
which is linear in $\lambda$. Choosing a basis for $U$ and $S_+$, this map is represented by
\begin{equation}
\varphi: S_+ \otimes R \longrightarrow U \otimes R \qquad \varphi = 
	\begin{pmatrix}
	\lambda^1_1 & \lambda^2_1 & \lambda^3_1 & \lambda^4_1 \\
	\lambda^1_2 & \lambda^2_2 & \lambda^3_2 & \lambda^4_2
	\end{pmatrix}.
\end{equation}
It is easy to check that the image of $\varphi$ indeed contains $I^2$, thus we can consider $\Gamma = \mathrm{coker}(\varphi)$ as an equivariant $R/I$-module in the pure spinor superfield formalism.

We display the Betti numbers of the minimal free resolution in Table~\ref{t:hyper}.
\begin{center}
	\begin{tabular}{c|ccccc} 
		& $0$ & $1$ & $2$ & $3$ \\
		\hline
		$0$ & $2$ & $4$ & $-$ & $-$ \\ 
		$1$ & $-$ & $-$ & $4$ & $2$ \\
	\end{tabular}
	\captionof{table}{Betti numbers for the six-dimensional hypermultiplet} \label{t:hyper}
\end{center}

The representations appearing in the minimal free resolution can be computed using \textit{Macaulay2} via the highest weight package. The minimal free resolution of $\Gamma$ in $R$-modules takes the form
\begin{equation}
	L^\bullet = R \otimes \left( U \xlongleftarrow{\varphi} S_{+} \xlongleftarrow{\epsilon} \wedge^3 S_{+}   \xlongleftarrow{\varphi^{T}} U \otimes \wedge^4 S_{+}  \right) \: ,
\end{equation}
and is a special case of the Buchsbaum--Rim complex \cite{Buchsbaum} (see~\cite[Appendix A.2.6]{MR1322960}  for a textbook presentation and a description of the differential $\epsilon$ in terms of the
$2\times2$ minors of $\varphi$).

Choosing a basis $e^i$ for $U$ and $s_{\alpha}$ a basis for $S_{+}$, we can write out the differentials in the complex as:
\begin{equation}
\label{eq:differentials6dhyper}
\begin{matrix}
d_1&:& \wedge^1 S_{+}  \longrightarrow  U  & \psi & \mapsto& \lambda_i^{\alpha} \psi_{\alpha} e^i \\
d_2 &:& \wedge^3 S_{+} \longrightarrow  \wedge^1 S_{+} & \psi^{+} &\mapsto& (\lambda^{\alpha}_{i} \lambda^{\beta}_{ j} \epsilon^{ij} )\psi_{\alpha \beta \gamma}^{+} s_{\gamma} \\
d_3 &:&  U \otimes \wedge^4 S_{+} \longrightarrow  \wedge^3 S_{+} & \phi^{+} &\mapsto& \lambda_{\alpha}^i\phi_i^{+} s^{\alpha}. \\
\end{matrix}
\end{equation}
In the last differential we identify $s^{\alpha}$ with $\epsilon^{\alpha \beta \gamma \delta} s_{\beta} \wedge s_{\gamma} \wedge s_{\delta}$ along the isomorphism $\wedge^3 S_{+} \cong S_-$.
The differential  $(\lambda^{\alpha}_{i} \lambda^{\beta}_{ j} \epsilon^{ij} )$ is the differential $\epsilon$ appearing in the Buchsbaum--Rim complex.

As expected, the hypermultiplet  consists of two scalars that form a doublet under $\mathfrak{su}(2)$ as well as fermions in $S_+$ that are neutral under $\mathfrak{su}(2)$ and their corresponding antifields. The two maps are expected to encode the respective equations of motions. We are thus dealing with an on-shell representation of the supersymmetry algebra. The multiplet can be equipped with a pairing which yields a BV structure.

We can use the zig-zag procedure to find representatives for the fields in the multiplet. These are expressed in terms of the basis $e_i$ of $U$.

\begin{table}[h]
	\caption{Representatives for the hypermultiplet in six dimensions organized by $\theta$-degree.}
	\begin{center}
		\begin{tabular}{|c|c|}
			\hline
			Field & Representative in the $\cD_0$-cohomology \\
			\hline
			$\phi$  & $\phi_i e^i$ \\
			\hline
			$\psi$ & $\psi_\alpha \theta^\alpha_i e^i$ \\
			\hline
			$\psi^+$ &$\lambda^\alpha_i \theta^\beta_j \theta^\gamma_k \varepsilon^{ij} \psi^+_{\alpha \beta  \gamma} e^l$ \\
			\hline
			$\phi^+$ & $\lambda^\alpha_i \theta^\beta_j \epsilon^{ij} \theta^\delta_l \theta^\gamma_k \varepsilon_{\alpha \beta \gamma \delta} \phi^{+ k} e^l$ \\
			\hline
		\end{tabular}
	\end{center}
	\label{tab:6d hyper reps}
\end{table}
From the resolution differential, we can easily read off the non-derivative supersymmetry transformations.
\begin{equation}
	\begin{split}
	\delta \phi_i &= \epsilon^\alpha_i \psi_\alpha \\
	\delta \psi_\alpha &= \epsilon^\beta_i \epsilon^\gamma_j \varepsilon^{ij} \psi^+_{\alpha \beta \gamma} \\
	\delta \psi^+_\alpha &= \epsilon_\alpha^i \phi^+_i \\
	\delta \phi^+_i &= 0
	\end{split}
\end{equation}
Again, we see the quadratic transformation involving the fermion and its antifield showing that the supersymmetry algebra only closes up to the equations of motion.

Consequently, the equations of motions are encoded in the transferred differential $\cD'$. There is a first order term $\cD'_1$ acting on the fermion. Given the representatives, it is easy to see that $\cD'_1$ acts by the Dirac operator
\begin{equation}
	\psi \mapsto \slashed{\partial} \psi \: .
\end{equation}
Further, there is a second order differential $\cD'_2$ induced via homotopy transfer which encodes the field equation of the scalar field and which acts via
\begin{equation}
	\cD'_2 = p\circ \left( \cD_1 \circ h \circ \cD_1\right) \circ i \: .
\end{equation}
Acting on the scalar, we find
\begin{equation}
	\cD'_2 \phi = p\left( \cD_1 \; h \; (\lambda_{[\alpha}^{[i} \theta_{\beta]}^{j]} \;  \partial^{[\alpha \beta]} \; \phi^i e_i)\right) \: .
\end{equation}
By degree reasons, applying the homotopy $h$ to the element in the brackets yields an expression in $\theta^2$. On purely representation theoretic grounds, we can see that there is a unique (up to a non-zero prefactor) expression which comes into question, namely
\begin{equation}
	\theta^{(i}_{[\alpha} \theta^{j)}_{\beta]} \; \partial^{[\alpha \beta]} \phi_{(i} e_{j)} \: .
\end{equation}
As a check, we may apply the differential $\cD_0$ to that representative. There we obtain
\begin{equation}
		\lambda^{(i}_{[\alpha} \theta^{j)}_{\beta]} \; \partial^{[\alpha \beta]} \phi_{(i} e_{j)} \: ,
\end{equation}
which, at first sight, does not look like the original element we started with. However, recall that we are working in the module $\Gamma$ which is the quotient $R^2/\mathrm{Im}(\varphi)$. In particular this means that $\lambda^i e_i = 0$ and hence
\begin{equation}
	0 = \lambda^i \theta^j e_i \phi_j = \lambda^{[i} \theta^{j]} e_{[i} \phi_{j]} + \lambda^{(i} \theta^{j)} e_{(i} \phi_{j)} \: ,
\end{equation}
such that we indeed get back our original expression (up to a non-zero prefactor). Moving on, we then easily find
\begin{equation}
	\cD'_2 \phi = (\lambda^\alpha_i \theta^\beta_j \epsilon^{ij} \theta^\delta_l \theta^\gamma_k \varepsilon_{\alpha \beta \gamma \delta}) \; \partial^\mu \partial_\mu \phi^{k} e^l \: ,
\end{equation}
such that the transferred differential indeed encodes the Laplace equation.

Summarizing, the multiplet has the following structure.
\begin{equation}
\begin{tikzcd}[row sep=0.7cm, column sep=0.7cm]
\Omega^0 \otimes U \arrow[drrr, "\star d \star d" ' near start] & \Omega^0 \otimes S_+ \arrow[dr,crossing over, "\slashed{\partial}"] \\
& & \Omega^0 \otimes S_- & \Omega^0 \otimes U  \\
\end{tikzcd}
\end{equation}

This multiplet was defined in Equation (3.2) of~\cite{Ced-6d} using shift symmetry.

\subsection{Lie algebra cohomology}
\label{ssec: CE}
Another natural source for equivariant modules are the Lie algebra cohomology groups of the super translation algebra $\ft$. This was already noted in~\cite{NV}. Recall that the Chevalley--Eilenberg complex takes the form
\begin{equation}
	\mathrm{C}^\bullet(\ft) = \left( \mathrm{Sym}^\bullet(\ft^\vee[1]) \; , \; \d_{CE} \right) \: .
\end{equation}
The Chevalley--Eilenberg differential is induced by the dual of the bracket, which is extended to the whole algebra according to the Leibniz rule. For the supertranslation algebra, the $\ZZ \times \ZZ/2$ grading of the Chevalley--Eilenberg complex lifts to a $\ZZ \times \ZZ$ grading by viewing the supertranslations as a graded Lie algebra as we have done above; recall that $V = \ft_+$ then sits in degree two and~$S = \ft_-$ in degree one.
If we totalize this bigrading, generators in~$V^\vee$ sit in degree $-1$ and  generators in~$S^\vee$ in degree zero. We can thus identify 
$\sym^q(S^\vee) = R = \CC[\lambda^\alpha]$, and write
\begin{equation}
	\mathrm{C}^{-p}(\ft) = \wedge^p(V^\vee) \otimes R  .
\end{equation}
Denoting a basis on $V^\vee$ by $v^\mu$, the Chevalley--Eilenberg differential acts on the generators by
\begin{equation}
\begin{split}
	\d_{CE} v^\mu &= \lambda^\alpha \Gamma^\mu_{\alpha \beta} \lambda^\beta \\
	\d_{CE} \lambda^\alpha &= 0 \: .
\end{split}
\end{equation}
Now two observations turn out to be crucial. First, the zeroth Chevalley--Eilenberg cohomology is nothing else then the structure sheaf of the nilpotence variety
\begin{equation}
	H^0(\mathrm{C}^\bu(\ft) ) = R/I = \cO_Y \: .
\end{equation}
Second, as the Chevalley--Eilenberg complex comes with the structure of a cdgsa, the cohomology is equipped with a multiplication which preserves the grading. Hence, all cohomology groups are $H^0(\mathrm{C}^\bu(\ft)) = R/I$-modules and can thus be used as input data for the pure spinor superfield formalism.

The analysis of examples suggests some speculations about dualities between the multiplets associated to Chevalley--Eilenberg cohomology groups in different degrees. For a start, it seems to be the case that the Chevalley--Eilenberg cohomology groups are concentrated in negative degrees up to $n := \dim(V)-\mathrm{codim}(Y)$. In all examples we have checked there is an isomorphism
\begin{equation}
	\mathrm{Ext}^{\mathrm{codim}(Y)} (R/I, R) \cong H^{-n} (\clie^\bu(\ft)) \: .
\end{equation}
In addition, for the example of ten-dimensional $\cN=1$ supersymmetry, we further observe dualities ``up to a copy of the free superfield'' for the multiplets associated to $H^i(\clie^\bu(\ft))$ and $H^{-n-i}(\clie^\bu(\ft))$.
\subsubsection*{Three-dimensional $\cN=1$}
As a motivating example let us consider again $\cN =1$ supersymmetry in three dimensions. Using \textit{Macaulay2} one can compute the Chevalley--Eilenberg cohomology. Only $H^0$ and $H^{-1}$ are non-vanishing. The zeroth cohomology is $R/I$ and thus gives rise to the gauge multiplet from~\S\ref{ssec: schemes}. As the length of the minimal free resolution is two---which equals the codimension of $Y$---we immediately see that $R/I$ is Cohen--Macaulay. The first cohomology group is represented as the cokernel of the map
\begin{equation}
	\varphi: R^3 \longrightarrow R^2 \qquad \varphi =
	\begin{pmatrix}
	\lambda_1 & 0 & \lambda_2 \\
	0 & \lambda_2 & \lambda_1
	\end{pmatrix}
\end{equation}
The resulting multiplet is the antifield multiplet of the gauge multiplet.
\begin{center}
	\begin{tabular}{c|ccccc} 
		& $0$ & $1$ & $2$ \\
		\hline
		$0$ & $2$ & $3$ & $-$ \\ 
		$1$ & $-$ & $-$ & $1$ \\
	\end{tabular}
	\captionof{table}{Betti numbers for the antifield multiplet.}
\end{center}
Note that, as discussed in~\S\ref{sec: data} we could have also obtained the antifield multiplet from $\mathrm{Ext}^2(R/I,R)$.

\subsubsection*{Four-dimensional $\cN =1$} The Chevalley--Eilenberg cohomology is concentrated in degrees zero, minus one and, minus two. As the zeroth cohomology is just $\cO_Y$, the corresponding multiplet is the gauge multiplet. The first cohomology group yields a multiplet with the following Betti numbers.
\begin{center}
	\begin{tabular}{c|ccccc} 
		& $0$ & $1$ & $2$ & $3$ & $4$\\
		\hline
		$0$ & $4$ & $7$ & $-$ \\ 
		$1$ & $-$ & $-$ & $6$ & $4$ & $1$
	\end{tabular}
	\captionof{table}{Betti numbers $H^{-1}(\clie^\bu(\ft))$.}
\end{center}
Decomposing the minimal free resolution equivariantly, we find
\begin{equation}
L^\bullet = R \otimes \left( S_+ \oplus S_- \xlongleftarrow{(d_L)_1} \wedge^2 V \oplus \C \xlongleftarrow{(d_L)_2} \wedge^3 V \oplus \C^2 \xlongleftarrow{(d_L)_3} S_+ \oplus S_- \xlongleftarrow{(d_L)_4} \C \right) \: .
\end{equation}
Thus we see that this multiplet contains a two-form. It would be interesting to interpret this as a field-strength multiplet.

The second Chevalley--Eilenberg cohomology yields two copies of the chiral multiplet.
\begin{samepage}
	\begin{center}
		\begin{tabular}{c|ccccc} 
			& $0$ & $1$ & $2$ \\
			\hline
			$0$ & $2$ & $4$ & $2$ \\ 
		\end{tabular}
		\captionof{table}{Betti numbers for $H^{-2}(\clie^\bu(\ft))$.}
	\end{center}
\end{samepage}
Note that this precisely matches with $\mathrm{Ext}^2(R/I,R)$ as described in~\S\ref{ssec: failure}.

\subsubsection*{Ten-dimensional $\cN=1$}
\label{ssec:tendstrans}
Let us further study the multiplets associated to the ten-dimensional Che\-val\-ley--Ei\-len\-berg cohomology of the ten-dimensional $\cN=1$ supertranslation algebra. These cohomology groups were already computed equivariantly in~\cite{MovshevSchwarzXu}
The multiplet associated to the first Chevalley--Eilenberg cohomology has the following Betti numbers.
\begin{center}
	\begin{tabular}{c|cccccccccccccccccccc} 
		& $0$ & $1$ & $2$ & $3$ & $4$ & $5$ & $6$ & $7$ & $8$ & $9$ & $10$ & $11$ & $12$ & $13$ & $14$ & $15$ & $16$  \\
		\hline
		$0$ & $16$ & $45$ & $-$ & $-$ & $-$ & $-$ & $-$ & $-$ & $-$ & $-$ & $-$ & $-$ & $-$ & $-$ & $-$ & $-$ & $-$  \\
		$1$ & $-$ & $16$ & $250$ & $720$ & $1874$ & $4368$ & $8008$ & $11440$ & $12870$ & $11440$ & $8008$ & $4368$ & $1820$ & $560$ & $120$ & $16$ & $1$ \\
		$2$ & $-$ & $-$ & $-$ & $-$ & $16$ & $10$ & $-$ & $-$ & $-$ & $-$ & $-$ & $-$ &$-$ & $-$ & $-$ & $-$ & $-$
	\end{tabular}
	\captionof{table}{Betti numbers for $H^{-1}(\clie^\bu(\ft))$.}
	\label{t: 10d H1}
\end{center}
We notice that the graded rank (with respect to the homological degree) of the associated vector bundle over spacetime---which, in physical terms, corresponds to the number of degrees of freedom---is given by
\begin{equation}
\begin{split}
	&-16 -45 \\
	&+ 16 + 250 + 720 + 1874 + 4368 + 8008 + 11440 + 12870 + 11440 + 8008 + 4368 + 1820 + 560 + 120 + 16 + 1 \\ 
	&-16 -10 \\
	&= 65792 = (2^{15} + 2^{15}) + (128 + 128) \: .
\end{split}
\end{equation}
This precisely matches the number of degrees of freedom of the supercurrent multiplet constructed in~\cite{crazylargesuperfield}. Further, recall that the free superfield just corresponds to the exterior algebra $\wedge^\bu S$ on 16 generators. Hence, its Betti numbers are precisely binomial coefficients $\binom{16}{i}$. We note that Table~\ref{t: 10d H1} contains precisely such coefficients, except for a missing $1$ in degree $(0,1)$. However, we can add a trivial pair in degrees $(0,1)$ and $(1,0)$. Then we can subtract the respective Betti numbers of the free superfield to obtain the following table.
\begin{center}
	\begin{tabular}{c|cccccccccccccccccccc} 
		& $0$ & $1$ & $2$ & $3$ & $4$ & $5$  \\
		\hline
		$0$ & $16$ & $45+1$ & $-$ & $-$ & $-$ & $-$  \\
		$1$ & $-$ & $-$ & $130$ & $160$ & $154$ & $-$ \\
		$2$ & $-$ & $-$ & $-$ & $-$ & $16$ & $10$
	\end{tabular}
	\captionof{table}{Subtracted Betti table.}
\end{center}
This is precisely the dual of the Betti table of $H^{-4}(\clie^\bu(\ft))$, which is displayed in Table~\ref{t:H4}. We remark that this ``almost-duality'' phenomenon is closely analogous to the structure sheaf of 4d $\N=1$; it reflects the failure of the module to be Cohen--Macaulay. We further note that the fields in the first row are a spinor, a two-form, and a scalar; it is  tempting to interpret this as a field-strength multiplet, containing the gaugino $\chi$ and the field strength $F$ of  the gauge  field, and subject to certain constraints.

The multiplet associated to $H^{-2}(\clie^\bu(\ft))$ is the stress-energy tensor multiplet or supercurrent multiplet.
Its Betti table is displayed in Table~\ref{t:H2}.
\begin{center}
	\begin{tabular}{c|cccccccccccccccccccc} 
		& $0$ & $1$ & $2$ & $3$ & $4$ & $5$ & $6$ & $7$ & $8$ & $9$ & $10$ & $11$ & $12$ & $13$ & $14$  \\
		\hline
		$0$ & $120$ & $720$ & $2130$ & $4512$ & $8008$ & $11440$ & $12870$ & $11440$ & $8008$ & $4368$ & $1820$ & $560$ & $120$ & $16$ & $1$  \\
		$1$ & $-$ & $-$ & $-$ & $136$ & $160$ & $45$ & $-$ & $-$ & $-$ & $-$ & $-$ & $-$ & $-$ & $-$ & $-$ \\
	\end{tabular}
	\captionof{table}{Betti numbers for $H^{-2}(\clie^\bu(\ft))$.}
	\label{t:H2}
\end{center}

The supercurrent multiplet can be constructed as
$$J_{\mu\nu\rho} = \tr \chi \gamma_{\mu\nu\rho} \chi.$$
Here $\gamma$ just represents the isomorphism $\wedge^2(S_+) \cong \wedge^3(V)$, and $\chi$ is the spinor superfield describing on-shell Yang-Mills theory \cite{Howe:1987ik} that corresponds to $H^1(\clie^\bu(\ft))$.  Alternatively it can be described as an abstract superfield satisfying the constraints \cite{Howe:1987ik}
\begin{align*}
D_{\alpha} J_{abc} = (\gamma_{[a} J^1_{bc]})_{\alpha} + (\gamma_{[ab} J^1_{c]})_{\alpha} + (\gamma_{[abc]} J^1)_{\alpha} 
\end{align*}
where the superfields $J^1_{bc \alpha}, J^1_{c \alpha},$ and $J^1_{\alpha}$ are three superfields in the representation $[0,1,0,1,0], [1,0,0,0,1]$
and $[0,0,0,1,0]$.  The total dimension of the constraints is $560 + 144 + 16 = 720.$ 
The leading component of $J_{abc}$ is in the $\wedge^3 V$ representation $[0,0,1,0,0]$ of dimension 120.

Again, introducing trivial pairs and subtracting precisely yields the dual of the Betti table of $H^{-3}(\clie^\bu(\ft))$, which we display in Table~\ref{t: H3}.

\begin{center}
	\begin{tabular}{c|cccccccccccccccccccc} 
		& $0$ & $1$ & $2$ & $3$ & $4$ & $5$ & $6$  \\
		\hline
		$0$ & $45$ & $160$ & $136$ & $-$ & $-$ & $-$ & $-$  \\
		$1$ & $-$ & $-$ & $144$ & $310$ & $160$ & $-$ & $-$ \\
		$2$ & $-$ & $-$ & $-$ & $-$ & $-$ & $16$ & $1$
	\end{tabular}
	\captionof{table}{Betti numbers for $H^{-3}(\clie^\bu(\ft))$.}
	\label{t: H3}
\end{center}

\begin{center}
	\begin{tabular}{c|cccccccccccccccccccc} 
		& $0$ & $1$ & $2$ & $3$ & $4$ & $5$  \\
		\hline
		$0$ & $10$ & $16$ & $-$ & $-$ & $-$ & $-$  \\
		$1$ & $-$ & $54$ & $160$ & $130$ & $-$ & $-$ \\
		$2$ & $-$ & $-$ & $-$ & $-$ & $46$ & $10$
	\end{tabular}
	\captionof{table}{Betti numbers for $H^{-4}(\clie^\bu(\ft))$.}
	\label{t:H4}
\end{center}
Finally, $H^{-5}(\clie^\bu(\ft)) \cong R/I$ again yields the vector multiplet. Note that $Y$ is Gorenstein and of codimension five, such that $\mathrm{Ext}^5(R/I,R) \cong R/I$.
\begin{center}
		\begin{tabular}{c|cccccc} 
		& $0$ & $1$ & $2$ & $3$ & $4$ & $5$ \\
		\hline
		$0$ & $1$ & $-$ & $-$ & $-$ & $-$  $-$ \\ 
		$1$ & $-$ & $10$ & $16$ & $-$ &$-$ & $-$ \\
		$2$ & $-$ & $-$ & $-$ & $16$ & $10$ & $-$ \\
		$3$ & $-$ & $-$ & $-$ & $-$ & $-$ & $1$ \\
	\end{tabular}
	\captionof{table}{Betti numbers for $H^{-5}(\clie^\bu(\ft))$.}
\end{center}

\subsection{Six-dimensional multiplets from line bundles}
Recall that the six-dimensional nilpotence variety can be identified with $\CC P^1 \times \CC P^3$ using the Segre embedding. Line bundles on $\CC P^n$ are classified by a single integer $j \in \ZZ$ and are denoted by $\cO(j)$. Using the projections
\begin{equation}
\begin{tikzcd}
\CC P^1 \times \CC P^3 \arrow[r, "\pi_3"] \arrow[d, "\pi_1"] & \CC P^3 \\
\CC P^1
\end{tikzcd}
\end{equation}
we can define a family of line bundles
\begin{equation}
	\cO(i,j) := \pi_3^*\cO(i) \otimes \pi_1^* \cO(j)
\end{equation}
on $\CC P^1 \times \CC P^3$.
This family has been investigated in the physics literature \cite{Kuzenko:2017zsw}.

Let us here list the corresponding multiplets for some integers $i$ and $j$.
Clearly $\cO(0,0)$ is just the structure sheaf of the nilpotence variety and hence the corresponding multiplet is the vector multiplet. $\cO(0,1)$ is the hypermultiplet, which we studied above. $\cO(0,2)$ is the antifield multiplet of the vector.

For $\cO(0,3)$ a multiplet with the following Betti numbers arises.
\begin{center}
	\begin{tabular}{c|ccccc} 
		& $0$ & $1$ & $2$ & $3$ \\
		\hline
		$0$ & $4$ & $12$ & $12$ & $4$ \\
	\end{tabular}
	\captionof{table}{Betti numbers for $\cO(0,3)$.}
\end{center}
The minimal free resolution of the module in $R$-modules takes the form
\begin{equation}
	L^\bullet = R \otimes \left( \CC^4 \xlongleftarrow{(d_L)_1} \CC^3 \otimes S_{+} \xlongleftarrow{(d_L)_2} \CC^2 \otimes \wedge^2 S_{+}   \xlongleftarrow{(d_L)_3} \CC^1 \otimes \wedge^3 S_{+}  \right) \: ,
\end{equation}

The multiplet for $\cO(0,4)$ is a building block in the construction of the ``relaxed hypermultiplet'' \cite{Howe:1982tm}. 
\begin{center}
	\begin{tabular}{c|ccccc} 
		& $0$ & $1$ & $2$ & $3$ & $4$\\
		\hline
		$0$ & $5$ & $16$ & $18$ & $8$ & $1$\\
	\end{tabular}
	\captionof{table}{Betti numbers for $\cO(0,4)$.}
\end{center}
The minimal free resolution of the module in $R$-modules takes the form
\begin{equation}
	L^\bullet = R \otimes \left( \CC^5 \xlongleftarrow{(d_L)_1} \CC^4 \otimes S_{+} \xlongleftarrow{(d_L)_2} \CC^3 \otimes \wedge^2 S_{+}   \xlongleftarrow{(d_L)_3} \CC^2 \otimes \wedge^3 S_{+}  \xlongleftarrow{(d_L)_4} \CC^1 \otimes \wedge^4 S_{+}  \right) \: ,
\end{equation}
The minimal free resolutions are ``twisted Lascoux'' complexes which are described with their differentials in \cite{MR3223878}.\footnote{The bestiary of multiplets from Lascoux complexes was partly inspired by the bestiary of fauna depicted in the Lascaux cave and the work of Tristan H\"ubsch.}

\subsection{Conormal modules}
Denoting the defining ideal of the nilpotence variety by $I$, the conormal module is defined as the quotient $I/I^2$. This gives another interesting module to consider as an input for the pure spinor superfield formalism. The resulting multiplets seem to often correspond to supergravity theories. We demonstrate this in low dimensions.
\subsubsection*{Three-dimensional $\cN =1$}
The resulting multiplet has the following Betti numbers.
\begin{center}
	\begin{tabular}{c|ccccc} 
		& $0$ & $1$ & $2$ \\
		\hline
		$2$ & $3$ & $2$ & $-$ \\ 
		$3$ & $-$ & $5$ & $4$ \\
	\end{tabular}
	\captionof{table}{Betti numbers for the conormal module in three-dimensional $\cN=1$.}
\end{center}
Investigating the Hilbert series, we find that all occuring representations are irreducible representations of the spin group $\mathrm{Spin}(3) \cong SU(2)$. Thus the first line contains a vector and a spinor, while the second line can be identified with a symmetric traceless tensor and the four-dimensional part of the decomposition
\begin{equation}
	S \otimes V \cong [1] \oplus [3] \: .
\end{equation}
\subsubsection*{Four-dimensional $\cN=1$}
In four dimensions the conormal module yields a multiplet with the following Betti numbers.
\begin{center}
	\begin{tabular}{c|ccccc} 
		& $0$ & $1$ & $2$ & $3$ \\
		\hline
		$2$ & $4$ & $4$ & $1$ & $-$ \\ 
		$3$ & $-$ & $9$ & $12$ & $4$ \\
	\end{tabular}
	\captionof{table}{Betti numbers for the conormal module in four-dimensional $\cN=1$.}
\end{center}
Investigating the Hilbert series we find that the representations in the first line are a vector, a Dirac spinor and a scalar. The nine-dimensional representation in the second line is a symmetric traceless tensor. The twelve dimensional representation has Dynkin labels $[2,1] \oplus [1,2]$. Thus the multiplet consists of one spin-2, two spin-$\frac{3}{2}$ and a single spin-1 field. In terms of Dynkin labels, the multiplet takes the following form.
\begin{equation}
\begin{tikzcd}[row sep=0.7cm, column sep=0.7cm]
\left[1,1 \right] \arrow[dr] & \left[1,0\right] \oplus \left[0,1\right] \arrow[dr] & \left[0,0\right] \arrow[dr] \\
 & \left[2,2\right] & \left[2,1\right] \oplus [1,2]  & \left[1,1\right]  \\
\end{tikzcd}
\end{equation}
This matches the field content of the massive spin-two multiplet in four-dimensional $\cN=1$ supersymmetry.
\subsubsection*{Ten-dimensional $\N=1$}
In this case, by a pleasing coincidence, the conormal module coincides with the module $H^{-4}(\clie^\bu(\ft))$ constructed above. The resolution was studied in~\cite[Corollary 4.4]{MR3833474}.

\subsection{Dimensional reduction and restriction to strata; the 4d $\N=2$ tensor multiplet}
There are interesting relations between the nilpotence varieties of supersymmetry algebras in different dimensions, for instance the nilpotence variety of a higher dimensional supersymmetry algebra may sit inside the nilpotence variety of a lower dimensional one. The resulting multiplets will then be related by dimensional reduction. We illustrate this by considering the relation between six-dimensional $\cN=(1,0)$ and four-dimensional $\cN=2$ supersymmetry. Recall that we described the nilpotence variety for six-dimensional $\cN=(1,0)$ supersymmetry by the $2\times 2$-minors of a $2 \times 4$-matrix with entries $\lambda^\alpha_i$. As explained in~\cite{NV} one obtains the nilpotence variety for four-dimensional $\cN=2$ supersymmetry by replacing
\begin{equation}
	\lambda^\alpha_i \longrightarrow (\lambda^\beta_i , \bar{\lambda}^{\dot{\beta}}_i) \: ,
\end{equation}
and throwing away the two minors which do not mix the different chiralities. Hence there is an inclusion
\begin{equation}
	i: Y(6;1,0) \hookrightarrow Y(4;2) \: ,
\end{equation}
whose image we denote by $Y_0$. In fact the global structure of $Y(4;2)$ is easily described. It consists of three strata; in addition to $Y_0$ there are two copies of $(S_{\pm} \otimes U) \cong \C^4$ corresponding to solutions where $\lambda = 0$ or $\bar{\lambda} = 0$ respectively:
\begin{equation}
	Y(4;2) = Y_0 \cup Y_1 \cup Y_2 \cong Y(6;1,0) \cup (S_+ \otimes U) \cup(S_- \otimes U) \: .
\end{equation}
Pushing forward the structure sheaf $\cO_{Y(6;1,0)}$ along $i$ we thus obtain $\cO_{Y_0}$. As we already discussed at multiple occasions, the structure sheaf $\cO_{Y(6;1,0)}$ produces the vector multiplet. Clearly, considering $\cO_{Y_0}$ in the pure spinor superfield formalism gives a multiplet with the same Betti numbers; only the weights have to be adapted to four dimensions. Resolving $\cO_{Y_0}$ equivariantly, we see that the six-dimensional vector splits up into a four-dimensional vector and two scalars. The fermion gives two Dirac fermions in four dimensions and the scalars remain scalars. Hence, the resulting multiplet is precisely the $\cN=2$ vector multiplet in four-dimensions as one can obtain it from dimensional reduction. A similar phenomenon holds in general: given a multiplet in dimension $d$, we can push the corresponding sheaf forward along the dimensional reduction map to obtain the dimensionally reduced multiplet.

Interestingly, considering $\cO_{Y(4;2)}$ as an input in the pure spinor superfield machinery gives a multiplet with the following Betti numbers.
\begin{center}
	\begin{tabular}{c|ccccc} 
		& $0$ & $1$ & $2$ & $3$ & $4$\\
		\hline
		$0$ & $1$ & $-$ & $-$ & $-$ & $-$\\ 
		$1$ & $-$ & $4$ & $-$ & $-$ & $-$\\
		$2$ & $-$ & $-$ & $9$ & $8$ & $2$
	\end{tabular}
	\captionof{table}{Betti numbers for the structure sheaf for four-dimensional $\cN=2$.}
\end{center}
Working equivariantly, the minimal free resolution gives
\begin{equation}
L^\bullet = R \otimes \left( \CC \xlongleftarrow{(d_L)_1} V \xlongleftarrow{(d_L)_2} \wedge^2 V \oplus \C^3   \xlongleftarrow{(d_L)_3} (S_+ \otimes U) \oplus (S_- \otimes U)  \xlongleftarrow{(d_L)_4} \C_{2} \oplus \C_{-2} \right) \: ,
\end{equation}
where $\C^3$ carries the adjoint representation of $SU(2)_R$ and has $U(1)_R$-charge $0$ while the two scalars in the top degree have $U(1)_R$-charges $+2$ and $-2$ as indicated by the subscript. This is the field content of a tensor multiplet as described in~\cite{deWitTensor,JurcoReview}.

Of course we can also restrict to the other strata. The minimal free resolutions are then exterior algebras $\wedge^\bu (S_\pm \otimes U)$, the resulting multiplets are thus chiral multiplets as described in~\cite{deWitChiral}.

\subsection{Outro}
There are, of course, many more constructions possible to obtain equivariant modules and all of these may be applied in the context of the pure spinor superfield formalism. For example it may be interesting to study tensor products, symmetric or exterior powers. Further, there are geometric constructions, such as tangent and cotangent sheaves, just to name a few. It would be particularly interesting to find physical interpretations for such constructions; one would also hope to better understand functorial properties and to develop the pure spinor formalism into an appropriate equivalence of categories,\footnote{We plan to return to such structural properties of pure spinor superfields in future work with C.~Elliott.} thus finally bringing order to the bestiary of supersymmetric multiplets.

\appendix

\section{Homotopy transfer for $L_\infty$ modules} \label{ap:modules}
Let $(L,\tilde{\mu}_k)$ be a (super) $L_\infty$ algebra and $(V,d_V,\rho^{(j)})$ an $L_\infty$ module for $L$. As was explained in~\cite{Lada}, the $L_\infty$ module structure gives rise to an $L_\infty$-structure on $L\oplus V$. Explicitly we can define (setting $\rho^{(0)} = d_V$)
\begin{equation} \label{L-mod sum}
\mu_k((x_1,v_1),\dots,(x_k,v_k)) = \left( \tilde{\mu}_k(x_1,\dots,x_k) , \sum_{i=1}^{k} \pm \rho^{(k-1)}(x_1,\dots,\hat{x}_i,\dots,x_k)v_i \right) \: .
\end{equation}
For example, if $(L,[.,.])$ is a (super) Lie algebra and $\rho$ is a strict action, we find
\begin{equation}
\begin{split}
\mu_1((x,v)) &= (0,d_V v) \\
\mu_2((x_1,v_1),(x_2,v_2) ) &= ([x_1,x_2] \: , \: \rho(x_1)v_2 - \rho(x_2)v_1 )
\end{split}
\end{equation}
All higher order operations vanish. Now suppose we have homotopy data
\begin{equation}
\begin{tikzcd}
\arrow[loop left]{l}{h}(V,d_V)\arrow[r, shift left, "p"] &(W, d_W)\arrow[l, shift left, "i"]
\end{tikzcd}
\end{equation}
and want to transfer an $L_\infty$ module structure on $V$ to a new $L_\infty$ module structure on $W$. The fact that these $L_\infty$ module structures can be thought of as $L_\infty$-structures on $L \oplus V$ and $L \oplus W$ suggests to extend the above homotopy data to
\begin{equation}
\begin{tikzcd}
\arrow[loop left]{l}{\mathrm{id} \oplus h}(L\oplus V,d_V)\arrow[r, shift left, "\mathrm{id} \oplus p"] &(L \oplus W, d_W)\arrow[l, shift left, "\mathrm{id} \oplus i"]
\end{tikzcd}
\end{equation}
and then to use the usual homotopy transfer for $L_\infty$-structures. Let us denote the transferred $L_\infty$-structure on $L\oplus W$ by $\mu'_k$. We can read off the transferred module action $\rho'^{(k)}$ as follows. Let
\begin{equation}
\pi : L \oplus V \longrightarrow V
\end{equation}
be the obvious projection. Then~(\ref{L-mod sum}) implies
\begin{equation} \label{rho and mu}
\rho'^{(k)}(x_1,\dots,x_k)w = \pi\left(\mu'_{k+1}((x_1,0),\dots,(x_k,0), (x_{k+1},w) ) \right) \: .
\end{equation}
As usual, the transferred $L_\infty$-structure $\mu'_k$ can be calculated by sum over trees formulas. Using this, one can also derive sum over tree formulas for the induced action $\rho'$. For our purposes we are only interested in the case where $L = \mathfrak{g}$ is a (super) Lie algebra and $\rho$ is a strict action. As explained above, this means that $(\mathfrak{g} \oplus V, \mu_k)$ is a dg-Lie algebra. In this case the $L_\infty$-structure on $\mathfrak{g}\oplus W$ is computed by the sum over all rooted binary trees by decorating each leaf with the inclusion $i$, each internal line with the homotopy $h$, and the root by the projection $p$. A vertex means the application of the product $\mu_2$. In the case of the binary product one writes:
\begin{equation}
\begin{tikzpicture}
\begin{feynman}
\vertex at (-2,0) {$\mu'_2 \ = $};
\vertex(a) at (0,0);
\vertex(b) at (-1,1) {$i$};
\vertex(c) at (-1,-1) {$i$};
\vertex(d) at (1,0) {$p$ \: .};
\diagram* {(a)--(b), (a)--(c), (a)--(d)};
\end{feynman}
\end{tikzpicture}
\end{equation}
In formulas this means
\begin{equation}
\mu'_2\left((x_1,w_1),(x_2,w_2)\right) = \left( [x_1,x_2] \: , \: p(\rho(x_1)i(w_2) \pm \rho(x_2)i(w_1)) \right) \: .
\end{equation}
Accordingly we find for the $L_\infty$ module action $\rho'$
\begin{equation}
\rho'^{(1)} = p \circ \rho \circ i \: .
\end{equation}
In the case of $\mu'_3$ we can write
\begin{equation}
\begin{tikzpicture}
\begin{feynman}
\vertex at (-2,0) {$\mu'_3 \ = $};
\vertex(a) at (-1,1) {$i$};
\vertex(b) at (-1,0) {$i$};
\vertex(c) at (-1,-1) {$i$};
\vertex(d) at (0,0.5);
\vertex(e) at (1,0);
\vertex(f) at (2,0) {$p$};
\diagram* {(a)--(d), (b)--(d), (d)--[edge label = $h$](e), (c)--(e), (f)--(e)};
\end{feynman}
\end{tikzpicture}
\begin{tikzpicture}
\begin{feynman}
\vertex at (-2,0) {$\pm$};
\vertex(a) at (-1,1) {$i$};
\vertex(b) at (-1,0) {$i$};
\vertex(c) at (-1,-1) {$i$};
\vertex(d) at (0,-0.5);
\vertex(e) at (1,0);
\vertex(f) at (2,0) {$p$};
\diagram* {(a)--(e), (b)--(d), (d)--[edge label = $h$](e), (c)--(d), (f)--(e)};
\end{feynman}
\end{tikzpicture}
\begin{tikzpicture}
\begin{feynman}
\vertex at (-2,0) {$\pm$};
\vertex(a) at (-1,1) {$i$};
\vertex(b) at (-1,0) {$i$};
\vertex(c) at (-1,-1) {$i$};
\vertex(d) at (0, 0.5);
\vertex(e) at (1,0);
\vertex(e') at (-0.48,0);
\vertex(e'') at (-0.18,0);
\vertex(f) at (2,0) {$p$ \: .};
\diagram* {(a)--(d), (c)--(d), (b)--(e'),(e'')--(e), (d)--[edge label = $h$](e), (f)--(e)};
\end{feynman}
\end{tikzpicture}
\end{equation}
This gives for $\rho'^{(2)}$
\begin{equation} \label{rho2}
\rho'^{(2)}(x_1,x_2) = p \circ \left( \rho(x_1)h\rho(x_2) \pm \rho(x_2)h\rho(x_1) \right) \circ i \: .
\end{equation}
In this manner we can also obtain a general sum over trees representation for $\rho'^{(k)}$ in terms of $\rho$. Using equations~(\ref{rho and mu}) and~(\ref{L-mod sum}) we see that $\rho'^{(k)}$ can be obtained from binary rooted trees with $k+1$ leaves by the following rules. Label the first $k$ leaves by elements  $x_1,\dots,x_k$ and the last one by the inclusion $i$. Keep only those trees where there are no vertices connecting two elements of $\mathfrak{g}$. As usual, each internal line carries the homotopy $h$ and the root is decorated by $p$. A vertex now means ``apply $\rho(x_i)$''. For example we can write for $\rho'^{(2)}$:
\begin{equation}
\begin{tikzpicture}
\begin{feynman}
\vertex at (-4,0) {$\rho'^{(2)}(x_1,x_2)$ \ \ $=$};
\vertex(a) at (-1,1) {$x_1$};
\vertex(b) at (-1,0) {$x_2$};
\vertex(c) at (-1,-1) {$i$};
\vertex(d) at (0,-0.5);
\vertex(e) at (1,0);
\vertex(f) at (2,0) {$p$};
\diagram* {(a)--(e), (b)--(d), (d)--[edge label = $h$](e), (c)--(d), (f)--(e)};
\end{feynman}
\end{tikzpicture}
\begin{tikzpicture}
\begin{feynman}
\vertex at (-2,0) {$\pm$};
\vertex(a) at (-1,1) {$x_1$};
\vertex(b) at (-1,0) {$x_2$};
\vertex(c) at (-1,-1) {$i$};
\vertex(d) at (0, 0.5);
\vertex(e) at (1,0);
\vertex(e') at (-0.48,0);
\vertex(e'') at (-0.18,0);
\vertex(f) at (2,0) {$p$ \: .};
\diagram* {(a)--(d), (c)--(d), (b)--(e'),(e'')--(e), (d)--[edge label = $h$](e), (f)--(e)};
\end{feynman}
\end{tikzpicture}
\end{equation}
Clearly this recovers~(\ref{rho2}).

\printbibliography
\end{document}